\documentclass[onecolumn]{article}
\usepackage{amsmath,amssymb,amsthm}
\usepackage{bbm}
\usepackage{hyperref}
\usepackage{tikz}
\usepackage{algorithm}
\usepackage{algpseudocode}
\usepackage{color}
\usepackage{array}
\usepackage[titletoc,toc,title]{appendix}
\usetikzlibrary{arrows,shapes,positioning,backgrounds}

\usepackage{graphicx}
\usepackage{caption}
\usepackage{subcaption}

\usepackage{anysize}
\marginsize{1in}{1in}{1in}{1in}

\theoremstyle{definition}
\newtheorem{theorem}{Theorem}[section] 
\newtheorem{assumption}{Assumption}[section] 
\newtheorem{remark}{Remark}[section] 
\newtheorem{definition}{Definition}[section]
\newtheorem{corollary}{Corollary}[section]
\newtheorem{proposition}{Proposition}[section]
\newtheorem{example}{Example}[section] 
\newtheorem{lemma}{Lemma}[section]

\numberwithin{equation}{section}

\newcommand{\pardiv}[2]{\frac{\partial #1}{\partial #2 }}

\newcommand{\length}[1]{\left\lvert #1\right\rvert}
\newcommand{\abs}[1]{\bigl\lvert #1\bigr\rvert}
\newcommand{\inv}[1]{#1^{-1}}
\newcommand{\set}[1]{\{#1\}}

\newcommand{\one}{\mathbbm{1}}
\newcommand{\real}{\mathbb{R}}
\newcommand{\nn}{\mathbb{N}}
\newcommand{\sd}{\,|\,} 
\newcommand{\goesto}{\rightarrow} 
\newcommand{\norm}[1]{\left\lVert #1\right\rVert}
\newcommand{\wo}{\backslash} 
\newcommand{\struct}[1]{\mathcal{#1}} 

\DeclareMathOperator{\Exp}{E} 
\DeclareMathOperator{\Var}{Var} 

\title{
  (Failure of the) 
  Wisdom of the crowds
  in an endogenous opinion dynamics
  model with 
  multiply biased 
  agents
} 
\author{Steffen Eger\thanks{My thanks go to Alexey Cherepnev for
    writing the Matlab programs 
  that are underlying this work's simulations and for some useful
  discussions. }}
\date{\today}

\begin{document}
\maketitle

\begin{abstract}
We study an \emph{endogenous} opinion (or, belief) dynamics model
where we 
endogenize the social network that 
models 
the link (`trust')
weights between agents. Our network adjustment mechanism is simple: an
agent increases her weight for another agent if that agent has been
close to truth (whence, our adjustment criterion is `past
performance'). Moreover, we consider \emph{multiply biased} agents
that do not learn in a fully rational manner but are subject to
persuasion bias --- 
they learn in a DeGroot manner, via a simple `rule of thumb' 
--- and that have
biased initial 
beliefs. In addition, we also study this 
setup under
\emph{conformity}, \emph{opposition}, and \emph{homophily} --- which
are recently suggested variants of DeGroot learning in social networks
--- thereby
taking into account further biases agents are susceptible to.  Our main
focus is on \emph{crowd wisdom}, that is, on the question whether the
so biased agents can adequately aggregate dispersed information and,
consequently, learn the true states of the topics they communicate
about. In particular, we present several conditions under which wisdom
fails.  
\end{abstract}

\section{Introduction}\label{sec:introduction}
Crowds can be amazingly wise, even wiser than the most accurate
individuals among them. An early formalization of this insight has
been Concordet's Jury theorem from 1785 \cite{Concordet1785}, which
states that a 
simple majority vote of the opinions of independent and fallible 
lay-people may provide near-perfect accuracy if the number of voters
is sufficiently large.\footnote{Which also requires 
that 
each individual in the group
  of voters is more likely correct than not.} Over a hundred years
later, 
in 1906, 
Francis Galton \cite{Galton1907} found strong 
empirical support of Concordet's theoretical finding at an
agricultural fair in Plymouth. At a 
weight-judging contest, participants were asked to privately estimate
the weight of 
a chosen live ox after it had been slaughtered and dressed (meaning
that the head and other parts were removed). The winner
was the one whose estimate was closest to the true weight of the
ox. When analyzing the results in a \emph{Nature} article the
following year, Galton found that the simple average of the entire
crowd was even more accurate than the winner and that the median of
the $787$ valid guesses, $1197$ pounds, was extremely close to the
true weight, $1198$ pounds (cf.\ Bahrami, Olsen, Bang, Roepstorff,
Rees, and Frith (2012) \cite{Bahrami2012}, Acemoglu and Ozdaglar
(2011) \cite{Acemoglu2011}). This finding was obtained
even though most participants were no `experts' in this contest, with
little specialized knowledge in butchery; yet, their estimates could
obviously contribute to the crowds' overall success. Galton took this
result 
as evidence that democratic political systems may work.

Yet, contradicting this optimistic viewpoint concerning the wisdom of
crowds, it has also been observed that groups of individuals may be
quite fallible, and possibly even more fallible than most or all of
their members. One result of this kind is already hidden in Concordet's
Jury theorem: namely, if each lay-person is just slightly `too
uniformed' (or slightly `too much mistaken'), then the majority vote
may be much less accurate than each individual's
estimate. Drawing upon empirical observations, a 
comprehensive illustration of `crowd madness' has been brought forward
in Scottish journalist Charles Mackay's (1841) \cite{Mackay1841}
work \emph{The 
  extraordinary and popular delusions and the madness of crowds},
where the author chronicles `humankind's collective follies', including
financial bubbles,
in the economics context, and other popular
`delusions' such 
as witch-hunts and fortune-telling (cf.\ Bahrami, Olsen, Bang, Roepstorff,
Rees, and Frith (2012) \cite{Bahrami2012}), thus challenging the claim
that ``two heads are better than one''. 

Today --- while,
according to scholars' opinions, the question of wisdom of crowds
continues to be one of the most important issues facing social
sciences in the twenty-first century ---\footnote{See the recent survey
at \url{http://bit.ly/hR3hcS}.} more is known on group wisdom and collective failure. On the
one hand, the 
mean of the opinions of several individuals may become increasingly
accurate, for large groups, merely as a consequence of the law of
large numbers. This holds under 
restrictive assumptions --- in particular, that 
the beliefs of individuals are independent and
probabilistically centered around truth such that, on an aggregate
level, individual errors cancel out. On the other hand, much
empirical literature, foremostly in psychology, has documented that,
frequently, 
``groups outperform individuals [...], although groups typically fall
short of the performance of their highest-ability members'' (Kerr,
MacCoun, and Kramer (1996) \cite{Kerr1996}, p.691).\footnote{Groups
  can also blatantly fail, as, e.g., in \emph{groupthink} (Janis
  (1972) \cite{Janis1972}), \emph{hidden profiles} (Wittenbaum and
  Stasser (1996) \cite{Wittenbaum1996}), etc. See the
  overview in Kerr and Tindale (2004) \cite{Kerr2004}.}
In fact, a very recent experiment by Lorenz, Rauhut, Schweitzer, and
Helbig (2011) \cite{Lorenz2011} finds that `social influence', in a
broad meaning, in a group causes individuals' beliefs to become more
similar over time, without improvements in accuracy, however. 
Hence, much depends on how groups aggregate or process individual
opinions and also on these initial predispositions of agents. 
Kerr, MacCoun, and
Kramer's (1996) \cite{Kerr1996} insight is that whether groups perform
better than individuals may depend, among other things, on the
following aspects: (1) the way that groups aggregate the
opinions of individuals (that is, the {group decision, or
belief integrating, process), (2) the \emph{bias} of individuals, and
(3) the type of bias. Concerning issue (1), the way agents in groups
process their peers' beliefs, we assume a specific structural form
below, which has empirically proved plausible for learning in the
domain we consider (social networks). 

Issue (2), individuals' bias,
will be another central notion in our work. The classical work of
Tversky and Kahnemann (1974) \cite{Tversky1974} documents several biases human
judgment is susceptible to. In particular, \emph{anchoring biases}
describe the psychological condition of humans to pay undue attention
to initial values --- e.g., typically, individuals estimate the
product $9\times 8\times 7\times 6\times 5\times 4\times 3\times
2\times 1$ to be higher than the product of factors in reverse order,
which is attributed to subjects' performing an initial approximate computation
based on the first few terms, which entails a biasing anchor (the same
effects may happen if the anchor is exogenously specified, e.g., by
providing the subjects with random numbers as anchors and then
querying them for their 
own judgement). \emph{Biases of availabilty} refer to the phenomenon
of assessing (and, consequently, possibly, misjudging) the probability
of an event by the `ease with which 
instances or occurrences can be brought to mind', and, finally,
\emph{biases of representativeness} lead subjects to assess the
probability that an object is of a particular class (e.g., that a
person has a certain profession) by the degree to which the object is
representative of the class, which may lead to judgement errors
because such reasoning ignores, e.g., base-rate frequencies. 
In another typological classification of bias, Kerr, MacCoun
and Kramer (1996) \cite{Kerr1996} distinguish between
\emph{judgmental sins of imprecision} 
(systematically deviating from prescribed and precise use of
information, such as ignoring Bayes' theorem when forming beliefs or
being affected by framing; see Kahnemann and Tversky, 1984
\cite{Kahnemann1984}),  
\emph{judgmental sins of commission} (using
irrelevant information to arrive at a decision, such as the
attractiveness of an accused) and \emph{judgmental sins of
  omission} (ignoring relevant information, such as base-rate
information). 

We now describe the setup investigated in the current work, relating
to the issues discussed above subsequently. 
We 
consider a \emph{(social) network} of individuals, or agents, that
form opinions, or beliefs, about an underlying state or a discussion
topic.\footnote{Opinions or beliefs are important, from an economics
  perspective, because they crucially shape economic behavior:
  consumers' opinions about a product determine the demand for that
  product and majority opinions set the political course, etc. See
  Buechel, Hellmann, and Kl\"{o}{\ss}ner (2013) \cite{Buchel2013}.}
  We assume that agents start with some 
initial beliefs, at time 
zero, and then, as time progresses,  
learn from each other through \emph{communication}. Communication between any
two individuals takes place if there is a link between them in the
network. In the current work, we assume a specific form of learning
paradigm, \emph{DeGroot learning}, that posits that agents update
their beliefs by taking weighted arithmetic averages of their peers'
past beliefs, whereby the weights are given by the (social) ties between
the agents in the network. Much has been said on the adequacy (or
inadequacy) of 
DeGroot learning --- a `boundedly rational' 
learning
paradigm that posits that agents are susceptible to \emph{persuasion
  bias}, not properly adjusting for the repetition of information they
hear --- which, as experiments claim 
(e.g., \cite{Chandra2012,Corazzini2012}), appears as a more plausible
standard 
of human social learning than, e.g., fully rational Bayesian
learning and 
we refer the reader to, e.g., DeMarzo, Vayanos, and Zwiebel
(2003) \cite{Demarzo2003}, Golub and Jackson (2010)
\cite{Golub2010} or Acemoglu and Ozdaglar (2011) \cite{Acemoglu2011}
for extensive discussions. 
While the
DeGroot model of opinion formation is quite old, dating back to Morris
H.\ DeGroot (1974)'s \cite{DeGroot1974} seminal work, the framework has only
more recently received increasing attention from the economics
community. 

In this context, one matter that has been put forth as a
central guiding 
question in DeGroot learning models, and which connects to our initial
discussion, is whether the `na\"{i}ve' DeGroot
learners, who commit the `sin of imprecision' of not (properly) applying 
Bayes' theorem, can, in fact, become `wise' \cite{Golub2010}. Here,
a society (set of agents) is called \emph{wise}, roughly, if it reaches a
consensus --- in the limit, as time (discussion
periods) goes to
infinity --- that corresponds to truth. In Golub and Jackson (2010)
\cite{Golub2010}, the question relating to wisdom has been answered in
the affirmative 
--- (even) na\"{i}ve (DeGroot) learners do become wise under rather
mild conditions; namely, all that is required is that no na\"{i}ve
learner is \emph{excessively influential}, whereby an agent is
excessively influential if his social influence (how limiting beliefs
depend on this agent's initial beliefs) does not converge to
zero as society grows. In undirected networks (social ties are mutual)
with uniform weights, an obstacle to wisdom would then, e.g., be that
each agent newly entering society assigns, e.g., 
a constant fraction 
of his links to
a particular agent, who would then be excessively influential. Hence,
as long as links are somewhat `democratically' balanced, na\"{i}ve
DeGroot learners would apparently become wise. 
While we hold the analysis of Golub 
and Jackson (2010) \cite{Golub2010} to be an 
important `benchmark' for DeGroot learning, we think that it is overly
optimistic in 
at least one of its critical two assumptions, namely,
the (1) \emph{unbiasedness} of agents'
initial beliefs.\footnote{The other critical assumption is biasedness
  of the belief formation process (DeGroot learning --- agents are
  prone to persuasion bias). Finally, of crucial relevance is also
  independence 
  of agents' initial beliefs, which we do not challenge here, however.} 

In the current work, we drop, in particular, the largely implausible, as we
find, assumption (1). Our central notion will be as illustrated in
Figure \ref{fig:intro}, which we adapt from Einhorn, Hogarth, and
Klempner (1977) \cite{Einhorn1977}. 
\begin{figure*}[!ht]
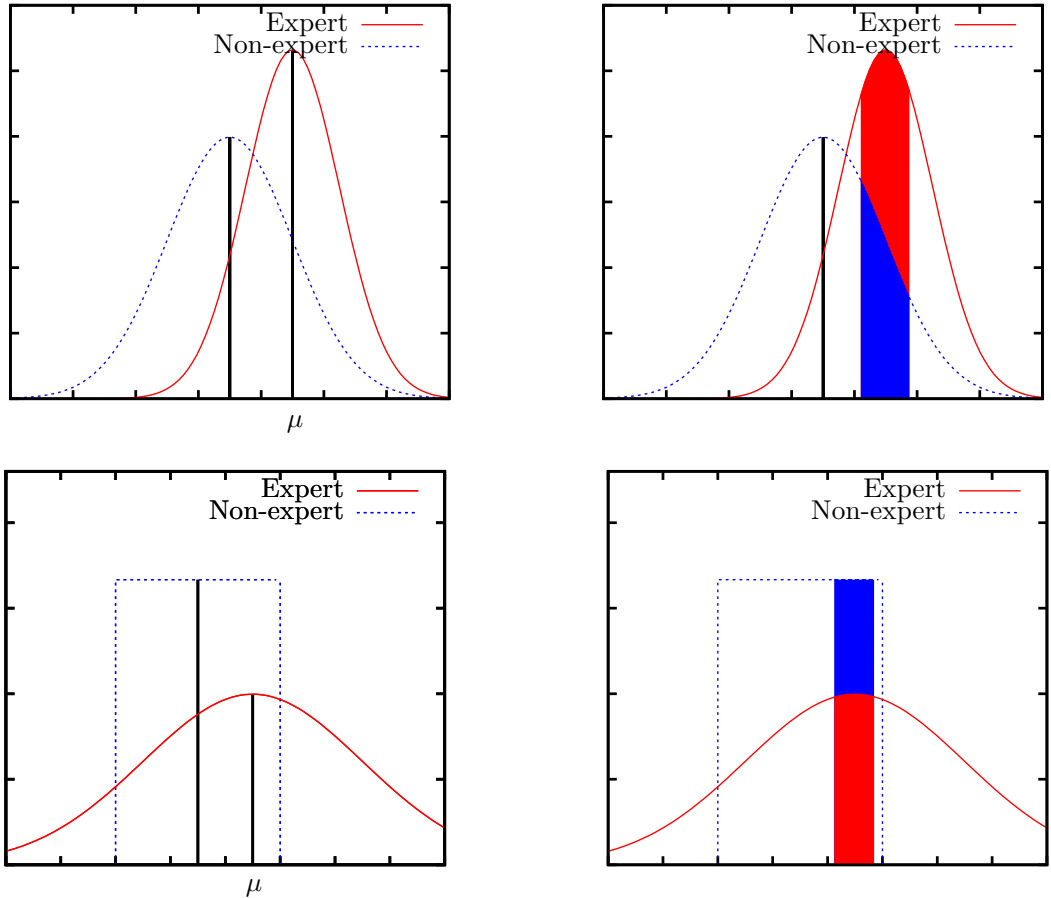

    \centering
    \begin{subfigure}{0.49\textwidth}
        \centering  
        \resizebox{1.0\textwidth}{!}{
        \input{plots/expert1.tex}}
    \end{subfigure}%
    \begin{subfigure}{0.49\textwidth}
        \centering 
        \resizebox{1.0\textwidth}{!}{
        \input{plots/expert2.tex}}
    \end{subfigure}\\%
    \begin{subfigure}{0.49\textwidth}
        \centering  
        \resizebox{1.0\textwidth}{!}{
        \input{plots/expert3-1.tex}}
    \end{subfigure}
    \begin{subfigure}{0.49\textwidth}
        \centering  
        \resizebox{1.0\textwidth}{!}{
        \input{plots/expert4-1.tex}}
    \end{subfigure}
    \caption
{ 
  Schematic illustration of experts' and non-experts' distribution of
    initial beliefs. Right figures show probability masses of falling
    within an (arbitrary) small interval around truth, for both
    experts and non-experts. 
}
    \label{fig:intro}
\end{figure*}
In words, we assume that some agents' initial beliefs are
\emph{biased}, with an expected value that is different from truth
$\mu$, and that other agents' initial beliefs are unbiased, with an
expected value that equals truth $\mu$ ---
we remark here that we abstract away from the precise type of bias
some agents' initial beliefs are subject to, that is, we are agnostic
about whether, e.g., agents commit sins of imprecision, commission, or
omission in forming their initial beliefs, simply assuming \emph{that}
at least some agents' initial beliefs are biased. 
We might, if we wish, label
the first kind 
of agents `non-experts' and the second `experts', although this might
be slightly misleading, as even experts can be biased, of
course; nonetheless, for convenience, we keep this terminology in the
following. 
A situation as sketched may be quite challenging to assess, for
individuals. Ignoring non-experts may be suboptimal, in some
circumstances, because their verdicts may still not be totally irrelevant 
in that their opinions
may have (relatively) large probability of being close to truth.
Consider, in
particular, the bottom part of Figure \ref{fig:intro} where the `expert' is
unbiased but has high variance. In this case, for each `closeness
interval' around truth, the non-expert's initial belief has higher
probability of falling within this interval than the expert's
beliefs. Thus, if there is exactly one expert and one non-expert, it
would be optimal, for an outside observer, to disregard the expert's
opinion and, in the absence 
of further information, adopt the non-expert's opinion. However, if
there are \emph{many} experts with identical and independent distributions
and also \emph{many} 
non-experts with identical and independent distributions, then an
optimal aggregation 
of information would ignore the non-experts' and average the experts'
opinions. 

We study this setup in an \emph{endogenous} DeGroot learning
model, where we endogenize the (social) network. In particular, we
assume that the `trust' links between agents in the network are based on `past
performance', which has been outlined as a relevant reputation
building criterion in the psychology 
literature (cf.\ Yaniv and Kleinberger (2000) \cite{Yaniv2000}, Yaniv
(2004) \cite{Yaniv2004}). 
We think that the endogenous model is the `right' setup for our
investigation of wisdom in DeGroot learning under biased initial
beliefs because if the network structure is assumed exogenous, then
one relatively uninteresting solution to the wisdom problem would,
e.g., be to ignore the biased agents --- in contrast, in the
endogenous model, 
the question arises what weighting scheme individuals \emph{actually}
learn for the biased (and unbiased) agents, under 
given 
assumptions concerning 
the agents' behavior. Since we learn the network structure
endogenously, by 
looking at agents' past performance (how often have they been close to
truth previously?), we necessarily study learning in a \emph{repeated
  setting}, where agents are involved in repeated communications over
\emph{multitudes} of topics, whereby past beliefs and their external
validation may inform today's beliefs and the network structure. Our
network learning rule is quite simple: we increment the weight that
one agent places upon another by some $\delta>0$ if the latter agent
has been 
in a predefined `$\eta$-radius' around truth for the current topic. We
show that this is a `utility maximizing' rule\footnote{Or at least a
  good approximation to a solution of a maximization problem.}
provided that agents 
\emph{expect}, subjectively, that all other agents' beliefs are
unbiased, which we call the \emph{bona fides} (or, `good faith')
assumption. Intruigingly, the 
bona fides assumption
concerning the unbiasedness of other agents' initial beliefs
may be consistent
with the `egocentric bias' hypothesis which suggests that
``interacting human agents operate under the assumption that their
collaborators' decisions and opinions share the same level of
reliability'' as their own; the upholding of this bias, even despite
potential 
collective failure, might then be due to the social obligation to treat
others as equal to oneself, despite their conspicuous inadequacy, or
due to the urge to contribute to the group 
(Bahrami, Olsen, Bang, Roepstorff, Rees, and Frith (2012)
\cite{Bahrami2012}).  

To summarize our model, in our setup, agents hold beliefs and learn,
via communication, 
about a multitude of 
topics $X_1,X_2,X_3,\ldots$, each with an associated `truth'
$\mu_1,\mu_2,\mu_3,\ldots$. 
Within each topic $X_k$, `discussion rounds' are indexed by discrete
time 
steps $t=0,1,2,\ldots$ and in each time step $t\ge 1$, agents \emph{update}
their beliefs on $X_k$ by integrating their peers' beliefs, starting
with some exogenously specified initial beliefs on $X_k$. After a 
topic has been communicated about (for an 
infinite amount of time), truth
$\mu_k$ is \emph{revealed}, whereupon agents adjust the `trust'
weights they 
assign to other agents (they `learn', or `grow', the network topology)
based on 
agents' past performance: if an agent has been close to truth for
topic $X_k$, agents
increase their trust for this agent by increasing the respective
weight by $\delta$.  
Our agents are multiply biased (or `\emph{na\"{i}ve}'):
\begin{itemize}
  \item[(i)] At least some agents' initial beliefs are \emph{systematically
    biased} in that the expected values of their initial beliefs are
    different from truth $\mu_k$, for all $k=1,2,3,\ldots$. For
    initial beliefs, we abstract away from the particular kind of bias
    agents are subject to, simply assuming \emph{that} some kind of
    bias plays a role.
  \item[(ii)] Agents are subject to \emph{persuasion bias} in updating their
    beliefs on $X_k$ in that they apply the DeGroot learning paradigm
    rather than a fully rational Bayesian belief updating framework. 
  \item[(iii)] In adjusting weights for other agents, agents are
    \emph{egocentrically biased}: they assume that their own judgments
    are relevant (more precisely, their initial beliefs are unbiased)
    and they assume that their peers' beliefs share the same level of
    reliability as their own (more precisely, that their peers'
    initial beliefs are also unbiased). This bias justifies the weight
    adjustment rule --- adding $\delta$ --- that we have sketched (see
    Section \ref{sec:modelECCS}). 
\end{itemize}
Besides this basic setup, we consider refinements of standard DeGroot
learning recently suggested --- DeGroot learning under opposition,
conformity, and homophily --- in each case incorporating our endogenized
network 
structure and, in addition, the three kinds of biases discussed
above. We show that, in these more refined versions of DeGroot
learning, which are supposed to  
endow the 
DeGroot learning paradigm with a more `realistic' 
structure, 
wisdom is even more difficult to arrive at, as we discuss below. 

Our main contributions over existing work are as follows. 
\begin{itemize}
  \item We more thoroughly investigate the concept of \emph{bias} in
    social (network) learning --- or more specificially, DeGroot
    learning --- than previous literature. In particular, as
    mentioned, we allow agents' initial beliefs to be biased and
    consider further biases, as discussed. 
  \item We endogenize the network structure in DeGroot learning and we
    do so by referring to the notion of `past performance'. Of course,
    in the vast literature on networks, (`endogenous') network formation
    processes are not novel; often, however, the network is adapted,
    in the literature more or less relevant to our setup, by adding or
    deleting (costly) links as in Jackson and Watts
    (2002) \cite{Jackson2002}, Goyal (2004) \cite{Goyal2004}, etc., rather
    than by increasing 
    link weight based on agents' past performance. In DeGroot
    learning, self-evolving networks are discussed, e.g., in the work
    on DeGroot learning and homophily (e.g., Pan (2010) \cite{Pan2010} and
    the Hegselmann and Krause models), but weight adjustments
    based on truth, as we model, must be considered distinct from
    these mechanisms.  
  \item As mentioned, we consider \emph{multitudes} of topics, rather
    than a single topic, in DeGroot learning, and we crucially allow
    truth to be \emph{revealed} at some stage. This differs from all
    the previous work, where agents have been in the unfortunate
    situation of eternally communicating about a given topic, without
    ever knowing its true state. 
  \item We incoporate other DeGroot variants in our setup. In
    particular, we provide an alternative to the homophily model
    designed by Hegselmann and Krause \cite{Hegselmann2002}, see
    Section \ref{sec:homophily}.
  \item We derive a microeconomic foundation for weight adjustments as
    we implement by defining an individual agent's optimization
    problem --- in particular, we assume that agents have negative
    utility from not knowing truth --- and by computing a closed-form
    solution to this problem.\footnote{In spirit, our approach is
    similar to that of DeMarzo, Vayanos, and Zwiebel (2003)
    \cite{Demarzo2003}.} We then show how our heuristic weight
    adjustment rule --- 
    adding $\delta$ --- corresponds to the solution of the
    optimization problem. 
\end{itemize}
Our main findings are as follows. 
\begin{itemize}
  \item For the standard model, we first show that agents reach a
    \emph{consensus} for almost all topics $X_k$, under weak
    conditions, in our endogenized DeGroot learning paradigm
    (Proposition \ref{prop:consensus} and Remark \ref{rem:distR}). This
    confirms the commonly held belief (cf.\ Acemoglu and Ozdaglar
    (2011) \cite{Acemoglu2011}) that the standard 
    DeGroot model leads agents to consensus (so easily) but also shows that,
    in our endogenized model, conditons that prevent consensus are, in
    fact, not satisfied.
  \item Next, we illustrate that if all agents' initial beliefs are
    \emph{unbiased}, then agents in fact reach a consensus that is even
    correct, for `large' topics $X_k$ and as agent group size $n$
    becomes large. This holds both when agents adjust weights based on
    limiting beliefs and on initial beliefs (Propositions 
    \ref{prop:centered1} and \ref{prop:centered2}, respectively);
    we define the 
    notions of relevant weight adjustment time points below. When there
    are \emph{biased} agents, then agents' limiting beliefs are generally a
    convex combination of the unbiased agents' initial beliefs and the
    biased agents' initial beliefs. We demonstrate the truthfulness of
    this claim under various parametrizations (Propositions \ref{prop:inf},
    \ref{prop:general1}, \ref{prop:general2}). We also give sufficient
    conditions on when agents may converge to truth, for large topics, even
    under the presence of biased agents (Propositions \ref{prop:inf}
    and \ref{prop:0}), 
    but these conditions are `low 
    probability events' (or require a sufficiently high valuation of
    truth) and they hold only under the particular 
    parametrization that agents stop learning the network topology in
    case `everything is 
    fine', as we define below. 

    That limiting consensus beliefs are convex combinations of biased
    and unbiased beliefs may imply that limiting beliefs are
    `arbitrarily' far off from truth, provided that the number of biased
    agents is sufficiently large (Corollary \ref{cor:noepswisdom}),
    thus demonstrating 
    that agents do not optimally aggregate information in our
    endogenized DeGroot learning model, at least under certain conditions. 
  \item Next, for opinion dynamics `under opposition', a recently
    suggested DeGroot learning variant where agents are 
    motivated by 
    `ingroup'/`outgroup' relationships \cite{Eger2013}, we show that
    even if all 
    agents' initial beliefs are unbiased and, more particularly,
    agents receive arbitrarily 
    accurate initial signals about topics, some agents may be
    arbitarily far off from truth. In other words, we show that if
    agents have additional incentives besides truth, namely, to
    disassociate from unliked others --- such agents must be thought
    of as additionally biased; namely, they must be thought of 
    as, e.g., committing the sin of omission to ignore the unliked
    others' relevant information and the sin of commission to
    incorporate irrelevant information, namely, the `opposite' of 
    unliked others' beliefs ---\footnote{They may generally be thought
    of as biased toward ingroup members, cf.\ Brewer
    (1979) \cite{Brewer1979}.} then wisdom is even more difficult 
    to attain. This, in particular, concerns several important fields
    of everyday life, such as the political arena. 
  \item Then, for DeGroot learning `under conformity' --- that is,
    when agents want to conform to a reference opinion (again, which
    may be thought of as a sin of commission) --- another recent
    variant of DeGroot learning \cite{Buchel2012}, we show that even
    if the unbiased agents have never been truthful in the past, they
    may become arbitrarily influential, something that is not possible
    in the standard model, and which, again, shows that additional biases
    may worsen the case for wisdom. 
  \item Finally, in case homophily also plays a role --- that is, when
    agents have the tendency to adjust the social network topology
    based on agents with similar beliefs --- then, again, wisdom is
    more difficult to arrive at. We show
    this (only) by simulation since this 
    process is (much) more difficult to analyze
    analytically as it deals with learning matrices
    that are changing over time (and not only across topics). In our 
    context, homophily can also be seen as a 
    \emph{search bias} 
    in which subjects overrate beliefs that are close to their own
    (cf.\ Kunda (1990) \cite{Kunda1990}). 
\end{itemize}
The structure of this work is as follows. In Section
\ref{sec:related}, we present related work, beyond what we have already
referred to. In Section \ref{sec:modelECCS}, we give a formal outline of
our model, and, in Section \ref{sec:justification}, a `justification'
of our network 
learning rule. In Section \ref{sec:notation}, we 
introduce relevant notation. Then, in Sections \ref{sec:standard},
\ref{sec:opposition}, \ref{sec:conformity}, and \ref{sec:homophily}, we
derive our results, as outlined above, on the standard model, and the
DeGroot variants under opposition, conformity, and homophily, 
respectively. In Section \ref{sec:conclusion}, we conclude. We list
several proofs in the appendices; there, we also report on a
`small-scale' experiment on the (un)biasedness (and the distribution)
of individuals' (initial) 
beliefs concerning several `common knowledge questions'. 

Before actually listing related work in Section \ref{sec:related},
we now briefly discuss this experiment and the lessons that we learn
from it. 


\bigskip
\noindent \textbf{A small-scale experiment concerning the
  (un)biasedness of individuals' beliefs}. 
As we have mentioned, some research papers have assumed that
individuals' initial beliefs on topics are unbiased, that is, centered
around truth. Certainly, this assumption may sometimes be plausible,
e.g., depending on the topic, but, as we have indicated, we do not think
that the condition holds across a large spectrum of circumstances. We
conducted an experiment where we asked individuals on \texttt{Amazon
Mechanical Turk}\footnote{Available
at \url{https://www.mturk.com/mturk/}.} $16$ `common knowledge
questions'. The questions 
ranged from, to our opinion, rather easy problems such as `What do you think is
the year the first world war started?' or `What do you think is
$17-4\times 2$?' to rather difficult problems, such as `What do you think is
the number of people per square mile in China's capital Beijing?' or
`What do you think is the diameter of the sun in miles?'. We list all
$16$ questions in Appendix \ref{sec:appendixb}. 

On all questions, more than $n=100$ subjects answered (between $n=110$
to $n=119$). Analyzing the answers (see Figures \ref{fig:a-2}
and \ref{fig:a-1}), we find that, typically, neither 
the mean of the answers nor the median are very close to the true
value. In fact, on only $8$ out of $16$ questions is the median (which
tends to be more reliable since it 
is not so much affected by outliers) within a $10\%$ interval around
truth, and on only $6$ out of $16$ questions does this hold for the
mean. Looking at $1\%$ intervals, these numbers drop to $6$ and
$2$, respectively (for the mean, these questions are about the start
of the first world war and the average height of an adult male US
American). Such low numbers were truly surprising if in fact the
assumptions of unbiasedness and independence of (initial) beliefs were
true, given the validity of the law of large numbers. A slightly more
detailed analysis is given in Appendix \ref{sec:appendixb}.

\section{Related Work}\label{sec:related}
Early and frequently cited predecessors of DeGrootian opinion dynamics 
are French (1956) \cite{French1956} and Harary (1959)
\cite{Harary1959}, although the now famous `averaging' model of opinion and
consensus formation 
has only been popularized through the seminal 
work
of DeGroot (1974) \cite{DeGroot1974}. 
At about the same time, Lehrer and Wagner
\cite{Wagner1978,Lehrer1981,Lehrer1983} have developed a model of rational
consensus formation in society that, in both its implications and its
mathematical structure, is very similar to the DeGroot model. 
In the sociology literature, Friedkin and Johnson
(1990) \cite{Friedkin1990} and Friedkin and Johnson
(1999) \cite{Friedkin1999} develop models of social influence that
generalize the DeGroot model. 
In more recent years, a renewed economic 
interest in the DeGroot model of opinion dynamics has
emerged, 
leading to a number of 
further extensions proposed. 
For example, DeMarzo, Vayanos, and Zwiebel
(2002) \cite{Demarzo2003}, 
besides sketching psychological justifications for DeGroot learning
relating to persuasion bias as discussed above, 
discuss time-varying weights on own beliefs that 
capture, e.g., the idea of a `hardening of positions': over time,
individuals may be more inclined to rely on their own beliefs rather
than on those of their peers. Further extensions of the
classical DeGroot model include Golub and Jackson (2010)
\cite{Golub2010}, whose 
contribution is to analyze weight structures such that DeGroot
learners whose 
initial beliefs are \emph{stochastically centered around truth} converge to a
consensus that is correct, and the works of Daron Acemoglu and
colleagues. For example, Acemoglu, Ozdaglar and ParandehGheibi (2010)
\cite{Acemoglu2010a} distinguish between regular and forceful
agents, such as, in an economic interpretation, monopolistic media
(forceful agents influence others 
disproportionately), and Acemoglu, 
Como, Fagnani, and Ozdaglar 
\cite{AcemogluForthcoming} distinguish between regular and stubborn
agents (the latter never update), to account for the phenomenon of
disagreement in societies; in Yildiz, Acemoglu, Ozdaglar,
Saberi, and Scaglione \cite{AcemogluSubmitted}, a discrete 
version of the DeGroot model with stubborn agents is analyzed in
which regular agents randomly adopt one of their neighbors' binary
opinions. 
Another interesting DeGroot variant is discussed in Buechel, Hellmann
and Kl\"{o}{\ss}ner (2012) \cite{Buchel2012} where agents' \emph{stated
opinions} may differ from their \emph{true} (or private) opinions and
where it is assumed that agents generally wish to state an opinion
that is close to that of their peer group even if their true opinions
may be very different (which is the `conformity' aspect of their
model); we review this work in more depth in
Section \ref{sec:conformity}. A similar approach is given in
Buechel, Hellmann and 
Pichler (2012) \cite{Buchel2012b}, where DeGroot learning is applied
to an overlapping generations model in which parents transmit traits to
their children. 
Receivers who deviate from the opinion signals sent by senders
--- \emph{rebels} --- are discussed in Cao, Yang, Qu, and Yang
(2011) \cite{Cao2011}; 
see also the modeling in Zhang,
Cao, Qin, and Yang (2013) \cite{Zhang2013} where such behavior is
interpreted in a `fashion' context. A model with more general
`ingroup/outgroup' relationships and opposition toward outgroup
members is described in Eger
(2013) \cite{Eger2013}, which we discuss in more detail in
Section \ref{sec:opposition}. 
Multi-dimensional
real opinion spaces have been considered in Lorenz (2006) \cite{Lorenz2006} and
a survey of generalizations of 
DeGroot models
developed within physicist communities (e.g., density-based approaches
in place of agent-based systems) is provided by Lorenz (2007)
\cite{Lorenz2007}. 
Groeber, Lorenz, and Schweitzer (2013) \cite{Groeber2013} provide
`dissonance minimization' 
as a general microfoundation of a variety of heterogenous DeGroot-like
opinion dynamics models. 

Concerning DeGroot models with \emph{endogenous} weight formation, one
pattern of endogenous weight formation that has been studied in the
literature is weight formation based on a \emph{homophily principle}, 
in which 
agents 
assign positive weights to those individuals whose current opinions are `similar'
with their own.
In Hegselmann and Krause (2002) \cite{Hegselmann2002} --- an approach
with many extensions such
as \cite{Hegselmann2005,Hegselmann2006,Douven2009,Douven2009b,Douven2010} --- 
this
leads to very interesting patterns of opinion formation in which, most
prominently, the paradigms of plurality, polarization and consensus
are observed, depending on specific parametrizations --- most
importantly, the definition of similarity, i.e., whether individuals are
tolerant or not toward other opinions, affects which opinion pattern
emerges. The model of Deffuant, Neau, Amblard, and
Weisbuch (2000) \cite{Deffuant2000} is identical in setup to the Hegselmann
and Krause model, except that two randomly determined agents, rather
than all agents, update beliefs in each time step. 
Pan (2010) \cite{Pan2010} discusses a homophily variant in which
agents assign trust weights to other agents \emph{in proportion} to agents'
current opinion distance --- rather than by thresholding, as done in the
Hegselmann and Krause models and in Deffuant, Neau, Amblard, and
Weisbuch (2000) \cite{Deffuant2000} --- which
typically entails a consensus, in the 
limit. 
Homophily and DeGroot learning is also investigated in Golub and
Jackson (2012) \cite{Golub2012}, where the relationship between
the speed of DeGrootian learning and homophily is discussed; in this
model, homophily is --- exogenously, however --- modeled by 
random networks where the link probability between different groups is
non-uniform, and is, in fact, higher between individuals of the same
group. 
Endogenous weight formation typically
implies time-varying weight matrices as belief updating operators and
mathematical results
on corresponding processes are, for instance, given in Lorenz
(2005) \cite{Lorenz2005}. 

Recent empirical and experimental evidence on the validity of the
DeGroot 
heuristic for learning in social networks has been provided in, e.g.,
Chandrasekhar, Larreguy, and Xandri
(2012) \cite{Chandra2012} and Corazzini, Pavesi, Petrovich, and
Stanca (2012) \cite{Corazzini2012}. Interesting in our context is also
the  
experiment by Lorenz, Rauhut, Schweitzer and Helbing
(2011) \cite{Lorenz2011}, where individuals are placed in a situation
consistent with our setup: individuals observe their peers' past
beliefs (on social/geopolitical issues) and may update their current
opinions accordingly. In addition, truth on each of the discussed topics
becomes revealed, by the experimenter, after a certain fixed amount of
time.

Social learning is also discussed in various other strands of
literature besides those discussed, such as in herding models
(cf. Banerjee (1992) \cite{Banerjee1992}, Gale and Kaiv
(2003) \cite{Gale2003}, Banerjee and 
Fudenberg (2004) \cite{Banerjee2004}), where agents usually converge
to holding the same 
belief as to an optimal action. This conclusion generally applies to the
observational learning setting (cf. Rosenberg, Solan and Vieille
(2006) \cite{Rosenberg2006}, Acemoglu, Dahleh, Lobel, and Ozdaglar
(2008) \cite{Acemoglu2008}), where agents are observing 
choices and/or payoffs of other agents over time and are updating
accordingly. See also the references and the discussion in Golub and
Jackson (2010) \cite{Golub2010}. General overviews over social learning,
whether Bayesian or non-Bayesian, whether based on communication or
observation, are, in the economics context, for
example given in Lobel (2000) \cite{Lobel2000} and Acemoglu and
Ozdaglar (2011) \cite{Acemoglu2011}. In Acemoglu
and Ozdaglar (2011) \cite{Acemoglu2011}, 
an extensive discussion of the `pros and cons'
of fully
rational learning models versus boundedly rational (most importantly,
DeGroot-like) heuristics is provided. 


As discussed in the introduction, group opinion and belief formation
and decision making 
also has a long history 
in psychology. A crucial difference between such models
and models of social learning is that, in the psychology studies and
models, it 
is usually assumed, and even explicity demanded, for the group of individuals to
reach a consensus in the course of the discussion process. 
A general
overview over group decision making is given in Kerr and Tindale
(2004) \cite{Kerr2004} and other 
relevant literature, besides that sketched in the introduction, is,
for example, Mannes (2009) \cite{Mannes2009} and Budescu, Rantilla, Yu,
and Karelitz (2003) \cite{Budescu2003}.

\section{Model}\label{sec:modelECCS}
A finite set $[n]=\set{1,2,\ldots,n}$ of $n$ agents discusses
a sequence $X_1,X_2,X_3,\ldots$ of topics. Each agent $i=1,\ldots,n$
holds \emph{initial beliefs} $b_{i}^k(0)\in S$ on issue $X_k$, where
$k=1,2,3,\ldots$ and where $S$ is a convex set that we may innocuously 
assume to be the whole of $\real$. Moreover, 
each topic has a corresponding \emph{truth} $\mu_k\in S$ which denotes the
`true evaluation' of topic $X_k$. Agents
\emph{update} their beliefs on 
$X_k$ by taking a \emph{weighted average} of all other agents'
beliefs, starting from initial beliefs:
\begin{align}\label{eq:degrootupdate}
  b_{i}^k(t+1) = \sum_{j=1}^n W_{ij}^{(k)}b_{j}^k(t),
\end{align}
where $t=0,1,2,3,\ldots$ and where $W_{ij}^{(k)}$ denotes the
\emph{weight} (`trust') that
agent $i$ assigns agent $j$ for topic $X_k$; in
Section \ref{sec:homophily}, we let $W_{ij}^{(k)}$ also depend on time
$t$, i.e., $W_{ij}^{(k)}=W_{ij}^{(k)}(t)$.  
We let the limiting beliefs of agent $i$ for issue $X_k$ be denoted by
$b_i^k(\infty)$. Moreover,
we
assume that 
weight matrix $\mathbf{W}^{(k)}$ --- which we
also interpret as a `learning matrix', or, as a \emph{(social)
network} --- is \emph{row-stochastic} for every
topic $k$, that is,  
\begin{align*}
  \forall i,j:\: 0\le W_{ij}^{(k)}\le 1,\quad\text{and}\quad \forall i:\:
  \sum_{j=1}^n W_{ij}^{(k)}=1, 
\end{align*}
which means that the weights that agents assign each other are
normalized to unity; we furthermore assume that weights carry over from one
topic to another, as we explicate below. 
Crucially, we consider an \emph{endogenous}
weight formation process where agents \emph{adjust} the weights they
attribute to other agents based on the foundational principle of truth.
\begin{itemize}
  \item If agent $j$ has known truth $\mu_k$ for issue
    $X_k$ (or, was `close enough'), then it seems natural for agent
    $i$ to increase his trust in $j$. Formally, we let
    \begin{align}\label{eq:adjusttruth}
      W_{ij}^{(k+1)} = 
      \begin{cases}
        {W}_{ij}^{(k)}+\delta\cdot
    T(\abs{N({\mathbf{b}^k(\infty)},\mu_k)}) & \text{if } 
        \norm{b_{j}^k(\tau)-\mu_k}< \eta,\\
        {W}_{ij}^{(k)}& \text{otherwise},
      \end{cases}
    \end{align}
    for all $k\ge 1$; by $\norm{\cdot}$, we denote the absolute
    distance and by $\length{A}$ the cardinality of set $A$. Here,  
    $N({\mathbf{b}^k(\tau)},\mu_k)\subseteq\set{1,\ldots,n}$ is the set of all 
    agents $i$ whose 
    belief $b_i^k(\tau)$ for $X_k$ at time $\tau$ is
    within an $\eta$-radius of $\mu_k$ and
    $T:\set{1,\ldots,n}\goesto [0,\infty]$ is a function for which we
    specify the following: 
    $m_1\le m_2 \:\implies\: T(m_1)\ge T(m_2)$ ($T$ is non-increasing in its
    argument; `knowing truth pays a weakly larger trust increment the
    less people know it'; see our discussion below). The 
    variable $\tau$ models the 
    relevant adjustment time point; we consider $\tau=0$ (`adjusting
    based on initial beliefs') and $\tau=\infty$ (`adjusting based on
    limiting beliefs'). We take the variables $\eta$, with $\eta\ge
    0$, and $\delta$, with $\delta>0$, 
    as exogenous variables. We also refer to the variable $\eta$ as
    the agents' \emph{tolerance} since it describes the interval
    within which agents are tolerant against deviations from truth. 
    
    Updating in case $b_j^k(\tau)$ is close to truth
    rather than exactly truth may be interpreted as a boundedly
    rational heuristic for agents who cannot assess truth with
    infinite precision. Note that after adjusting weights, we
    renormalize weight matrices in order for them to satisfy the
    row-stochasticity condition.  
\end{itemize}

\subsection*{Discussion}

Our endogenous DeGroot model appears quite simple and natural --- we
let agents adjust 
the network $\mathbf{W}$ in a way that incorporates `past
performance': whenever an agent has been close enough to truth, agents
increase their trust for this agent by $\delta$ --- except, possibly, 
for the weight adjustment time points and the factor
$T(\cdot)$. Concerning weight adjustment time points, the question
is \emph{what is the relevant time point that an agent's belief should be (or
is) 
compared to truth $\mu_k$} 
for some issue $X_k$. Note that, for any
issue $X_k$, 
there are infinitely many possible such time points ---
$t=0,1,2,3,\ldots$ --- so this question admits no straightforward
answer. We consider 
two relevant time points, namely, the beliefs that agents hold
\emph{initially}, at the beginning of communication, and the beliefs that
agents hold in the \emph{limit}, as time goes to infinity; these beliefs are
agents' limiting beliefs, after communication on topic $X_k$ has
terminated. Both time points have some intuitive appeal, as we
think. Initial 
beliefs say something on an agent's `innate ability', \emph{before}
learning from others, and limiting
beliefs may possibly be a more realistic reference point for weight
adjustments if agents are
perceived of as having `limited memory' (limiting
beliefs are the `most recent' observations).\footnote{Of course, it
could be argued that an agent's `average belief', somehow weighted
over time, might also be a quantity that could be compared with truth
$\mu_k$.} Concerning the function
$T(\cdot)$, our intuition is as follows. The 
larger the group of 
agents who know the 
correct answer (or, as we consider as equivalent throughout, are `close
enough' to truth) for a given topic --- that is, the larger is the set
of 
agents whose limiting beliefs are correct --- the smaller should be
the trust weight increment that agents assign each other. Intuitively,
the number of agents who are correct for a topic 
may be indicative of the 
topic's `difficulty' or `hardness'.\footnote{To make a crude example,
`everyone' may know what $3+3$ is --- so that correct knowledge of
this answer may 
not justify increased trust --- but it took an 
Euclid to first discover the infinitude of the set of prime numbers.} 
If $T(x)=0$, for some $x$, then this
means that the network is not adjusted if at least $x$ agents know the
truth on any one topic.\footnote{The condition $T(x)=0$
may also be paraphrased as meaning that `if at least $x$ agents know truth
on any one topic, then the network need not be changed' --- that is,
`everything is fine' if at least $x$ (e.g., $x=n$) agents know truth.}
Consider Figure \ref{fig:T} for three typical exemplars of
$T(\cdot)$, as we have in mind.
\begin{figure}[!htb]
\centering
\input{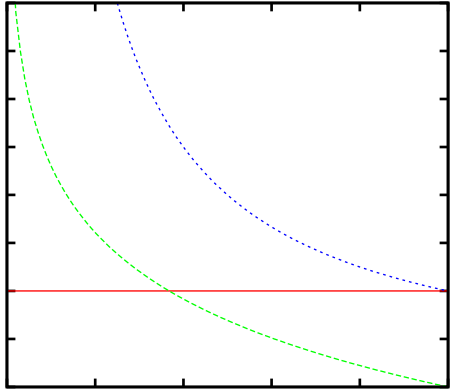}
\caption{Three possible specifications of the function $T$
in \eqref{eq:adjusttruth}. Note that $T$ is always non-increasing, as
we have defined. For example, the red function is $T(\cdot)=1$ and the
dashed green 
function has $T(n)=0$ and $T(x)>0$ for $x<n$.}
\label{fig:T} 
\end{figure}

We also point out that, in our model, we logically differentiate between
what is `innate knowledge' (or, simply, `\emph{ability}') and what
is \emph{socially learned} from others in that we think of initial beliefs as
capturing ability and updating beliefs based on others' beliefs as the social
learning 
process. Finally, we remark that we generally think of topics $X_k$ as
`of the same kind' --- that is, all of them are on sports or
mathematics or natural science or politics or the stock market,
etc. --- in order to 
justify why network weights should carry over from one topic to another; see
also our dicussion below. 

Almost all throughout the work, we assume that agents
are \emph{homogenous} with respect to the tolerances $\eta$, the
weight increments $\delta$, and adjustment time points $\tau$.

\section{A justification of our weight adjustment
  procedure}\label{sec:justification}
\subsection*{The choice of a rational agent}
We now derive a (micro-founded) justification of our weight adjustment
rule \eqref{eq:adjusttruth}. We first assume that agents
$i=1,\ldots,n$ have disutilities from not knowing truth for topic
$X_k$, for $k=1,2,\ldots$. 
More precisely, we 
assume that agent $i$ has utility function $U_i$
from weight structure $\mathbf{W}^{(k)}$ for issue $X_k$ as 
\begin{align}\label{eq:prime}
  U_i(\mathbf{W}^{(k)}) =
  U_i(\mathbf{W}^{(k)};\mu_k,b_1^k(0),\ldots,b_n^k(0)) =
  -F\bigl(\: d({b_i^k(\infty),\mu_k})\:\bigr), 
\end{align}
where 
$F:\real_{\ge 0}\goesto\real_{\ge 0}$ is monotonically
increasing and $d$ is 
a \emph{metric} --- that is, in particular, $d(a,b)=0$ if and only
if $a=b$ ---  
and
where we assume that $\mu_k$ and initial beliefs $b_1^k(0),\ldots,b_n^k(0)$
are exogenous; also note how $b_i^k(\infty)$ depends on
$\mathbf{W}^{(k)}$ (and $b_1^k(0),\ldots,b_n^k(0)$) via
process \eqref{eq:degrootupdate}. In  
other words, according to utility function \eqref{eq:prime}, a larger distance between $i$'s limiting belief 
$b_i^k(\infty)$ and truth $\mu_k$ does not lead to larger utility of
agent $i$ and when $b_i^k(\infty)=\mu_k$, then agent $i$ attains
maximum possible 
utility. For technical ease, we assume 
that $d$ is the Euclidean distance and $F$ has
the simple quadratic form $F(z)=z^2$ such that 
\begin{align*}
  U_i(\mathbf{W}^{(k)}) = -\norm{b_i^k(\infty)-\mu_k}^2 =
  -\bigl(b_i^k(\infty)-\mu_k\bigr)^2. 
\end{align*}
Now, we assume that $[\mathbf{W}^{(k)}]_i$, by which we denote the
$i$-th row of $\mathbf{W}^{(k)}$, are the endogenous variables
of agent $i$ she wants to set in such a way as to maximize her 
utility $U_i$.
Since agent $i$ cannot affect the weight structure choices of agents
$i'$, with $i\neq i'$, we write $U_i$ 
as a function of 
$[\mathbf{W}^{(k)}]_i$, rather than 
$\mathbf{W}^{(k)}$. Hence, we
write 
\begin{align*}
        U_i\bigl([\mathbf{W}^{(k)}]_i\bigr) =
        -\bigl(b_i^k(\infty)-\mu_k\bigr)^2. 
\end{align*}
 Assume next that agents $i=1,\ldots,n$ have `limited
foresight' or  
`finite horizon' in that they cannot foresee the dynamics 
of belief updating process \eqref{eq:degrootupdate} (which
would also require knowledge of the other agents' weight
choices) but that they take $b_i^k(1)$ as a reference variable, rather
than $b_i^k(\infty)$. 
\begin{assumption}\label{ass:1}
  Agents $i=1,\ldots,n$ have \emph{limited foresight} or \emph{finite
  horizon}. They choose weights $[\mathbf{W}^k(0)]_i$ to maximize 
  \begin{align*}
  U_i\bigl([\mathbf{W}^{(k)}]_i\bigr) =
        -\bigl(b_i^k(1)-\mu_k\bigr)^2 = -\Bigl(\sum_{j=1}^n
  W_{ij}^{(k)}b_j^k(0)-\mu_k\Bigr)^2.  
  \end{align*}
\end{assumption}
Our next assumption is that initial beliefs $b_1^k(0),\ldots,b_n^k(0)$
are \emph{random variables}. 
\begin{assumption}\label{ass:2}
  Initial beliefs $b_1^k(0),\ldots,b_n^k(0)$ are \emph{random
  variables}. 
\end{assumption}
From Assumption \ref{ass:2}, it follows that agents become 
\emph{expected utility
  maximizers}: they choose weights $[\mathbf{W}^{(k)}]_i$ to maximize 
  \begin{align*}
  \Exp_i\left[U_i\bigl([\mathbf{W}^{(k)}]_i\bigr)\right].
  \end{align*}
Our final assumption says that 
agents
expect their own and other agents' initial beliefs to be correct, which we
call the \emph{bona fides} (``good faith'') assumption.
\begin{assumption}\label{ass:3}
  Agents $i=1,\ldots,n$
  are \emph{bona fide}, that is,
  \begin{align*}
  \Exp_i[b_j^k(0)] = \mu_k, \quad\text{for all } j=1,\ldots,n, \text{
    and all } k=1,2,3,\ldots
  \end{align*}
\end{assumption}
Now, we derive agents' maximization problem under
Assumptions \ref{ass:1} to \ref{ass:3}. To this end, let $X$ denote the
random variable
\begin{align}\label{eq:X}
        X=\sum_{j=1}^n
  W_{ij}^{(k)}b_j^k(0).
\end{align}
With this notation, agents' utility maximization problems become,
under our named assumptions, for each agent $i=1,\ldots,n$:
\begin{equation}\label{eq:optim}
  \begin{split}
  \max_{[\mathbf{W}^{(k)}]_i}&
  \:\: \Exp_i\left[U_i\bigl([\mathbf{W}^{(k)}]_i\bigr)\right] =  \Exp_i\left[-\Bigl(\sum_{j=1}^n
  W_{ij}^{(k)}b_j^k(0)-\mu_k\Bigr)^2\right] =
  -\Exp_i\bigl[(X-\Exp_i[X])^2\bigr] \\
  \text{s.t. } & W_{i1}^{(k)}+\ldots+W_{in}^{(k)}=1,
  \end{split}
\end{equation}
since $\Exp_i[X] = \sum_{j=1}^nW_{ij}^{(k)}\Exp_i[b_j^k(0)]
= \mu_k \sum_{j=1}^nW_{ij}^{(k)} = \mu_k$ and where we assume
row-stochasticity of $\mathbf{W}^{(k)}$. Now,
$\Exp_i\bigl[(X-\Exp_i[X])^2\bigr]=\Var_i[X]$ and hence, agents'
utility maximization problems may be rewritten as
\begin{equation}\label{eq:optim2}
  \begin{split}
  \max_{[\mathbf{W}^{(k)}]_i}&
  -\Var_i[X] = \min_{[\mathbf{W}^{(k)}]_i}
  \Var_i[X]\\
  \text{s.t. } & W_{i1}^{(k)}+\ldots+W_{in}^{(k)}=1,
  \end{split} 
\end{equation}
that is, agents strive to set weights $W_{i1}^{(k)},\ldots,W_{in}^{(k)}$
such that $\Var_i[X]$ is minimized subject to the row-stochasticity
condition on $\mathbf{W}^{(k)}$. To simplify the solution to
problem \eqref{eq:optim2}, we additionally assume 
independence of $b_1^k(0),\ldots,b_n^k(0)$.
\begin{assumption}\label{ass:independence}
  The variables $b_1^k(0),\ldots,b_n^k(0)$ are \emph{independent} random
  variables. 
\end{assumption}
Finally, we assume that agents expect the variables
$b_j^1(0),b_j^2(0),b_j^3(0),\ldots$ to be independent with \emph{identical}
variances. If this were not the case, agents' reliability across topics
would vary so that 
statistical regularities --- inference from past
performance to current performance --- could not be exploited. 
\begin{assumption}\label{ass:regularity}
   Each agent $i\in[n]$ expects the random variables
   $b_j^1(0),b_j^2(0),b_j^3(0),\ldots$ to be independent 
  random variables with identical variances, 
  that is,
  \begin{align*}
  \Var_i[b_j^1(0)] = \Var_i[b_j^2(0)] = \Var_i[b_j^3(0)] = \cdots
  \end{align*}
  for all $j=1,\ldots,n$. 
\end{assumption}

Under Assumptions \ref{ass:independence} and \ref{ass:regularity},
$\Var_i[X]$ may be written as 
\begin{align*}
  \Var_i[X] = \sum_{j=1}^n (W_{ij}^{(k)})^2\Var_i[b_j^k(0)]
  =  \sum_{j=1}^n \alpha_{ij}^2\sigma_{ij}^2,
\end{align*}
where we let, for short, $\alpha_{ij}=W_{ij}^{(k)}$ (here, we may omit
the dependence on $k$ since $W_{ij}^{(k)}$ are optimization variables that
do not, intrinsically, depend on topic $X_k$) and 
$\sigma_{ij}^2=\Var_i[b_j^k(0)]$ (here, we may omit the dependence on $k$ due to Assumption \ref{ass:regularity}). Thus, to solve problem \eqref{eq:optim2}
under Assumptions \ref{ass:independence} and \ref{ass:regularity},
each agent $i=1,\ldots,n$ 
minimizes the `Lagrange' function
\begin{align}
  \mathcal{L}_i(\alpha_{i1},\ldots,\alpha_{in})
  =  \sum_{j=1}^n \alpha_{ij}^2\sigma_{ij}^2-\lambda(\sum_{j=1}^n\alpha_{ij}-1)
\end{align}
for some `Lagrange multiplier' $\lambda$. Via the first-order
conditions, this leads to
\begin{align*}
  \alpha_{ij} = \frac{\lambda}{2\sigma_{ij}^2},
\end{align*}
and from $\sum_{j=1}^n\alpha_{ij}=1$, we find that 
\begin{align*}
  \sum_{j=1}^n \frac{\lambda}{2\sigma_{ij}^2} =
  1 \quad\iff\quad \lambda = \frac{2}{\sum_{j=1}^n \sigma_{ij}^{-2}}. 
\end{align*}
Thus, under Assumptions \ref{ass:1} to \ref{ass:regularity}, a 
rational agent chooses weights $W_{ij}^{(k)}$ that satisfy
\begin{align}\label{eq:optimalweights}
  W_{ij}^{(k)} = \alpha_{ij}
  = \frac{\frac{1}{\Var_i[b_j^k(0)]}}{\sum_{j'=1}^n \bigl(\Var_i[b_{j'}^k(0)]\bigr)^{-1}}
  \propto \frac{1}{\Var_i[b_j^k(0)]}, 
\end{align}
which is quite an intuitive result: the larger the variance of agent
$j$'s estimate $b_j^k(0)$ --- or, more precisely, what agent $i$ thinks
of this variance to be --- the lower should the weight be that agent 
$i$ assigns $j$, since $j$'s initial belief tends to be `away from
truth' more frequently --- or, more precisely, $i$ expects $j$'s initial
belief to be so.  
\subsection*{A comparison with the heuristic
weight adjustment rule \eqref{eq:adjusttruth}} 
To compare the `optimal' weight adjustment rule under
Assumptions \ref{ass:1} to \ref{ass:regularity} with the heuristic
rule \eqref{eq:adjusttruth}, note first that
weight adjustment rule \eqref{eq:adjusttruth} amounts to (weighted)
`counting' of how often a particular agent $j$ has been in an
$\eta$ interval around truth $\mu_k$, since, each time $j$ has
been within this interval, the weight of $i$ for $j$ is increased by
the term $\delta\cdot T(\cdot)$. Hence, denoting the weights defined
via rule \eqref{eq:adjusttruth} by $\tilde{W}_{ij}^{(k)}$ for the
moment and the remainder of this section, we have
\begin{align*}
        \tilde{W}_{ij}^{(k)}\propto R_j^k(\eta),
\end{align*}
where 
$R_j^k(\eta)$ is the number of times
agent $j$ has been in an $\eta$-interval around truth within the
first $k$ discussion topics,
\begin{align*}
R_j^k(\eta)=\abs{\set{h\in\set{1,\ldots,k}\sd \norm{b_j^h(\tau)-\mu_k}< \eta}}. 
\end{align*}
Now, 
if Assumptions \ref{ass:2}, \ref{ass:3}, \ref{ass:independence}, and
\ref{ass:regularity} 
hold and if 
$\tau=0$, then clearly, $R_j^k(\eta)$ is inversely related to
$\Var_i[b_j^k(0)]$, for all $j=1,\ldots,n$, since if $R_j^k(\eta)$
is low, then $i$ thinks that $j$ has high variance (around $j$'s
presumed expected value of $\Exp_i[b_j^k(0)]=\mu_k$) and analogously
if $R_j^k(\eta)$ is high.  Hence, under these 
assumptions, 
weight adjustment rule \eqref{eq:adjusttruth} entails
weights $\tilde{W}_{ij}^{(k)}$ which satisfy
\begin{align*}
  \tilde{W}_{ij}^{(k)}\propto \frac{1}{\Var_i[b_j^k(0)]}. 
\end{align*}
Thus, to summarize, if
\begin{itemize}
        \item Assumptions \ref{ass:1} to \ref{ass:regularity} hold
        and if,
        \item $\tau=0$ (adjusting based on initial beliefs), 
\end{itemize}
then, heuristic weight adjustment rule \eqref{eq:adjusttruth}
corresponds, 
by analogy, 
to an adjustment rule that a
rational agent would implement, under the named assumptions. 

\subsection*{Discussion}
Some of the assumptions we have made require further
discussion. Assumption \ref{ass:1}, which says that agents have
limited foresight and want to minimize the distance between $b_i^k(1)$
and $\mu_k$, rather than between $b_i^k(\infty)$ and $\mu_k$, may not
only be perceived as the choice of a boundedly rational agent. In contrast, if
agent $i$ knows, or at least assumes, that all agents are similarly
rational as her (and share the same information structure, etc., that
is, are perfectly homogeneous) ---
whence all agents are faced with the same 
optimization problems to which they derive identical solutions
$[\mathbf{W}^{(k)}]_i$ --- then, in fact,
$\mathbf{b}^{k}(1)=\mathbf{b}^{k}(\infty)$ since $\mathbf{W}^{(k)}$,
for each $k$, is
identical in each row and is row-stochastic (see Lemma
\ref{lemma:identicalRows} in Appendix 
\ref{sec:appendix}). So, under this prerequisite, agents could also be thought of as
having `perfect foresight', knowing that $b_i^k(1)$ will equal
$b_i^k(\infty)$ anyways. Assumption \ref{ass:2} is innocuous, while
Assumption \ref{ass:3} is the bona fides assumption discussed in the
introduction, which we thought of as being based on egocentric
biases. Next, Assumption \ref{ass:independence}, that agents' initial
beliefs are independent, is highly implausible, of course: individuals
go to the same or similar schools, are influenced by the same or similar
media, etc., all of which may induce correlation in individuals
beliefs (possibly even if we think of these beliefs as prior to social
communication); we make this assumption for technical ease, as
otherwise deriving closed-form solutions to the optimization problems
in question may be quite challenging. Finally, Assumption
\ref{ass:regularity} demands that topics are of the same general `area',
as we have indicated in Section \ref{sec:modelECCS}, whence one may expect
individuals' reliability (for this field of human expertise) to be
predictable across a multitude of topics. Forfeiting the assumption
would mean to present agents with a problem where nothing can be
learned, in terms of adjusting the network structure $\mathbf{W}$,
across various topics.  

We also mention that our above analysis has assumed that
$\delta\cdot T$ is strictly positive (for all or at least infinitely
many topics $X_k$), for, e.g., otherwise $\tilde{W}_{ij}^{(k)}$ would
not be proportional to $R_j^k(\eta)$. In Section \ref{sec:standard},
we also consider the case when $\delta\cdot T$ is zero for all but
finitely many topics. We treat this case, which allows us to derive
wisdom results in certain circumstances even under the presence of
biased agents, as a special (or, `extreme')
case of our model that differs, however, from the choice a rational
agent would pursue, as we have sketched. 

\subsection*{Illustration}
To illustrate the relationship between $W_{ij}^{(k)}$ as set by a
(boundedly) rational agent and as set via (heuristic) weight
adjustment rule 
\eqref{eq:adjusttruth}, consider the following exemplary
situation. Let there be $n$ agents, all of whose initial beliefs
$b_i^k(0)$ are \emph{normally} distributed around $\mu_k$, for all
topics $X_k$. Assume that there are two types of agents, $L$ and $H$, with
variances $\sigma_L^2$ and $\sigma_H^2$, respectively, such that
$\sigma_L^2<\sigma_H^2$. Let there be $n_L$ agents of type $L$ and
$n_H$ agents of type $H$ such that $n_L+n_H=n$. In other words, for
each $L$-type agent $i_L$, we 
have, for all $k=1,2,3,\ldots$, 
\begin{align*}
  b_{i_L}^k(0) \sim \mathtt{N}(\mu_k,\sigma_L^2),
\end{align*}
and, accordingly, for each $H$-type agent $i_H$, we have
\begin{align*}
  b_{i_H}^k(0) \sim \mathtt{N}(\mu_k,\sigma_H^2). 
\end{align*}
Thus, under Assumptions \ref{ass:1} to \ref{ass:regularity}, a
rational agent $i$ would assign weights,
\begin{align}\label{eq:opt}
  W_{ij}^{(k)} = \frac{1}{\sigma_T^2}\cdot \frac{1}{C},
\end{align}
where $T\in\set{L,H}$, depending on whether $j$ is of type $L$ or
$H$, and where $C$ is the constant
$C=\frac{n_L}{\sigma_L^2}+\frac{n_H}{\sigma_H^2}$. In contrast, an
agent who sets weights according to the rule \eqref{eq:adjusttruth},
would set
\begin{align}\label{eq:approx}
  \tilde{W}_{ij}^{(k)} \propto \text{Pr}[\norm{b_j^k(0)-\mu_k}<\eta] =
  \int_{-\eta}^{\eta}
  \frac{1}{\sqrt{2\pi\sigma_T^2}}\exp(-\frac{x^2}{2\sigma_T^2})\,dx
  = 2\int_{0}^{\eta}
  \frac{1}{\sqrt{2\pi\sigma_T^2}}\exp(-\frac{x^2}{2\sigma_T^2})\,dx,
\end{align}
depending on whether $j$ is of type $T=L$ or $T=H$. In Figure
\ref{fig:illustration}, we plot the behavior of \eqref{eq:opt}
vs.\ \eqref{eq:approx} for specific values of $\sigma_L^2$ and
$\sigma_H^2$, namely $\sigma_L^2=1$ and $\sigma_H^2=2$. 
For the
values of $\sigma_L^2$ and $\sigma_H^2$ discussed, the optimal rule
under Assumptions \ref{ass:1} to 
\ref{ass:regularity} would accord total weight mass for $L$-types of
$n_L\cdot W_{ij}^{(k)}=\frac{3}{4}$ (for $j$ of type $L$), and total
weight mass for $H$-types of $n_H W_{ij}^{(k)}=\frac{1}{4}$ (for $j$ of
type $H$).      
\begin{figure}[!htb]
\centering
\begin{subfigure}{.48\textwidth}
  \centering
  \includegraphics[width=0.98\linewidth]{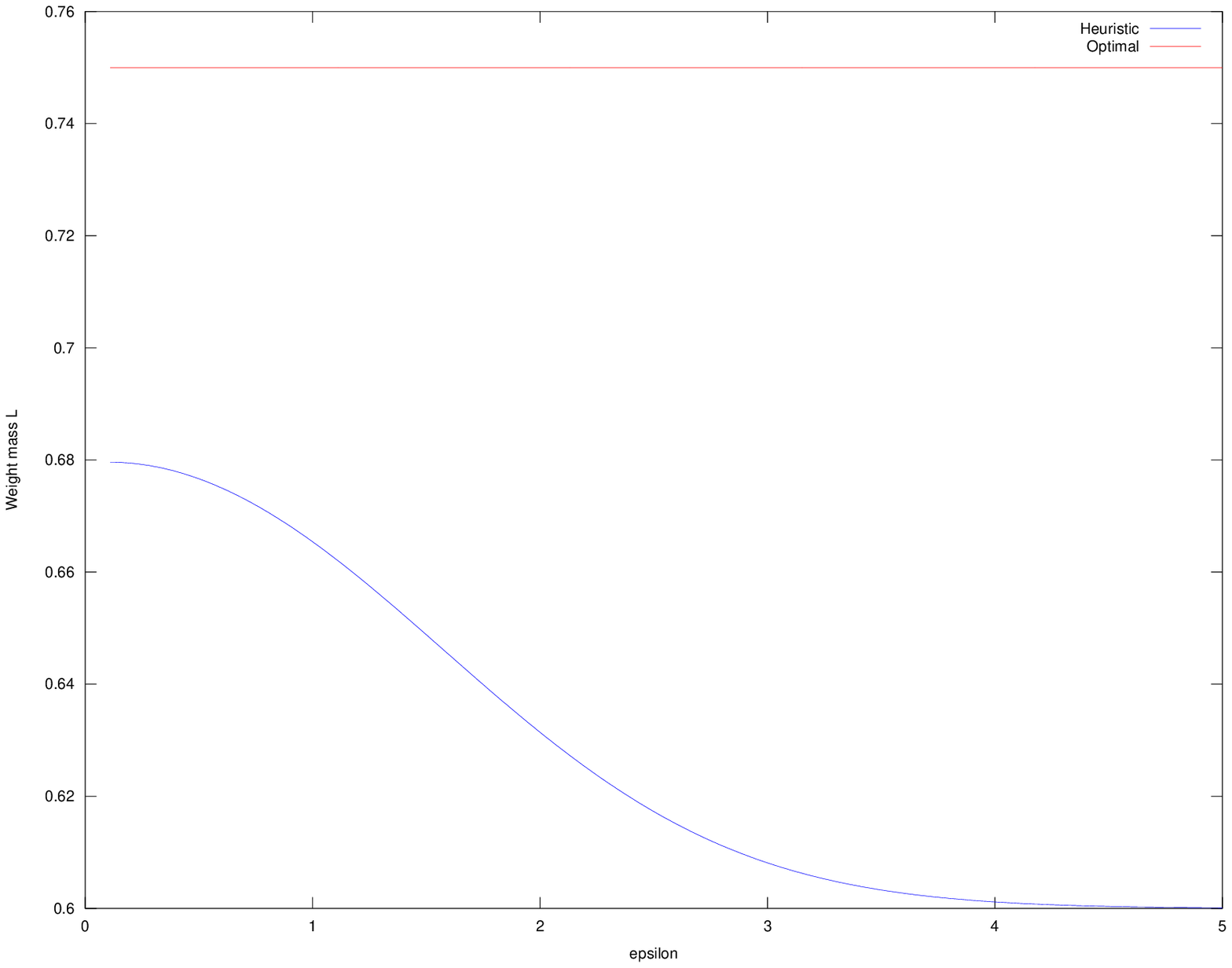}
  \caption{{\small Weight mass for $L$ types; optimal vs.\ heuristic,
    $\tilde{W}_{ij}^{(k)}$, as a 
    function of $\eta$; $\sigma_L^2=1$ and $\sigma_H^2=2$ fixed.}} 
  \label{fig:UpdateInitLow}
\end{subfigure}%
\begin{subfigure}{.48\textwidth}
  \centering
  \includegraphics[width=0.98\linewidth]{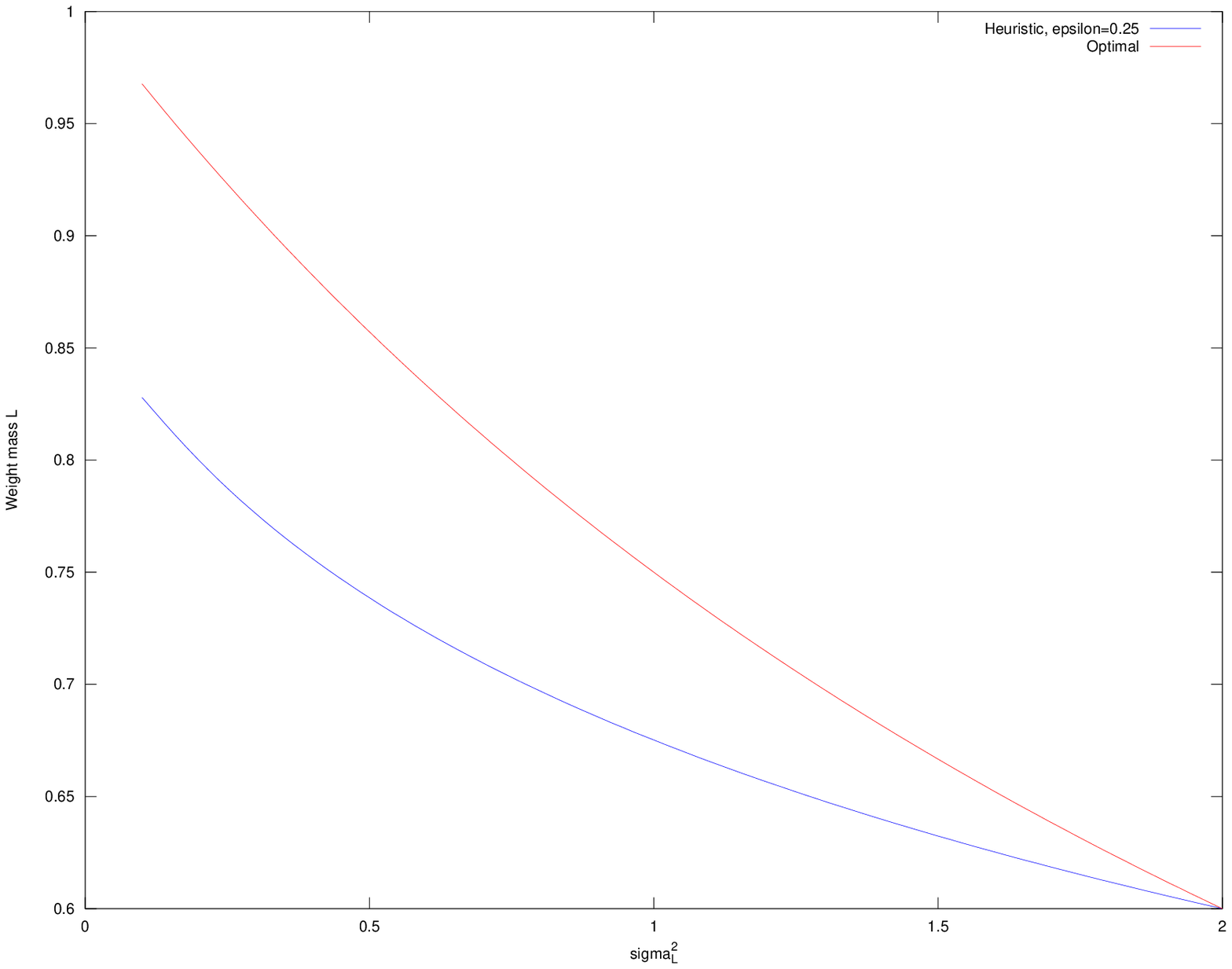}
  \caption{{\small Weight mass for $L$ types; optimal vs.\ heuristic,
    $\tilde{W}_{ij}^{(k)}$, as a
    function of $\sigma_L^2$; $\eta=0.25$ fixed.}}
  \label{fig:UpdateInitHigh}
\end{subfigure}\\
\begin{subfigure}{.48\textwidth}
  \centering
  \includegraphics[width=0.98\linewidth]{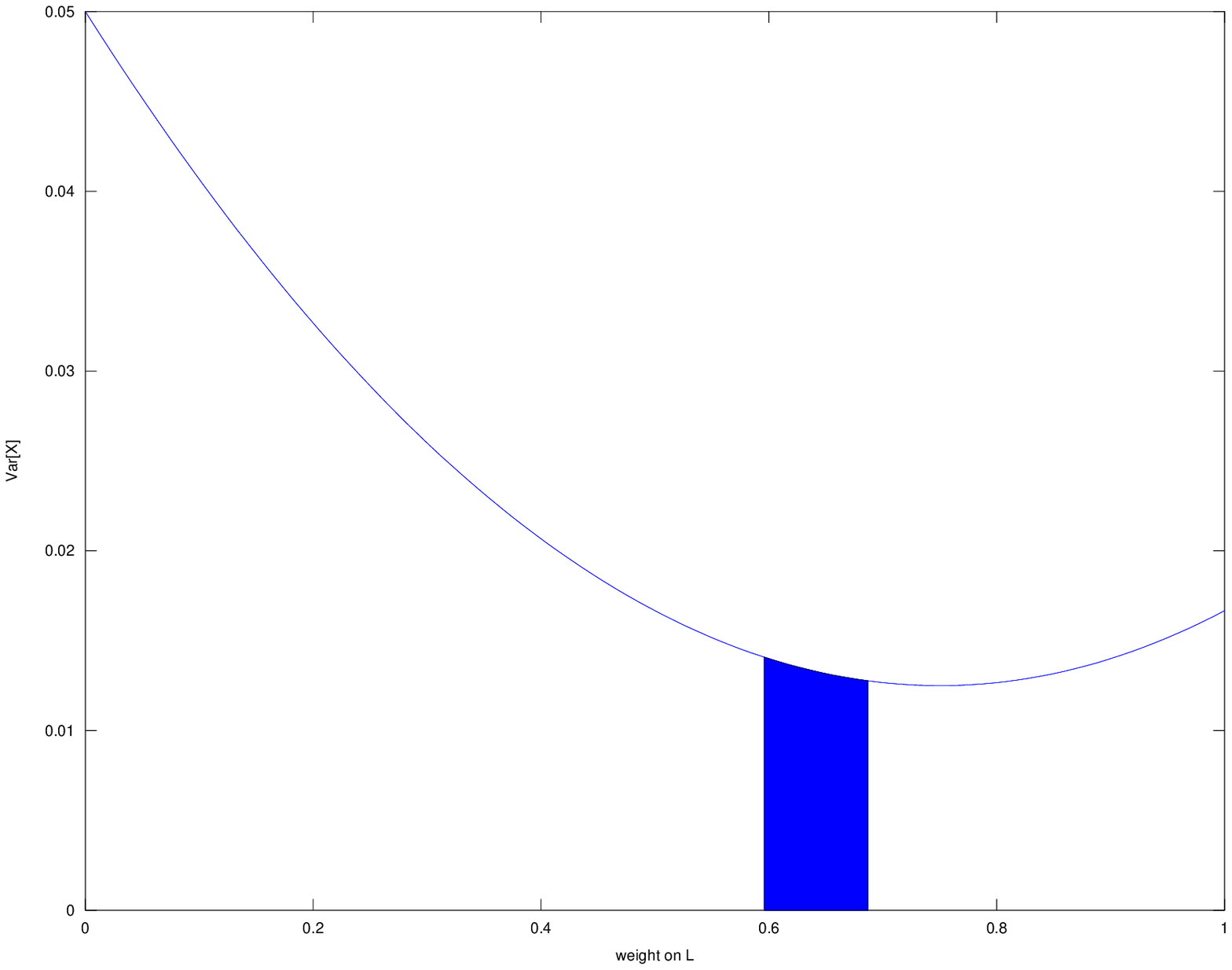}
  \caption{{\small Variance of $X$ as defined in \eqref{eq:X} as a function of
  weight mass assigned to $L$ types. The colored area gives the range
  of $\tilde{W}_{ij}^{(k)}$ as illustrated in (a).}}
  \label{fig:sub3}
\end{subfigure}
\begin{subfigure}{.48\textwidth}
  \centering
  \includegraphics[width=0.98\linewidth]{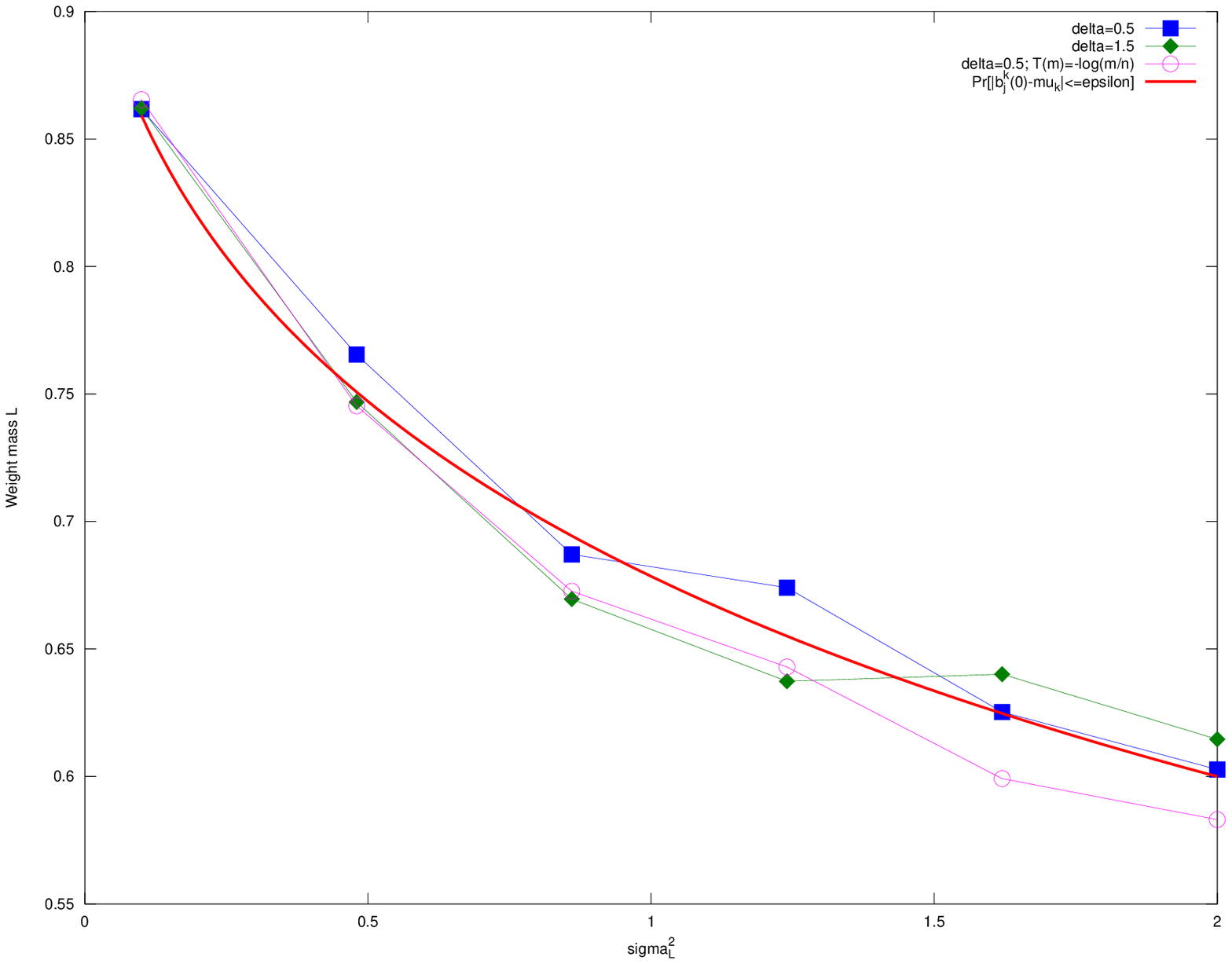}
  \caption{
    {\small $\tilde{W}_{ij}^{(k)}$ as theoretically computed according to
      \eqref{eq:approx} (as in (b)) as a function of $\sigma_L^2$,
      and, as a comparison, 
      $\tilde{W}_{ij}^{(k)}$ as given by sample runs, for two different values
    of $\delta=0.5,1.5$, and functions $T$ ($T\equiv 1$ and
    $T(m)=-\log(m/n)$); averaged over $10$ runs; 
    $\eta=0.25$ fixed.}
  }
  \label{fig:sub4}
\end{subfigure}
\caption{Throughout $\sigma_H^2=2$ and $n=100=60+40=n_L+n_H$.}
\label{fig:illustration}
\end{figure}
In contrast, as the graphs show, if weights are set according to
\eqref{eq:approx}, 
then, $\tilde{W}_{ij}^{(k)}$ is, for the $L$ types, always lower than
$\frac{3}{4}$, as the optimal rule would prescribe. Depending on
$\eta$, $\tilde{W}_{ij}^{(k)}$ ranges from $0.60$, if $\eta$ is
large, to about $0.68$, for small $\eta$. The value for $\eta$
large is obvious since if $\eta$ is sufficiently large in size,
then each agent will receive identical weight $\tilde{W}_{ij}^{(k)}$,
$\frac{1}{n}$, and, hence, total weight mass for $T$-types is
$\frac{n_T}{n}$, for $T\in\set{L,H}$. The figure also shows the
inverse relationship between $\tilde{W}_{ij}^{(k)}$ and $\sigma_T^2$ (for
$T=L$, in this case; Figure \ref{fig:illustration} (b)), the closeness
of $\tilde{W}_{ij}^{(k)}$ to 
`optimality' (Figure \ref{fig:illustration} (c)), and a comparison
between the theoretic value $\tilde{W}_{ij}^{(k)}$ is proportional to,
$\text{Pr}[\norm{b_j^k(0)-\mu_k}< \eta]$, and actual realizations of
  $\tilde{W}_{ij}^{(k)}$ as a function of $\delta$ and $T$ (Figure
\ref{fig:illustration} (d); cf.\ Equation
  \eqref{eq:adjusttruth}).  

\section{Notation and definitions}\label{sec:notation}
We introduce the following helpful notation and definitions.


\begin{definition}
  Let any $\epsilon\ge 0$ be fixed. 
  We call an agent $i$ \emph{$\epsilon$-intelligent for (topic) $X_k$} if $i$'s
  initial belief on $X_k$ is ($\epsilon$) `close to truth', i.e., 
  $\norm{b_i^k(0)-\mu_k}< \epsilon$. We call $i$
  \emph{$\epsilon$-intelligent}, if $i$ is {intelligent} for all topics $X_k$. 
\end{definition}

This definition captures the idea that an agent's initial beliefs,
which we think of as
 not influenced by peers (or their beliefs), 
express something 
\emph{innate} to agent $i$, his \emph{hidden ability} or, simply,
\emph{intelligence}. However, we say nothing here on how $i$ has
arrived at his initial beliefs, e.g., whether it was through hidden
ability in a proper sense or, for instance, `merely' through 
guessing. 
We also remark that the concept of $\epsilon$-intelligence (or
$\epsilon$-wisdom, as we define below) is clearly
related to our 
weight adjustment 
rule; in particular, for given
tolerance $\eta$, agents increase their weight for an agent $j$ if
this agent is $\epsilon$-intelligent (or $\epsilon$-wise) for a topic
$X_k$ and for all
$\epsilon\le \eta$. 

When $i$ is `close to truth' in the limit of the DeGroot
learning process, we call $i$ \emph{wise}.

\begin{definition}
  We call an agent $i$ \emph{$\epsilon$-wise for (topic) $X_k$} if $i$'s
  limit belief on $X_k$ is `close to truth', i.e., 
  $\norm{b_i^k(\infty)-\mu_k}< \epsilon$. We call $i$
  \emph{$\epsilon$-wise}, if $i$ is {wise} for all topics $X_k$. 
\end{definition}

We also introduce stochastic analogues of the above definitions. If an
agent has initial beliefs 
stochastically centered
around truth for a topic, we call the agent \emph{stochastically
  intelligent} 
\emph{for this
  topic}. 

\begin{definition}
  We call an agent $i$ \emph{stochastically intelligent for (topic)
    $X_k$} if $i$'s 
  initial belief on $X_k$ is `stochastically centered around truth', i.e., 
  $b_i^k(0)=\mu_k+\sigma_{ik}$, 
  where $\sigma_{ik}$ is some
  individual and topic-specific white-noise variable. We call $i$
  \emph{stochastically intelligent}, if $i$ is {stochastically
    intelligent} for all topics $X_k$.  
\end{definition}
We omit the corresponding definition for wisdom since we rarely 
make
use of a concept of `stochastic wisdom' in the remainder of this work. 

Next, fix a level of intelligence or wisdom $\epsilon\ge 0$. 
For convenience, let us denote the open $\epsilon$-interval around
truth, within with agents 
are considered $\epsilon$-intelligent (or $\epsilon$-wise), by
$B_{k,\epsilon}$ and its complement by $B_{k,\epsilon}^c$. Formally,
we have: 
\begin{definition}
  \begin{align*}
    B_{k,\epsilon} &:= (-\mu_k-\epsilon,\mu_k+\epsilon),\\
    B_{k,\epsilon}^c &= S\wo B_{k,\epsilon}.
  \end{align*}
\end{definition}

Below, in the main sections of our work, our principal modeling 
perspective --- although we may occasionally deviate from or slightly
generalize this perspective --- is the 
notion of two groups of agents, $\struct{N}_1$ and $\struct{N}_2$ with
$\struct{N}_1\cup\struct{N}_2=[n]$ and
$\struct{N}_1\cap\struct{N}_2=\emptyset$, one of whose initial beliefs
are 
\emph{unbiased} --- group $\struct{N}_1$'s --- and the 
other's initial beliefs are \emph{biased}, whereby we define bias as
\begin{align*}
  \beta_{i,k} = \norm{\Exp[b_i^k(0)]-\mu_k}.
\end{align*}
Hence, for members $i$ of $\struct{N}_1$, we assume that $\beta_{i,k}=0$
and 
for members $i$ of $\struct{N}_2$, we assume that $\beta_{i,k}>0$ for all
topics $X_k$. 
In addition, we think of the two groups of agents as having
independent and identical distributions of initial beliefs, with
distribution functions 
$F_{\struct{N}_l,k}(A)=\text{Pr}[b_i^k(0)\in A]$, for $l=1,2$ and
$A\subseteq S$, where, of
course, identical distribution refers to within group and independence
refers to both within and across group relations. 
Finally, for fixed level of tolerance $\eta\ge 0$, we assume that
$F_{\struct{N}_l,k}(B_{k,\eta})$ does not depend upon $k$, that is,
$F_{\struct{N}_l,k}(B_{k,\eta})=F_{\struct{N}_l,k'}(B_{k',\eta})$, for
all $k,k'$. This means that agents' probability of being
within an $\eta$-interval around truth --- for initial beliefs --- is
the same across topics. This assumption is very similar, in spirit, to
Assumption \ref{ass:regularity} and captures predictability of
agents.
We also think of this invariant probability as
denoting 
an agent's ability or reliability. 

To conclude this section, we introduce notation regarding convergence
(and consensus) of our endogenous opinion dynamics paradigm. 
\begin{definition}
  Let $k\ge 1$ be arbitrary. 
  We say that $\mathbf{W}^{(k)}$ is
  \emph{convergent for opinion vector $\mathbf{b}(0)\in S^n$} if 
  $\lim_{t\goesto\infty}(\mathbf{W}^{(k)})^t\mathbf{b}(0)$
  exists.
  Moreover,
  we say that  
  $\mathbf{W}^{(k)}$ \emph{induces a consensus for opinion
  vector $\mathbf{b}(0)$} if
  $\mathbf{W}^{(k)}$ is {convergent for $\mathbf{b}(0)$} and
  $\lim_{t\goesto\infty}(\mathbf{W}^{(k)})^t\mathbf{b}(0)$ is a
  \emph{consensus}, that is, a vector $\mathbf{c}\in S^n$ with all
  entries identical.
\end{definition}
Rather than saying that $\mathbf{W}^{(k)}$ converges, we may
occasionally also say that beliefs converge (under $\mathbf{W}^{(k)}$)
or that our DeGroot learning / opinion dynamics paradigm converges. We
also mention that we typically assume matrix $\mathbf{W}^{(1)}$ to be
the $n\times n$ identity matrix (in the absence of further
information, agents 
follow their own signals), which sometimes facilitates analytical
derivations, but we also consider more general forms of the matrix
$\mathbf{W}^{(1)}$, where we find that such a generalization is
worthwhile mentioning. 

Throughout our work, we assume that weight matrices $\mathbf{W}^{(k)}$ are
\emph{row-stochastic}, that is, 
\begin{align*}
  \sum_{j=1}^n [\mathbf{W}^{(k)}]_{ij} = 1,
\end{align*}
for all $i\in[n]$. We denote the entries of an arbitrary matrix
$\mathbf{A}$ by $A_{ij}$ or $[\mathbf{A}]_{ij}$. We denote by
$\mathbf{I}_n$ the $n\times n$ identity matrix and by $\one_n$ the
vector of $n$ $1$'s, i.e., $\one_n=(1,\ldots,1)^\intercal$. We may
omit the dimensionality if it is clear from the context. 

\section{The standard DeGroot model}\label{sec:standard}
In the subsequent sections, we derive a few results regarding the
standard DeGroot learning model under our endogenous weight formation
paradigm. First, we show that, in our setup, agents almost always
reach a consensus (Proposition \ref{prop:consensus} and the subsequent
remark), that is, for
almost all topics $X_k$, under very mild conditions. Then, in Section
\ref{sec:unbiased}, we show that if agents are unbiased and receive
initial belief signals that are centered around truth, then agents'
beliefs converge to truth for topics $X_k$, as $n,k\goesto\infty$,
irrespective of whether agents adjust weights based on limiting or on
initial 
beliefs. Next, in Section \ref{sec:biased}, we illustrate that agents may be
arbitrarily far off from truth as the number of biased agents involved
in the opinion dynamics process becomes large, thus demonstrating that
crowd wisdom may fail under these circumstances. For the situation 
when $T(n)=0$, we also give sufficient conditions on when crowd wisdom
does not fail, even under the presence of biased agents. In Section
\ref{sec:demarzo}, we discuss weights on own beliefs as a (simple)
extension of the classical DeGroot learning paradigm and as discussed
by DeMarzo, Vayanos, and Zwiebel (2003) \cite{Demarzo2003}. 

We start our discussion with a theorem given in the original DeGroot
paper \cite{DeGroot1974}, which helps us determining when our
endogenous opinion 
dynamics process leads agents to a consensus. 
\begin{theorem}\label{theorem:degroot}
  If there exists a positive integer $t$ such that every element in at
  least one column of the matrix $\mathbf{W}^t$ is positive, then 
  $\mathbf{W}$ induces a consensus for any vector $\mathbf{b}(0)\in
  S^n$. 
\end{theorem}
Theorem \ref{theorem:degroot} can be used in a straightforward manner to
derive conditions, in our setup, under which agents reach a
consensus. Namely, 
during the course of dicussing issues $X_1,X_2,X_3,\ldots$, as
long as no agent has been $\eta$-intelligent (resp.\ $\eta$-wise),
agents do not adjust 
their weights to other agents, and, consequently, agents reach a
consensus if and only if $\mathbf{W}^{(1)}$ induces a consensus. At
the first time point that some agent has been $\eta$-intelligent
(resp.\ $\eta$-wise),
all agents subsequently adjust weights for this agent, and, hence, (at
least) one 
column of 
the respective weight matrix is strictly positive for the subsequent
topic. Hence, for this topic, all agents reach a consensus. But note
that this column remains positive for \emph{all} weight matrices
corresponding to discussion topics discussed 
thereafter (as can easily be shown inductively) because even
redistribution of weight mass to other agents, via weight
normalization, cannot make a matrix entry zero once it has been
positive. Now, we formalize these simple ideas. Then, we
generalize to the setting when agents have individualized tolerances
$\eta_i$. 

Let $A_i$ be the set of time points agent $i$ is
$\eta$-intelligent (resp.\ $\eta$-wise) for some topic $X_k$,
\begin{align*}
  A_i = \set{k\in\nn\sd \norm{b_i^k(\tau)-\mu_k}< \eta}\subseteq\nn,
\end{align*}
where $\tau=0$ (resp.\ $\tau=\infty$) 
and let $a_i$ be the first time that $i$ is $\eta$-intelligent for
some topic $X_k$,
\begin{align*}
  a_i = \min A_i.
\end{align*}
Then, we have the following proposition, for which we assume that
$T(\cdot)>0$ on its whole domain. This assumption is innocuous here;
if it does not hold, the proposition may easily be adjusted 
to account for the different setup. 
\begin{proposition}\label{prop:consensus}
  Let $\eta\ge 0$ be fixed. Let $\tau=0$ (resp.\ $\tau=\infty$).  
  Let $r=\min_{i\in[n]} a_i$ be the earliest time point that some
  agent is $\eta$-intelligent (resp.\ $\eta$-wise) for topic $X_r$. 
  (a) Then agents reach a consensus for all topics $X_k$ with $k>r$,
  independent of their initial beliefs. (b) For topics $1,\ldots,r$,
  agents reach a consensus if and only if $\mathbf{W}^{(1)}$ induces a
  consensus.  
\end{proposition}
\begin{proof}
  (a) By the proposition, we know that some agent $i$ is
  $\eta$-intelligent (resp.\ $\eta$-wise) for topic
  $X_r$. Accordingly, agents increase 
  their weight to $i$ by $\delta\cdot T(\cdot)>0$ at time $r+1$. Hence, weight
  matrix $\mathbf{W}^{(r+1)}$ has a strictly positive column and so do,
  in general, have all matrices $\mathbf{W}^{(k)}$, for $k>r$. By Theorem
  \ref{theorem:degroot}, agents thus reach a consensus for all issues
  $X_k$, with $k>r$. 

  (b) For issues $X_1,\ldots,X_r$, no weight adjustments are made, whence
  $\mathbf{W}^{(1)}=\cdots=\mathbf{W}^{(r)}$ and a consensus is reached if and
  only if $\mathbf{W}^{(1)}$ induces a consensus. 
\end{proof}
\begin{remark}
  Assume, for the moment, that agents have individualized tolerances 
  $\eta_i$. Then part (a) of Proposition \ref{prop:consensus} is
  true if we 
  replace $A_i$ as above by
  \begin{align*}
    A_i = \set{k\in\nn\sd \norm{b_i^k(\tau)-\mu_k}< \min_{j\in[n]}\eta_j},
  \end{align*}
  and we define $a_i$ as above as $a_i=\min A_i$. 
\end{remark}
\begin{remark}\label{rem:distR}
  Consider $\tau=0$ for this remark. 
  If initial beliefs are random variables, then $r$, as specified in
  Proposition \ref{prop:consensus}, is a random variable (which we
  could consider a `stopping time'). Accordingly, its distribution
  might be of interest. Assuming 
  agent $i$'s initial opinions for
  each topic $X_k$ 
  to be 
  distributed with distribution function $F_{i,k}$,
  that is, $P[b_i^k(0)\in A]=F_{i,k}(A)$, for $A\subseteq S$, 
  we have that the
  probability that at least one agent $i$ is $\eta$-intelligent
  for topic $X_k$ is given by
  $p_{k,\eta}=1-\prod_{i\in[n]}F_{i,k}(B_{k,\eta}^c)$, due to independence of
  agents' initial beliefs. Then, if $F_{i,k}(B_{k,\eta}^c)$ does not depend
  on $X_k$ but only on $\eta$, we have that $r$ has a geometric
  distribution with probability $p_{\eta}$ (where we omit, in the
  notation, the 
  dependence on $k$ due to our assumption), that is, 
  \begin{align*}
    P[r=\nu] = (1-p_\eta)^{\nu-1}p_\eta,\quad \text{for }\nu=1,2,3,\ldots
  \end{align*}
  From the specification of $p_\eta$, we thus see that if
  $F_{i,k}(B_{k,\eta}^c)<1$ for all $i$, then $p_\eta\goesto 1$ as
  $n\goesto\infty$. Accordingly, the distribution of $r$ converges to
  the degenerate distribution with $P[r=1]=1$ and $P[r\neq
    1]=0$. Thus, in this situation, agents `almost always' --- that is,
  with possibly  
  only finitely many, namely, one, exceptions, topic $X_1$ --- reach a 
  consensus for topics $X_k$, for $k=1,2,3,\ldots$. 
\end{remark}
We also find the next simple result which states that if \emph{all}
agents start with initial beliefs within a precision of $\epsilon$
around truth, then
agents will also end up with limiting beliefs with level of wisdom of
$\epsilon$, provided that agents' beliefs convergence at all, as time
goes to infinity. 
\begin{proposition}\label{prop:wiseTrivial}
  Let level of intelligence $\epsilon\ge 0$ be fixed. 
  If all agents are $\epsilon$-intelligent for $X_k$ and the DeGroot
  learning  
  process \eqref{eq:degrootupdate} converges, then all are
  $\epsilon$-wise for 
  topic 
  $X_k$.
\end{proposition}
\begin{proof}
  This simply follows from the fact that the interval
  $B_{k,\epsilon}=(\mu_k-\epsilon,\mu_k+\epsilon)$ is a convex set and weights
  are always row-stochastic in our model setup. Thus, if all agents
  start their beliefs in $B_{k,\epsilon}$, limit beliefs will also be
  in $B_{k,\epsilon}$, 
  provided that they converge.
\end{proof}
As we have seen in Proposition \ref{prop:consensus}, whether or not
the DeGroot learning process \eqref{eq:degrootupdate} converges on the
first $r$ topics depends on the initial weight matrix
$\mathbf{W}^{(1)}$. Thereafter, convergence (even to consensus) is
guaranteed. 
Hence, using Proposition \ref{prop:wiseTrivial}, we obtain: 
\begin{corollary}
  Let level of intelligence $\epsilon\ge 0$ be fixed. 
  If all agents are $\epsilon$-intelligent (i.e., for all topics
  $X_k$), then all agents are $\epsilon$-wise for all topics $X_k$,
  with $k>r$, where $r$ is defined as in Proposition
  \ref{prop:consensus}. 
\end{corollary}

\subsection{Unbiased agents}\label{sec:unbiased}
In this setup, we assume that \emph{all} agents receive initial signals 
\begin{align}\label{eq:centered}
  b_i^k(0) = \mu_k+\epsilon_{ik},
\end{align}
where $\mu_k$ is truth for issue $X_k$ and $\epsilon_{ik}$ is white
noise (i.e., with mean zero and independent of other variables) with
variance $\sigma_i^2=\Var[\epsilon_{ik}]$ (note that we assume the
variance to be independent of the issue $X_k$). As
throughout, we assume agents' initial signals to be independent. 

We consider first the situation when agents adjust weights based on
limiting beliefs, i.e., $\tau=\infty$. In the next proposition, we
show that agents become $\epsilon$-wise in this situation (for any
$\epsilon>0$), in the 
limit as both $n$, population size, and $k$, which indexes topics, go
to infinity. The intuition behind this result is simple: 
since, in our setup, agents tend
toward a consensus (see Proposition \ref{prop:consensus}), agents will
generally all be \emph{jointly} $\eta$-wise (where $\eta$ is agents'
tolerance) or not.  
Then, since agents
adjust based on limiting beliefs, 
agents receive the same increments (or not) to their weight
structure, so that, as $k$ becomes large, $\mathbf{W}^{(k)}$ is the
matrix with entries $\frac{1}{n}$, approximately. Then, the law of
large number implies convergence to truth, as $n$ becomes large, since
initial beliefs are stochastically centered around truth by
\eqref{eq:centered}.  
\begin{proposition}\label{prop:centered1}
  Let $\eta\ge 0$ be fixed. 
  Assume that agents' initial beliefs are centered around truth in the
  form \eqref{eq:centered}. Moreover, assume that agents initially
  follow their own beliefs, that is, $\mathbf{W}^{(1)}$ is the
  $n\times n$ identity matrix $\mathbf{I}_n$. Finally, assume that
  agents adjust 
  weights based on limiting beliefs, i.e., $\tau=\infty$. Let
  $T(\cdot)> 0$. 
  Then, as $k,n\goesto\infty$, all agents become $\epsilon$-wise for
  topics $X_k$, for all $\epsilon>0$, almost surely. 
\end{proposition}
\begin{proof}
As before, let
$r=\min_{i\in[n]}a_i$ be the first time point that one agent is
$\eta$-intelligent for topic $X_r$ (which is the same as
$\eta$-wise for topic $X_r$, as $\mathbf{W}^{(1)}$ is the $n\times
n$ identity matrix, by assumption). 
For
simplicity, assume first that, for topic $X_r$, \emph{all} agents are
$\eta$-intelligent (and hence, $\eta$-wise); we then treat 
the more general case where only some agents are
$\eta$-intelligent for $X_r$ as an analogous situation. In this case,
$\mathbf{W}^{(r+1)}$ looks as follows, after weight adjustments,
\begin{align*}
  \mathbf{W}^{(r+1)} = \frac{1}{1+n\tilde{\delta}}\begin{pmatrix}
    1+\tilde{\delta} & \tilde{\delta} &\ldots & \tilde{\delta}\\
    \tilde{\delta} & 1+\tilde{\delta} & \ldots & \tilde{\delta}\\
    \vdots & \cdots & \ddots & \vdots\\
    \tilde{\delta} & \tilde{\delta} & \ldots & 1+\tilde{\delta}
    \end{pmatrix},
\end{align*}
where we let $\tilde{\delta}=\delta\cdot T(\cdot)$. 
Consider any matrix $\mathbf{A}$ of the form 
\begin{align}\label{eq:formA}
  \mathbf{A} = 
  \begin{pmatrix}
    \beta & \alpha & \ldots & \alpha\\
    \alpha & \beta & \ldots & \alpha\\
    \vdots & \ldots & \ddots & \vdots\\
    \alpha & \alpha & \ldots & \beta
  \end{pmatrix}
\end{align}
such that $\beta+(n-1)\alpha=1$ (that is, $\mathbf{A}$ is
row-stochastic), with $0<\alpha,\beta<1$. 
In Appendix \ref{sec:appendix}, we show 
that matrix  $\mathbf{A}$ has one eigenvalue $\lambda=1$, to which
corresponds an eigenvector $\mathbf{c}=(c,\ldots,c)^\intercal$, and
$(n-1)$ 
identical eigenvalues of absolute size smaller than $1$. Moreover,
since $\mathbf{A}$ is symmetric, it is diagonalizable of the form
$\mathbf{A}=\mathbf{U}\mathbf{V}\mathbf{U}^\intercal$, where
$\mathbf{V}$ is a diagonal matrix that contains the eigenvalues of
$\mathbf{A}$ on the diagonal and $\mathbf{U}$ is orthonormal, that is,
$\mathbf{U}\mathbf{U}^\intercal=\mathbf{I}_n$; without loss of generality,
assume that the eigenvalues in $\mathbf{V}$ are arranged by size,
i.e., $V_{11}=1> V_{22}=\cdots= V_{nn}$ and the corresponding
eigenvectors are located in the respective columns of $\mathbf{U}$,
i.e., the first column of $\mathbf{U}$ is the vector $\mathbf{c}$. 
 We have
\begin{align*}
  \mathbf{A}^t = \mathbf{U}\mathbf{V}^t\mathbf{U}^\intercal.
\end{align*}
As $t\goesto\infty$, $\mathbf{V}$ converges to the matrix with one
entry equal to $1$ and all other entries equal to zero (due to the
eigenvalue structure of $\mathbf{A}$). Thus, we then have
\begin{align*}
  \lim_{t\goesto\infty}\mathbf{A}^t = \begin{bmatrix}\mathbf{c} &\mathbf{0} &\ldots
    &\mathbf{0}\end{bmatrix}\mathbf{U}^\intercal = \begin{bmatrix}\mathbf{c} &\mathbf{0} &\ldots
    &\mathbf{0}\end{bmatrix}\begin{bmatrix}\mathbf{c}^\intercal\\ \mathbf{c}_2^\intercal\\ \vdots\\ \mathbf{c}_n^\intercal\end{bmatrix}
  = c^2\begin{pmatrix}1 & \ldots & 1\\ \vdots & \ddots & \vdots \\ 1 & \ldots & 1\end{pmatrix},
\end{align*}
where $\mathbf{c}_2,\ldots,\mathbf{c}_n$ are the eigenvectors
corresponding to eigenvalues $\lambda_2$ to $\lambda_n$. 
Moreover, since $\mathbf{A}$ is row-stochastic, $\mathbf{A}^t$ is
row-stochastic for every $t$, and, accordingly, $\lim_t A^t$ is
row-stochastic. Therefore $c^2=\frac{1}{n}$. In other words, if each
agent is $\eta$-wise for topic $X_r$, then for topic $X_{r+1}$, we
have 
\begin{align}\label{eq:average}
  \mathbf{b}^{r+1}(\infty)=\lim_{t\goesto\infty}(\mathbf{W}^{(r+1)})^t\mathbf{b}^{r+1}(0)
  =
  \left(\sum_{j=1}^n\frac{1}{n}b_j^{r+1}(0)\right)\begin{pmatrix}1\\ \vdots
    \\ 1\end{pmatrix}
  = \left(\frac{\sum_{j=1}^nb_j^{r+1}(0)}{n}\right)\begin{pmatrix}1\\ \vdots
    \\ 1\end{pmatrix}. 
\end{align}
Now, for all topics $X_k$, with $k>r$, agents reach a
consensus by Proposition \ref{prop:consensus}. Hence, agents are
either all jointly $\eta$-wise or 
none of them is, for all $k>r$. Therefore, all weight matrices
$\mathbf{W}^{(k)}$, 
for $k>r$, 
have the form \eqref{eq:formA} (either all entries receive an
increment of $\tilde{\delta}$ and are then renormalized, or none
receives an 
increment). Hence, agents' limiting beliefs are 
always weighted averages of their initial beliefs, where the weights are
$\frac{1}{n}$. Applying the law of large numbers then implies that
agents become $\epsilon$-wise as $n\goesto\infty$ for any $\epsilon>0$
almost surely, for all $k>r$. 

For the more general case when not all agents are $\eta$-wise for
topic $X_r$, one can show that agents' limiting beliefs for topic
$X_{r+1}$ are (uniform) averages of the initial beliefs of the agents
who were $\eta$-wise for $X_r$, rather than averages of all
agents' initial beliefs. As topics progress, either all agents are
jointly $\eta$-wise or they are not (since agents always reach a
consensus for topics $X_k$, with $k>r$). Hence, since agents adjust
weights 
based on limiting beliefs, the entries in the weight matrices
$\mathbf{W}^{(k)}$ all either receive jointly an increment of $\tilde{\delta}$
or not (in fact, increments of $\tilde{\delta}$ are added infinitely
often, almost surely, as $k\goesto\infty$ since initial beliefs are
centered around 
truth). Hence, $\mathbf{W}^{(k)}$ tends toward a matrix with all 
entries $\frac{1}{n}$ as $k\goesto\infty$ and the law of large numbers
takes care for almost sure convergence. 
\end{proof}
Next, we state that Proposition \eqref{prop:centered1} holds true also
if agents adjust weights based on initial beliefs. This is
understandable: if agents adjust weights based on limiting beliefs, 
weights converge to $\frac{1}{n}$ as $k$ increases. However, this
weighting structure 
is not optimal, as it ignores the different variances of agents'
initial beliefs, but agents' final beliefs still converge to truth in the
limit. Hence, if agents set weights `closer to optimality' as they do
when they adjust based on initial beliefs 
(cf.\ Section \ref{sec:justification}), they
should certainly also converge to truth. We prove the proposition more
formally by referring, in Appendix \ref{sec:appendix}, to results
developed in Golub 
and Jackson (2010) \cite{Golub2010}, which generalize the `ordinary'
law of large numbers. 
\begin{proposition}\label{prop:centered2}
  Let $\eta\ge 0$ be fixed. 
  Assume that agents' initial beliefs are centered around truth in the
  form \eqref{eq:centered}. Moreover, assume that agents initially
  follow their own beliefs, that is, $\mathbf{W}^{(1)}$ is the
  $n\times n$ identity matrix $\mathbf{I}_n$. Finally, assume that
  agents adjust 
  weights based on initial beliefs, i.e., $\tau=0$. Let
  $T(\cdot)>0$. 
  Then, as $k,n\goesto\infty$, all agents become $\epsilon$-wise for
  topics $X_k$, for all $\epsilon>0$, almost surely. 
\end{proposition}

\subsection{Biased agents}\label{sec:biased}

\subsubsection*{The case $T(n)=0$} 
In the biased agent setup, we start with the following conditions. 
Fix a level of wisdom $\epsilon>0$, with $\epsilon\le \eta$, agents'
tolerance. Let there be $n=n_1+n_2$ agents,  
and 
denote by $\struct{N}_1$ and
$\struct{N}_2$ the respective agent sets such that
$[n]=\struct{N}_1\cup\struct{N}_2$. The agents in $\struct{N}_1$ 
are $\epsilon$-intelligent and we think of them as having unbiased initial
beliefs about any topic $X_k$; in particular,
we think of their initial 
beliefs as distributed according to $\mu_k+\epsilon_{ik}$, where
$\epsilon_{ik}$ is white noise, appropriately restricted such that
$\mu_k+\epsilon_{ik}\in B_{k,\epsilon}$. 
Conversely, let the $n_2$ agents in 
$\struct{N}_2$ have initial beliefs distributed according to a random 
variable $Z_k$ (that depends on topic $X_k$) with 
distribution function $F_{Z_k}(A)=P[b_i^k(0)\in A]$, for $A\subseteq S$
(in particular, agents in $\struct{N}_2$ all have the same
distribution of initial beliefs).
Assume that $F_{Z_k}(A)>0$ for all non-empty intervals $A\subseteq S$.
We think of the agents in $\struct{N}_2$ as biased in that it holds
that $\beta_k=\norm{\Exp[Z_k]-\mu_k}>0$ for all topics $X_k$. 
Finally, assume 
that $T(m)>0$ for all $m<n$ and $T(n)=0$ and 
let 
$\mathbf{W}^{(1)}$ be the $n\times n$ identity matrix. 
For short, we will also refer to the $n_2$ agents in $\struct{N}_2$ as
`biased' agents. 

Our first result, concerning weight adjustment at $\tau=\infty$, states
that agents' limiting beliefs, in expectation,  
in this context will be a mixture of truth $\mu_k$ and $\Exp[Z_k]$
unless no biased agent `guesses' truth for topic $X_1$, the 
first topic to be discussed, in which case all agents reach level of
wisdom $\epsilon$ for all topics $X_k$. In other words, 
if a biased agent is true for the initial topic $X_1$, 
then agents will
always mix truth with a biased variable. 
That agents do not mix 
when no biased agent is true for $X_1$ 
crucially depends on the condition $T(n)=0$. Namely,
if no biased agent is close enough to truth for topic $X_1$, only the
$\epsilon$-intelligent agents will be, such that, for topic
$X_2$, agents only increment 
weights to agents in $\struct{N}_1$; consequently, as we show, for
topic $X_2$, limiting consensus beliefs will be uniform means of these
agents' beliefs so
that all agents are $\epsilon$-wise for $X_2$; but, since $T(n)=0$,
no more weight adjustments occur whatsoever, so that all agents are
$\epsilon$-wise for all topics $X_k$ to come. We also remark that if
agents' limiting beliefs are mixtures of truth and a biased variable,
this does not mean that agents would not be $\epsilon$-wise for a
certain topic (which depends both on the biased agents' bias and on
$\epsilon$); it 
solely means that agents mix truth with something that 
distracts them away from truth. 

For the proof of the result, we 
make use of 
the insight that if
someone is wise (or intelligent) at a more refined level, he is also
wise (or intelligent) at a coarser level; the following lemma, which
restates this, is
self-explanatory and needs no proof. 
\begin{lemma}\label{lemma:trivial}
  Let $0\le \epsilon_1\le \epsilon_2$. If an agent $i$ is
  $\epsilon_1$-wise ($\epsilon_1$-intelligent) for some topic $X_k$,
  then she is also $\epsilon_2$-wise ($\epsilon_2$-intelligent) for
  $X_k$. 
\end{lemma}
In the following proposition, $\epsilon_1$ will be $\epsilon$, the
level of wisdom to be obtained, and $\epsilon_2$ will be $\eta$,
agents' tolerance. 
\begin{proposition}\label{prop:inf}
  Let the weight adjustment time point be $\tau=\infty$. Let tolerance
  $\eta\ge 0$ be fixed and fix a level $\epsilon\ge 0$ of wisdom, with
  $\epsilon\le \eta$.

  Under the outlined conditions, if $N_\eta(\mathbf{b}^1(\tau),\mu_k)$
  contains only unbiased $\epsilon$-intelligent agents --- that is, 
  $N_\eta(\mathbf{b}^1(\tau),\mu_k)=\struct{N}_1$ ---  
  then all agents become
  $\epsilon$-wise for all topics $X_k$, with $k>1$. 
  If $N_\eta(\mathbf{b}^1(\tau),\mu_k)$ contains also agents from the set
  $\struct{N}_2$,\footnote{But not all of them. If
    $N_\eta(\mathbf{b}^1(\tau),\mu_k)=[n]$, then the set
    $N_\eta(\mathbf{b}^1(\tau),\mu_k)$ should be replaced by
    $N_\eta(\mathbf{b}^2(\tau),\mu_k)$, etc.} then agents'
  limiting beliefs, in expectation, are 
  given by 
  $\lambda_Z\Exp[Z_k]+\lambda_\mu\mu_k$, for all topics $k>1$, where
  $\lambda_Z$ and $\lambda_\mu$ are coefficients such that
  $\lambda_\mu=\frac{n_1}{\length{N_\eta(\mathbf{b}^1(\tau),\mu_k)}}$ and
  $\lambda_Z=\frac{\length{N_\eta(\mathbf{b}^1(\tau),\mu_k)\cap\struct{N}_2}}{\length{N_\eta(\mathbf{b}^1(\tau),\mu_k)}}$
  so that $\lambda_\mu+\lambda_Z=1$.   
\end{proposition}
\begin{proof}
  For convenience, we consider the situation when only one agent, $i=1$,
  is $\epsilon$-intelligent. The more general case is a
  straightforward extension of our arguments. We also assume that
  agent $i=1$ holds beliefs $b_i^k(0)=\mu_k$, for all $k\ge 1$. 

  Let $N_\eta(\mathbf{b}^1(\tau),\mu_k)$
  contain only $\epsilon$-intelligent agents. 
  Since $\mathbf{W}^{(1)}$ is
  the identity matrix,  
  the limiting beliefs of agents $1,\ldots,n$ on topic
  $X_1$ are as follows:
  \begin{align*}
    b_1^1(\infty)=\mu_1,\: b_2^1(\infty)=b_2^1(0),\:\ldots,\:
  b_n^1(\infty)=b_n^1(0).         
  \end{align*}
  Moreover, since initial beliefs of the agents in $\struct{N}_2$ are in
  $B_{k,\eta}^c$, 
  the agents in $\struct{N}_2$ are, consequently, also
  not $\eta$-wise for topic 
  $X_1$, in contrast to the $\epsilon$-intelligent agent, who is
  $\eta$-wise for topic $X_1$. 
  Thus, the weight structure at the
  beginnning of discussion of topic 
  $X_2$ looks as follows, after weight adjustment and renormalization
  \begin{align*}
  \mathbf{W}^{(2)} =
  \frac{1}{1+\tilde{\delta}}
  \begin{pmatrix}
  1 & 0 & 0 & \cdots & 0\\
  {\tilde{\delta}} & 1 & 0 & \cdots &
  0\\
  \vdots & \cdots & \ddots & \\
   {\tilde{\delta}} & 0 & 0 & \cdots &
  {1}\\
  \end{pmatrix};
  \end{align*}
  recall our convention that $\tilde{\delta}=\delta\cdot T(\cdot)$. 
  Limiting beliefs for
  topic $X_2$ are thus given by
  \begin{align*}
  \mathbf{b}^2(\infty) = \lim_{t\goesto\infty}(\mathbf{W}^{(2)})^t\mathbf{b}^2(0),
  \end{align*}
  where the initial belief vector $\mathbf{b}^2(0)$ is
  $(\mu_2,b^2_2(0),\ldots,b^2_n(0))^\intercal$. It is 
  not difficult to see that powers of any matrix with structure $\begin{pmatrix}
  1 & 0 & 0 & \cdots & 0\\
  \alpha & 1-\alpha & 0 & \cdots &
  0\\
  \vdots & \cdots & \ddots & \\
   \alpha & 0 & 0 & \cdots &
  1-\alpha\\
  \end{pmatrix}  $ have the form
  \begin{align*}
  \begin{pmatrix}
  1 & 0 & 0 & \cdots & 0\\
  \alpha & 1-\alpha & 0 & \cdots &
  0\\
  \vdots & \cdots & \ddots & \\
   \alpha & 0 & 0 & \cdots &
  1-\alpha\\
  \end{pmatrix}^t = \begin{pmatrix}
  1 & 0 & 0 & \cdots & 0\\
  \alpha\sum_{i=1}^{t-1}(1-\alpha)^i & (1-\alpha)^t & 0 & \cdots &
  0\\
  \vdots & \cdots & \ddots & \\
   \alpha\sum_{i=1}^{t-1}(1-\alpha)^i & 0 & 0 & \cdots &
  (1-\alpha)^t\\
  \end{pmatrix}.
  \end{align*}
  For $0<\alpha\le 1$, the right-hand side of the last equation
  obviously converges to the matrix with all entries identical to
  zero, except for the first column, which consists of $n$ entries
  $1$. Hence, by this fact, $\mathbf{b}^2(\infty)$ is the vector with
  all entries $\mu_2$ and all agents are, consequently,
  $\epsilon$-wise for topic $X_2$, and, thus, also $\eta$-wise (by
  Lemma \ref{lemma:trivial}).  Since 
  in this case, it holds that $\abs{N_\eta(\mathbf{b}^2(\infty),\mu_2)}=n$, 
  we have $T(\abs{N_\eta(\mathbf{b}^2(\infty),\mu_2)})=0$ by assumption, so
  that agents do not adjust weights for topic $X_3$ (more precisely,
  the adjustment increment is zero). Hence,
  $\mathbf{W}^{(3)}=\mathbf{W}^{(2)}$, and agents will also be
  $\epsilon$-wise for topic $X_3$ since agent $i=1$ is
  $\epsilon$-intelligent for $X_3$. Inductively, this holds for all
  $X_k$, with $k>1$.  

  Now, 
  assume that at least one 
  agent in $\struct{N}_2$ 
  happens to know truth for topic $X_1$ (that is, his initial belief
  is within an $\eta$ 
  radius of truth), which may always occur 
  since $F_{Z_k}(A)>0$ for all intervals $A\subseteq S$ by assumption. 
  For convenience, we 
  assume that exactly one 
  agent in $\struct{N}_2$, say, agent $2$, happens to
  know truth for topic $X_1$. 
  Then, at the beginning of the
  discussion of topic $X_2$, agents increase their weights for agents
  $1$ and $2$, resulting in
  the following structure:
  \begin{align*}
  \mathbf{W}^{(2)}=
  \frac{1}{1+2\tilde{\delta}}
  \begin{pmatrix}
  {1+\tilde{\delta}} & \tilde{\delta} & 0 &
  0 & \cdots & 0\\ 
  \tilde{\delta} & {1+\tilde{\delta}}
  & 0 & 0 & \cdots & 0 \\
  {\tilde{\delta}} & {\tilde{\delta}}
  & {1} & 0 & \cdots & 0 \\
  \vdots & \vdots & \cdots & \ddots \\
  \tilde{\delta} & \tilde{\delta} & 0 & 0 & \cdots & 1
  \end{pmatrix}
  \end{align*} 
  Again, limiting beliefs for topic $X_2$ are then given by 
  \begin{align*}
    \mathbf{b}^2(\infty)
  = \lim_{t\goesto\infty}\bigl(\mathbf{W}^{(2)}\bigr)^t\mathbf{b}^2(0).
  \end{align*}
  It is not difficult to see that powers of matrices with structures
  as in the given $\mathbf{W}^{(2)}$ converge to the matrix with the first two
  columns being 
  $\frac{1}{2}\one_n$ 
  and the remaining
  columns are zero vectors. Thus, limiting beliefs of all agents are
  just the average of the first two agents' initial beliefs. This
  implies a limiting consensus such that all
  agents are either jointly $\eta$-wise or not $\eta$-wise for topic
  $X_2$. If all are $\eta$-wise, 
  no weight adjustments occur for topic $X_3$ (since $T(n)=0$), but if
  they are not $\eta$-wise, 
  no weight adjustments occur as well (no one was right). Thus, 
  as before,
  $\mathbf{W}^{(2)}=\mathbf{W}^{(3)}=\mathbf{W}^{(4)}=\cdots$, 
  such that for all topics
  to come, limiting beliefs of all agents will always be averages of
  the first agent's (who is $\epsilon$-intelligent) and the second
  agent's (who was just 
  lucky for topic $X_1$) initial beliefs. 
  Hence, in expectation, agents' limiting (consensus) beliefs will be 
  \begin{align*}
    \frac{1}{2}\Exp[Z_k]+\frac{1}{2}\mu_k.
  \end{align*}
  The more general forms of $\lambda_Z$ and $\lambda_\mu$ can be
  straightforwardly derived in an analogous manner in the more general
  setting.  
\end{proof}
\begin{example}
  We illustrate Proposition \ref{prop:inf} in Figure
  \ref{fig:inf_illustrate}, where we let $S=[0,1]$, $\mu_k=0$ for all
  $k\ge 1$,
  $[n]=\set{1,\ldots,50}$, $\struct{N}_1=\set{1}$ and $F_{Z_k}$ is the
  random uniform distribution on $S$, $\epsilon=0$ and $\eta=0.2$. 
  \begin{figure}
    \centering
    \includegraphics[scale=0.25]{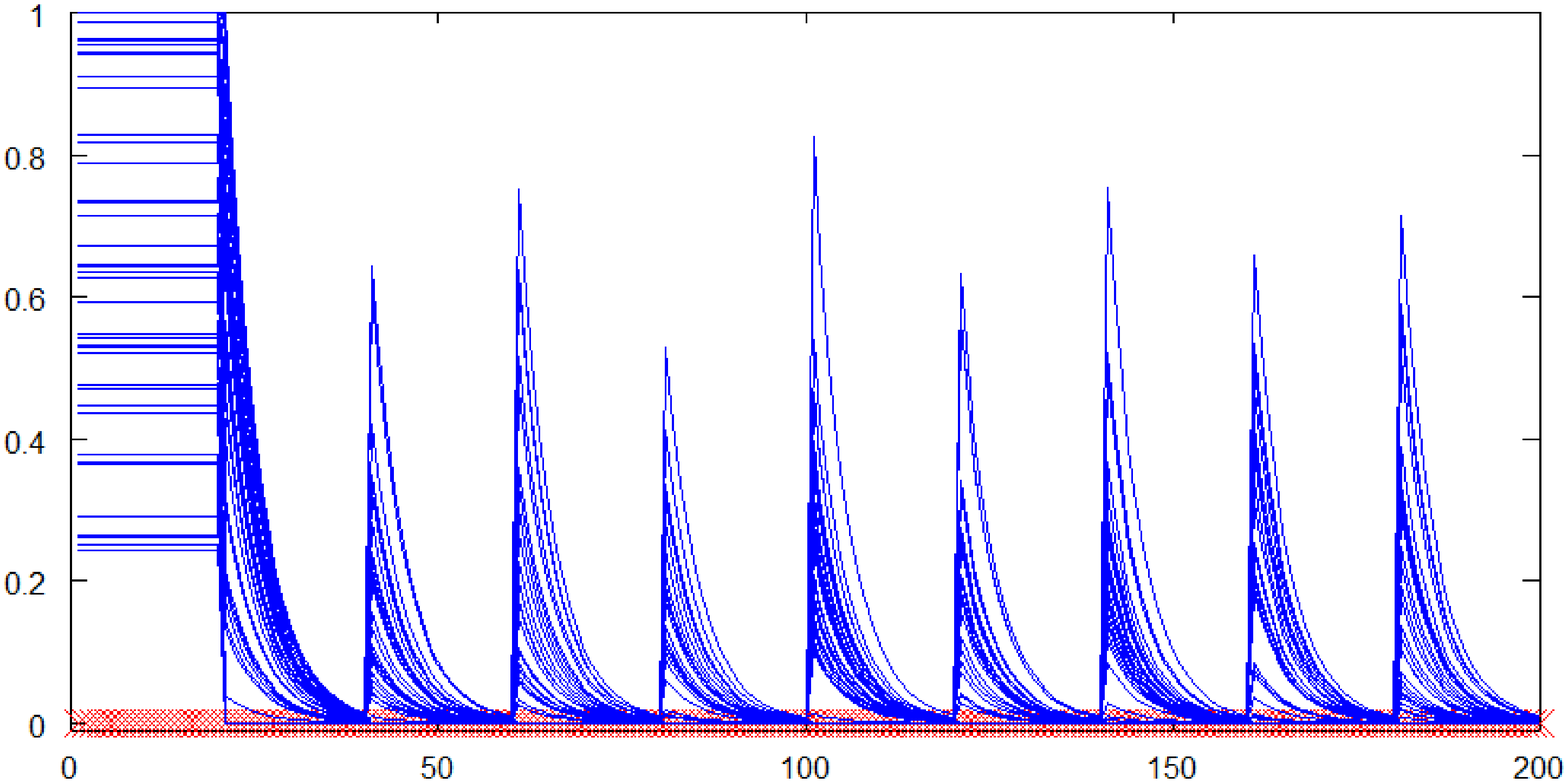}
    \includegraphics[scale=0.25]{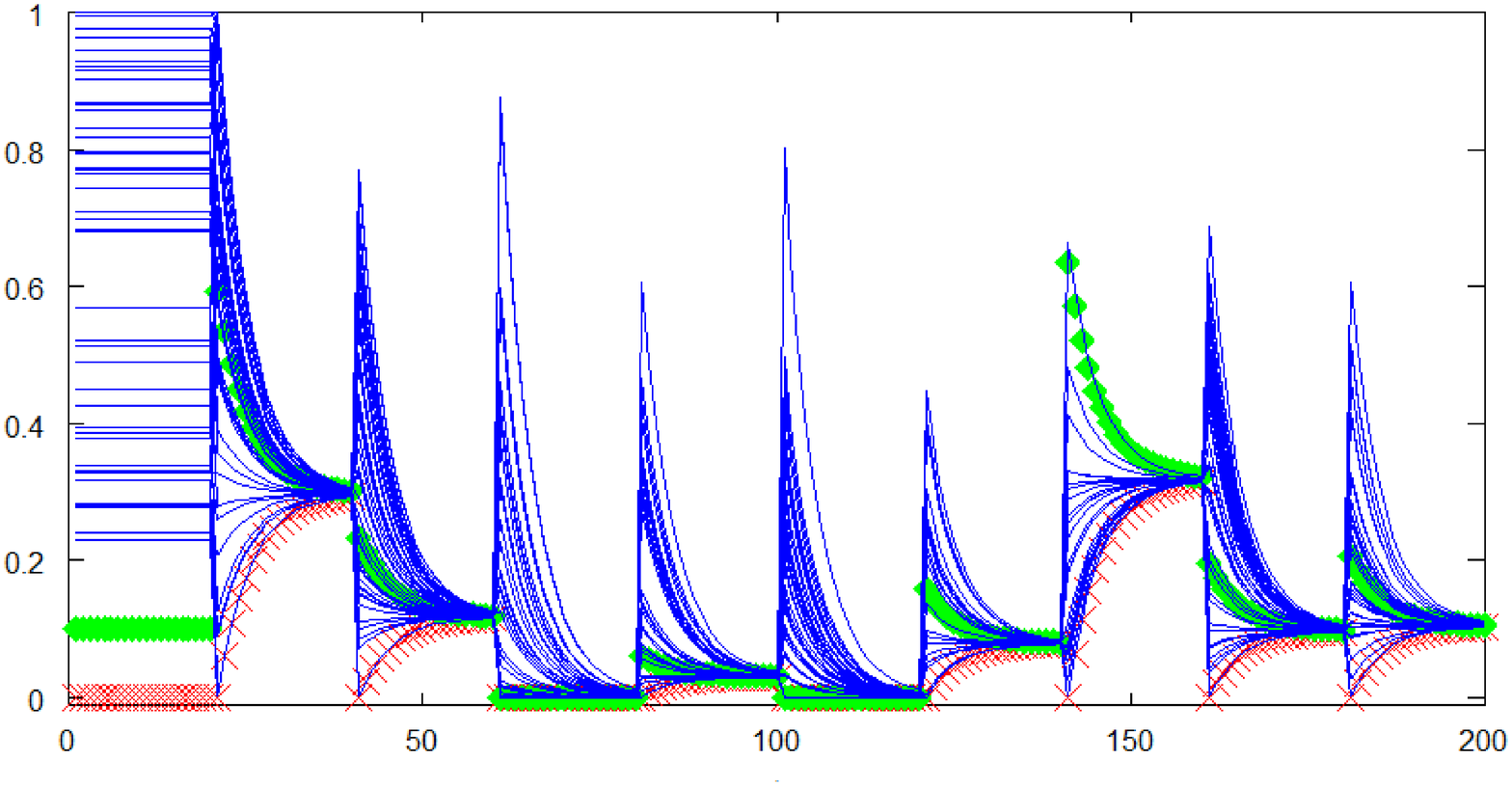}
    \caption{
      Description in text; also $\tau=\infty$, $\delta=0.2$ and
      $T(m)=1$ for $m<n$ 
      and $T(n)=0$. For each topic, we plot discussion rounds
      $t=0,1,2,\ldots,20$. 
      Left: $N_\eta(\mathbf{b}^1(\infty),\mu_1)=\set{1}=\struct{N}_1$
      ($\epsilon$-intelligent agent in red)
      such that all agents are $\epsilon$-wise for all topics $X_k$,
      with $k>1$. Right:
      $N_\eta(\mathbf{b}^1(\infty),\mu_1)=\set{1,2}$ contains also
      one biased agent (in green) such that limiting beliefs of agents are
      mixtures of truth and $Z_k$. 
    }
    \label{fig:inf_illustrate}
  \end{figure}
\end{example}
\begin{remark}
  As an application of Proposition \ref{prop:inf}, consider the
  situation when the number $n$ of agents goes to infinity. Then, if
  the fraction $\frac{n_2}{n}$ of agents in $\struct{N}_2$ converges
  to zero, agents become $\epsilon$-wise, in the limit, as
  $n\goesto\infty$. Namely, first, the coefficient $\lambda_Z$
  converges to zero in this case since $\lambda_Z\le
  \frac{n_2}{n_1}=\frac{n_2}{n-n_2}$, so that agents' expected
  consensus is indeed $\mu_k$ as $n\goesto\infty$. Moreover, not only
  do agents' limiting beliefs converge to $\mu_k$ \emph{in expectation}, but
  agents become indeed $\epsilon$-wise in the limit, as the support of
  the distribution of the $\epsilon$-intelligent agents' initial
  beliefs is 
  $B_{k,\epsilon}$. 

  As examples of $\frac{n_2}{n}$ converging to zero, of course, if
  the number $n_2$ remains constant as $n\goesto 0$, then
  $\frac{n_2}{n}$ goes to zero. But even if, for example, $n_2$ grows
  as in $\sqrt{n}$, all agents finally become $\epsilon$-wise. 
\end{remark}
Proposition \ref{prop:inf} may be restated in the following way;
agents' limiting beliefs, in expectation, are given by
$\lambda_Z\Exp[{Z_k}]+\lambda_\mu\mu_k$, where $\lambda_Z=0$ if
$N_\eta(\mathbf{b}^1(0),\mu_k)=\struct{N}_1$. 
We can then determine the probability that $\lambda_Z=0$. 
\begin{corollary}
  Under the conditions of Proposition \ref{prop:inf}, 
  with probability 
  exactly $F_{Z_k}(B_{k,\eta}^c)^{n_2}>0$, we have 
  $\lambda_Z=0$. 
\end{corollary}
\begin{proof}
  The event that the $n_2$ biased agents' initial beliefs
  $b_i^1(0)$ are in $B_{k,\eta}^c$ is, by the iid property,
  $F_{Z_k}(B_{k,\eta}^c)^{n_2}$. 
\end{proof}
\begin{remark}
  According to the corollary, the probability that $\lambda_Z=0$ is
  strictly positive but decreasing as the number of 
  biased agents 
  increases. Hence, as $n_2$ becomes large, agents' limiting beliefs are very
  likely mixtures of truth $\mu_k$ and $\Exp[Z_k]$, a 
  value that is different from truth. 
\end{remark}
\begin{remark}
  We may consider the setup of Proposition \ref{prop:inf} as a `type
  inference' problem. What the proposition says and shows is that,
  since agents adjust their weights based on 
  limiting beliefs, they cannot infer the intelligent agents once a
  biased non-intelligent agent has guessed truth because agents always reach
  a consensus in our situation (cf.\ also Proposition
  \ref{prop:consensus}). Thus, the intelligent agents cannot properly  
  signal their type in this case because all agents' limiting beliefs
  are indistinguishable. 
\end{remark}
Now, consider the exact same situation as in Proposition
\ref{prop:inf}, except that agents adjust weights based on initial
beliefs, i.e., $\tau=0$. In this
situation, a sufficient condition for wisdom is that agents find
truth sufficiently valuable, i.e., $\delta$ is sufficiently
large. In this case, wisdom, in the limit as $k\goesto\infty$, obtains
almost surely, 
namely, all 
that is required is that only the $\epsilon$-intelligent agents in
$\struct{N}_1$ are 
initially true for some topic $X_k$.

\begin{proposition}\label{prop:0}
  Let the weight adjustment time point be $\tau=0$.

  Under the conditions as in Proposition \ref{prop:inf},  if $\delta$
  is sufficiently large, 
  then, almost surely,  
  there exists a (time point) $M\in\nn$ such that all agents are
  $\epsilon$-wise 
  for all topics $X_k$, with $k>M$. 
\end{proposition}
\begin{proof}
  Let $M$ be the first time point that (1) only $\epsilon$-intelligent
  agents in $\struct{N}_1$ happen to know truth, initially, for topic
  $X_M$, that is, 
  $b_i^M(0)\in B_{k,\eta}$ for all $i\in \struct{N}_1$ and no
  $i\in \struct{N}_2$; and, (2) not all agents are $\eta$-wise for
  $X_M$ (such that $T(\cdot)>0$). 
  Then, weight adjustment at $M+1$ will add $\tilde{\delta}>0$ to
  the weights of the $\epsilon$-intelligent agents in
  $\struct{N}_1$. If $\tilde{\delta}$ is 
  sufficiently large, after normalization, weights for the
  non-intelligent agents become arbitrarily small and (arbitrarily
  close to) uniform for the $\epsilon$-intelligent agents. 
  In particular, $\delta$ may be so large that all agents' beliefs
  $b_i^{M+1}(1)$ lie in $B_{k,\epsilon}$. 
  Since this is a convex set 
  and weight matrices are row-stochastic, beliefs will remain in
  $B_{k,\epsilon}$ for all time periods $t$; hence, agents will be
  $\epsilon$-wise in the limit for topic $X_{M+1}$, and, consequently,
  also $\eta$-wise. Since $T(n)=0$, no
  more adjustments will occur after time point $M+1$ and all agents
  become 
  $\epsilon$-wise for all topics $X_k$, with $k>M$, since their
  weights are now (sufficiently close to) uniform for the
  $\epsilon$-intelligent agents in $\struct{N}_1$. 
\end{proof}
\begin{example}
  We illustrate Proposition \ref{prop:0} in Figure
  \ref{fig:0_illustrate}, where we let $S=[0,1]$, $\mu_k=0$ for all
  $k\ge 1$,
  $[n]=\set{1,\ldots,50}$, $\struct{N}_1=\set{1}$ and $F_{Z_k}$ is the
  random uniform distribution on $S$, $\epsilon=0$ and $\eta=0.05$. 
  \begin{figure}
    \centering
    \includegraphics[scale=0.25]{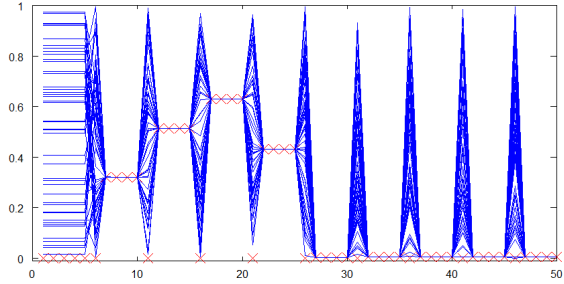}
    \caption{
      Description in text; also $\tau=0$, $\delta=100$ and $T(m)=1$ for $m<n$
      and $T(n)=0$. For each topic, we plot discussion rounds
      $t=0,1,2,\ldots,5$. 
      For topic $M=5$, we have
      $N_\eta(\mathbf{b}^M(0),\mu_M)=\struct{N}_1$ such that all
      agents are $\epsilon$-wise for all topics $X_k$, with $k>M$. 
    }
    \label{fig:0_illustrate}
  \end{figure}
\end{example}
\begin{remark}
  To summarize, 
  the intelligent agents in $\struct{N}_1$ can now correctly signal
  their type. All that 
  is required is that only $\epsilon$-intelligent agents in
  $\struct{N}_1$ happen to 
  know truth for some topic, in which case they will receive such a
  large weight increment that they lead society to $\epsilon$-wisdom;
  then, no more weight adjustments occur because the `right guys' have
  been identified. 
\end{remark}
\begin{remark}
  In our current setup, the difference between weight adjustment at
  $\tau=0$ vs.\ at $\tau=\infty$ is as follows. While adjusting at
  $\tau=0$ leads agents to $\epsilon$-wisdom almost surely provided that
  they find truth sufficiently valuable, that is, $\delta$ is
  large enough; updating at $\tau=\infty$ leads agents to
  $\epsilon$-wisdom provided that biased agents do not know (or,
  perhaps, 
  `guess') truth for topic $X_1$. The latter condition is difficult to
  satisfy if we assume that the number of biased agents becomes
  large, while the condition of sufficiently large $\delta$ 
  also
  depends on population size $n$ and, in particular, on $n_2$, the
  population 
  size of the biased agents. In other words, if $T(n)=0$, we
  can 
  specify sufficient conditions for wisdom even under the presence of
  biased agents, but these are rather
  challenging. 
\end{remark}

\subsubsection*{The case $T(\cdot)>0$}
Now, we consider the same setup as in the last subsection, except that
we assume that $T(\cdot)>0$ on its whole domain. In this case, agents
continuously adjust their weights to other agents, which is also the
rational behavior of an agent who assumes the conditions outlined in
Section \ref{sec:justification}; recall our previous discussion. 

We consider a slightly more general situation here than in the last
subsection in that we allow each agent to have initial beliefs
distributed according to individualized distribution functions, rather
than to assume groups with identical distribution functions; the more
restrictive setting is then a special case of our generalization. 
Accordingly, assume that agent $i$'s initial belief for topic $X_k$
is distributed 
according to 
random variable $Z_{i,k}$ with distribution function 
$F_{i,k}(A)=\text{Pr}[Z_{i,k}\in A]$ for all $A\subseteq S$ and all
topics $X_k$, for $k\in\nn$. 
We assume that $F_{i,k}(B_{k,\eta})$, which gives the probability that
agent $i$ is within an $\eta$-radius around truth $\mu_k$, does not
depend on topic $X_k$, that is, 
$F_{i,k}(B_{k,\eta})=F_{i,k'}(B_{\eta,k'})$ for all $k,k'$, which
means that the probability that agent $i$ is 
truthful
is the same across topics. 
We then have 
the following proposition.
\begin{proposition}\label{prop:general1}
  Let tolerance $\eta\ge 0$ be fixed. 
  Assume that agents adjust weights based on initial beliefs, i.e.,
  $\tau=0$, and 
  assume that $T(\cdot)>0$. Then, 
  as $k\goesto\infty$, agents' limiting consensus beliefs on issue $X_k$ are
  distributed according to
  \begin{align*}
    b_i^k(\infty) \sim \sum_{j=1}^n \lambda_jZ_{j,k}
  \end{align*}
  where
  \begin{align*}
    \lambda_j \propto F_{j,k}(B_{k,\eta})
  \end{align*}
  with $\sum_{j=1}^n \lambda_j=1$ (note that $\lambda_j$ does not
  depend upon $k$ by assumption). In particular, we have
  \begin{align*}
    \Exp[b_i^k(\infty)] = \sum_{j=1}^n \lambda_j\Exp[Z_{j,k}]. 
  \end{align*}
\end{proposition}
\begin{proof} 
  Our proof is not rigorous. 

  Since agents are 
  homogenous with respect to tolerance $\eta$, they will all
  \emph{jointly} increase their weight to a 
  particular agent $j$ (or they will \emph{jointly} not do
  so). Therefore, as $k$ increases, rows of $\mathbf{W}^{(k)}$ become more
  and more similar, independent of the initial conditions
  $\mathbf{W}^{(1)}$ (if weight matrix $\mathbf{W}^{(1)}$ is identical
  in each row, this will propagate to any $\mathbf{W}^{(k)}$ with
  $k>1$, but even if not, rows will become more and more similar by
  the homogeneity of agents). The
  weight mass that any particular agent $i$ assigns to any particular
  agent $j$ is clearly proportional to $F_{j,k}(B_{k,\eta})$ (cf.\ Figure
  \ref{fig:intro}) since this value indicates how frequently agent $j$ is
  truthful. Hence, since rows of $\mathbf{W}^{(k)}$ are
  (approximately) identical, as $k$ becomes large, with each entry
  $[\mathbf{W}^{(k)}]_{ij}$ 
  being proportional to $F_{j,k}(B_{k,\eta})$, limiting beliefs of agents
  are given by,  
  \begin{align*}
    b_i^k(\infty) \approxeq b_i^k(1)=\sum_{j=1}^n\lambda_j b_j^k(0),
  \end{align*}
  where $\lambda_j\propto  F_{j,k}(B_{k,\eta})$. 
  This completes the proof. 
\end{proof}
\begin{remark}
  The coefficients $\lambda_j$ have a very intuitive
  interpretation. Since they indicate how limiting consensus beliefs
  are formed in terms of initial beliefs, their standard
  interpretation is that of \emph{social influence weights} (cf.,
  e.g., 
  Golub and Jackson
  (2010) \cite{Golub2010}). Clearly, in our 
  endogenous weight formation model, with weight sizes dependent upon
  `past performance', an agent's social influence is intuitively given by his
  likelihood of correctly predicting truth. 
\end{remark}
\begin{example}
  Considering the distribution of limiting consensus beliefs, we note that if
  two $Z_{j,k}$, for $j=x,y$, are normally distributed with parameters
  $(\mu_x^k,\sigma_{x,k}^2)$ and $(\mu_y^k,\sigma_{y,k}^2)$, then both
  $\lambda_j Z_{j,k}$ as well as $\sum_j \lambda_jZ_{j,k}$ are normally
  distributed; the latter sum has 
  normal distribution with parameters
  $(\lambda_x\mu_x^k+\lambda_y\mu_y^k,\sigma_{x,k}^2+\sigma_{y,k}^2)$. Hence, if
  all agents' initial beliefs are normally distributed, their limiting
  beliefs are also normally distributed. 

  Moreover, if there are several `types' or `groups' of agents,
  $\mathcal{N}_1,\ldots,\mathcal{N}_m$, of sizes $n_1,\ldots,n_m$,
  where each group has 
  identical and independent initial distribution (within groups), 
  then agents in each group receive about the same weight mass, which
  is proportional to (see example below)
  $\lambda_{\mathcal{N}_l}\frac{1}{n_l}$, for $l\in\set{1,\ldots,m}$,
  so that 
  if sizes $n_1,\ldots,n_m$ of groups become large, then, by the 
  central limit theorem, $\sum_{j\in \mathcal{N}_l}\lambda_{\mathcal{N}_l}\frac{1}{n_l}Z_{j,k}=$ is
  approximately normally distributed.
  Thus, by our above remark,
  $\sum_{j\in[n]}\lambda_jZ_{j,k}$ is also approximately normally
  distributed. In other words, we would generally expect agents'
  limiting beliefs to be normally distributed, in this setup. 
%
%
\end{example}
\begin{example}\label{example:3groups}
  Consider three groups of agents,
  $\mathcal{N}_1,\mathcal{N}_2,\mathcal{N}_3\subseteq[n]$ with 
  $\mathcal{N}_1\cup \mathcal{N}_2\cup \mathcal{N}_3=[n]$ and where the $\mathcal{N}_l$'s are pairwise
  mutually disjoint. The
  first group, 
  which we call experts, 
  has initial beliefs distributed according to
  $\mathtt{N}(\mu_k,\sigma_1^2)$, where $\sigma_1^2>0$ is fixed (that
  is, each member in $\mathcal{N}_1$ has the given distribution function, and we
  assume members' initial beliefs to be independent). The
  second and third groups are biased. Assume, for illustration, that group
  two has distribution $\mathtt{N}(\mu_k-a,\sigma_2^2)$ and group
  three has $\mathtt{N}(\mu_k+b,\sigma_3^2)$. Assume the groups have
  sizes $n_1=\frac{1}{5}n$, and $n_2=n_3=\frac{2}{5}n$, that is, the
  group of experts is smallest in size (but still growing in
  $n$). Moreover, let, for instance,    
  $a=3$, $b=1$ and $\sigma_1^2=\sigma_2^2=\sigma_3^2=1$, and let
  $\eta=0.25$. Then, each 
  expert has $\lambda_j$ of about $\lambda_j\propto{0.19741}$,
  members of group two have $\lambda_j\propto {0.0024}$ and
  members of group three have $\lambda_j\propto
  {0.0278}$. Since the $\lambda$'s must sum to one, we have
  about $\lambda_{\mathcal{N}_1}\approxeq\frac{0.19751n_1}{C_0}$ for experts, and
  $\lambda_{\mathcal{N}_2}\approxeq\frac{0.0024n_2}{C_0}$ and
  $\lambda_{\mathcal{N}_3}\approxeq\frac{0.0278n_3}{C_0}$ for groups two and
  three, respectively, and where
  $C_0=\lambda_{\mathcal{N}_1}+\lambda_{\mathcal{N}_2}+\lambda_{\mathcal{N}_3}$ and
  $\lambda_{\mathcal{N}_l}=\sum_{j\in \mathcal{N}_l}\lambda_j$ for $l=1,2,3$. For $n=100$, this
  is about $\lambda_{\mathcal{N}_1}\approxeq 0.44$, $\lambda_{\mathcal{N}_2}\approxeq
  0.01$, and $\lambda_{\mathcal{N}_3}\approxeq 0.55$, which is also,
  approximately, the limiting structure of the distribution of
  $\lambda$ as $n\goesto\infty$. Hence, in the limit, as
  $n\goesto\infty$, these agents beliefs' would converge to a consensus that is
  off by about $\lambda_{\mathcal{N}_1}\cdot 0+\lambda_{\mathcal{N}_2}\cdot
  (-3)+\lambda_{\mathcal{N}_3}\cdot 1=0.51227$ from truths $\mu_k$. More
  precisely, the agents' limiting consensus values are distributed
  according to a 
  normal distribution with mean $\mu_k+0.51227$ and variance that
  converges to zero in $n$; in particular, variance
  of limiting consensus values is given by
  $\frac{\sigma_1^2}{n_1}+\frac{\sigma_2^2}{n_2}+\frac{\sigma_3^2}{n_3}$,
  which is $\frac{1}{10}$ for our example. We plot the (predicted and
  theoretical) limiting distribution of $b_i^k(\infty)$ 
  and a sample
  histogram from an actual simulation 
  in Figure \ref{fig:sim}.  
  \begin{figure*}[!ht]
    \centering
    \begin{subfigure}{0.49\textwidth}
    \resizebox{1.0\textwidth}{!}{
        \input{plots/example3groups.tex}}
    \end{subfigure}
    \begin{subfigure}{0.49\textwidth}
      \includegraphics[scale=0.425]{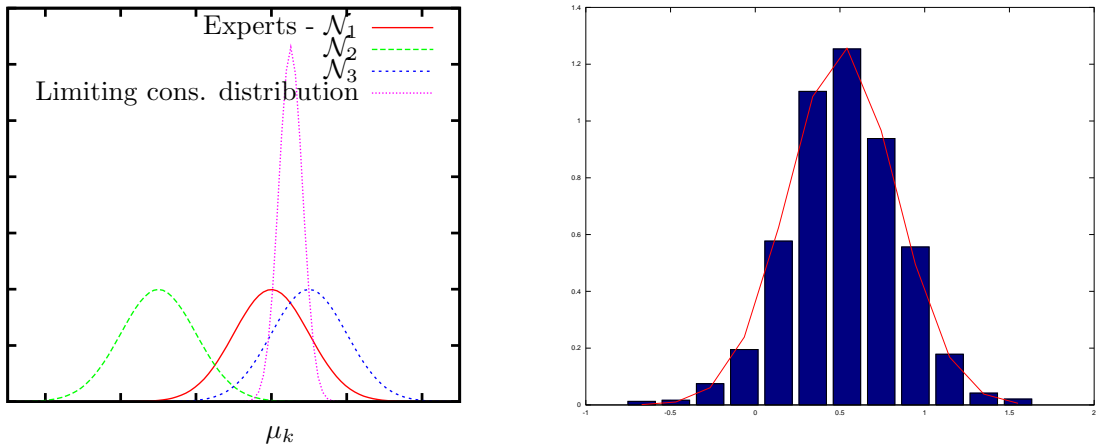}
    \end{subfigure}
    \caption
{ 
  Left: The distribution function of three groups of agents as
  discussed in Example \ref{example:3groups}; experts'
  initial beliefs are always centered around truth, for all topics
  $X_k$, while there are two biased groups, one which underestimates
  truth and one which overestimates truth. If the relative sizes of
  groups are as described in the text, agents distribution of limiting
  beliefs, as $k$ becomes large, is given by the high-peak normal
  distribution indicated, whose mean is off from truth by about
  $0.5$. 
  Right: Sample distribution from a simulation vs.\ predicted
  distribution according to Proposition \ref{prop:general1}.  
}
\label{fig:sim}
\end{figure*}
\end{example}
In the next proposition, we discuss weight adjustment based on
limiting beliefs. We assume that $F_{i,k}(A)>0$ for all non-empty
intervals 
$A\subseteq S$. 
\begin{proposition}\label{prop:general2}
  Let tolerance $\eta\ge 0$ be fixed. 
    Assume that agents adjust weights based on limiting beliefs, i.e.,
    $\tau=\infty$, and 
  assume that $T(\cdot)>0$. Then, 
  as $k\goesto\infty$, agents' limiting consensus beliefs on issue $X_k$ are
  distributed according to
  \begin{align*}
    b_i^k(\infty) \sim \sum_{j=1}^n \frac{1}{n}Z_{j,k}
  \end{align*}
  In particular, we have
  \begin{align*}
    \Exp[b_i^k(\infty)] = \sum_{j=1}^n \frac{1}{n}\Exp[Z_{j,k}]. 
  \end{align*}
\end{proposition}
\begin{proof}
  Since agents reach a consensus for topics $X_k$, with $k>r$, and
  agents adjust weights based on limiting beliefs, weight matrix entries for
  all agents 
  converge to $\frac{1}{n}$. Convergence to $\frac{1}{n}$ is assured
  since all agents have $F_{i,k}(B_{k,\eta})>0$ by assumption such
  that the probability that agents' limiting consensus is within an
  $\eta$-interval around truth is at
  least $F_{i,k}(B_{k,\eta})^n>0$, from which it follows that agents
  adjust weights infinitely often (which each time entails an
  increment of $\tilde{\delta}$ and, thus, implies convergence of
  weight matrix entries to $\frac{1}{n}$) with probability $1$.  
\end{proof}
\begin{remark}
  We see here, again, that adjusting based on limiting beliefs is
  `worse' than adjusting based on initial beliefs, since 
  limiting beliefs 
  are formed through social interaction and may thus not indicate the
  inherent `intelligence' of an 
  agent. 

  To quantify the difference by way of illustration, in Example
  \ref{example:3groups}, agents' beliefs would now converge to a
  consensus, as $k\goesto\infty$, that is off from truths by about
  $\frac{n_1}{n}\cdot 0+\frac{n_2}{n}\cdot(-3)+\frac{n_3}{n}\cdot
  1=-\frac{4}{5}$, which is further away than the value of about 
  $0.51$ given in the situation when agents adjust weights based on initial
  beliefs. In particular, agents in group $\mathcal{N}_2$, who are
  very poor at estimating 
  truth, now receive much larger social influence than in the situation
  where 
  agents adjust based on initial beliefs. 

  However, qualitatively, the results do not change (by much): in
  both circumstances, $\tau=0$ and $\tau=\infty$, agents' 
  limiting 
  beliefs, under our endogenous weight adjustment process, are given
  by convex combinations of all agents' initial beliefs, 
  whereby adjusting based on initial beliefs captures, in the
  social influence weights $\lambda_j$, the intelligence of agents while 
  adjusting based on limiting beliefs leads agents to uniform social
  influence 
  weights $\lambda_j$. 
\end{remark}
Now, consider, again, the setup where there are two groups of agents, which we
denote by $\mathcal{N}_1$ and $\mathcal{N}_2$, respectively; the
first groups' initial beliefs are unbiased while the second groups'
initial beliefs are biased, where we assume that agents within each group
have independently and identically distributed initial
beliefs. Assume, furthermore, that 
$F_{i,k}(B_{k,\eta})>0$ for all 
agents $i=1,\ldots,n$. 
\begin{corollary}\label{cor:noepswisdom}
  Let $\tau=0$ or $\tau=\infty$ and 
  let $\eta\ge 0$, the radius within which agents are considered to be
  truthful, be fixed. 
  Then, if the group of biased agents $\mathcal{N}_2$ is `large
  enough' (e.g., relative to $\mathcal{N}_1$), 
  agents will not become $\epsilon$-wise almost surely as
  $n,k\goesto\infty$,  
  for any $\epsilon\in(0,\norm{\mu_k-\Exp[Z_{\mathcal{N}_2,k}]})$, whereby
  $Z_{\mathcal{N}_2,k}$ denotes a random variable that represents the
  distribution of initial beliefs of any agent from group $\mathcal{N}_2$. 
\end{corollary}
\begin{proof}
  By Proposition \ref{prop:general1} and its proof, if $\tau=0$,
  $\lambda_{\mathcal{N}_l}=\sum_{j\in \mathcal{N}_l}\lambda_j\approxeq
  \frac{F_{\mathcal{N}_l,k}(B_{k,\eta})n_l}{C_0}$ as 
  $k\goesto\infty$ (by $F_{\mathcal{N}_l,k}$, we denote the distribution
  function of an agent from group $\mathcal{N}_l$; also note that
  $F_{N_l,k}(B_{k,\eta})$ does not depend on $k$ by assumption), where 
  $l=1,2$ and $C_0=\lambda_{\mathcal{N}_1}+\lambda_{\mathcal{N}_2}$. Thus, if $n_2$ is
  large enough (relative to $n_1$), $\lambda_{\mathcal{N}_2}$ be may
  arbitrarily 
  close to $1$ such that, in expectation, agents' limiting consensus
  belief will 
  be arbitrarily close to $\Exp[Z_{\mathcal{N}_2,k}]$, whereby, by assumption,
  $Z_{\mathcal{N}_2,k}$ is a 
  biased variable. As $n\goesto\infty$, limiting beliefs will converge
  to $\Exp[Z_{\mathcal{N}_2,k}]$ almost surely, in this case, by the
  law of large numbers. 

  If $\tau=\infty$, Proposition \ref{prop:general2} leads to the same
  conclusion. 
\end{proof}
\begin{remark}
  What Corollary \ref{cor:noepswisdom} shows is that agents may not
  become \emph{infinitely} wise under our endogenous weight adjustment
  process if the group of agents with biased initial beliefs becomes large,
  as, in this case, this group's social influence will become
  arbitrarily large. 
  But the corollary shows more: agents may not become $\epsilon$-wise
  for any $\epsilon\in (0,\norm{\mu_k-\Exp[Z_{\mathcal{N}_2,k}]})$, 
  which may be an
  arbitrarily large interval, depending on the bias of the agents in
  $\mathcal{N}_2$.  
  In other words, if the number of
  biased agents is large (relative to the number of
  intelligent agents), the wisdom that society as a whole can attain
  is limited by the latter agents' bias. 
\end{remark}

\subsection{Varying weights on own beliefs}\label{sec:demarzo}
DeMarzo, Vayanos, and Zwiebel (2003) \cite{Demarzo2003} consider a
slight generalization of belief updating
process \eqref{eq:degrootupdate} where agents may place varying
weights on their own beliefs such that \eqref{eq:degrootupdate} reads as 
\begin{align}\label{eq:degrootupdatedemarzo}
  \mathbf{b}^k(t+1)
  = \Bigl((1-\lambda_t)\mathbf{I}_n+\lambda_t\mathbf{W}^{(k)}\Bigr)\mathbf{b}^k(t) 
\end{align}
whereby $0< \lambda_t\le 1$ (note that we treat $\lambda_t$ as an
exogenous variable). Such a weighting scheme may be
empirically plausible, as it has been found (cf., e.g., Mannes
(2009) \cite{Mannes2009}) that 
people often tend to overweight their own beliefs relative to that of
outsiders, probably because individuals have access to their own
motivations for beliefs while they do not have such justification for
others' beliefs. This reasoning would imply that $\lambda_t$ is
`relatively small'. However, as long as weights on others' beliefs do not
drop to zero too quickly, belief updating
rule \eqref{eq:degrootupdatedemarzo} leads to the same limiting
beliefs as the original DeGroot updating
rule \eqref{eq:degrootupdate} where $\lambda_t=1$, for all $t$,
provided that the latter converges; 
convergence may take sufficiently 
longer, however. Hence, under these circumstances, all our previous
results remain valid. The following proposition is a straightforward
generalization of the corresponding theorem, Theorem 1, in DeMarzo,
Vayanos, and Zwiebel (2003) \cite{Demarzo2003}, which restates the
lessons we have just mentioned. 
\begin{proposition}\label{prop:demarzo}
  Assume that $\mathbf{W}^{(k)}$ converges (for all initial belief
  vectors $\mathbf{b}(0)$), then if,
  $\sum_{t=1}^\infty\lambda_t=\infty$, updating
  process \eqref{eq:degrootupdatedemarzo} also converges (for all
  initial belief vectors $\mathbf{b}(0)$) and leads to the same
  limiting beliefs as \eqref{eq:degrootupdate} where $\lambda_t=1$ for
  all $t$.  
\end{proposition} 
We list the proof in the appendix.

\bigskip

In all subsequent sections, we only discuss the situations when
$\tau=0$ and $T(\cdot)>0$, as the other cases may be derived in a
manner similar to what we have sketched in this section.

\section{Opposition}\label{sec:opposition}
In this section, we consider the situation when two subsets of agents
`oppose' each other. Such opposition may derive, for example, from
in-group 
vs.\ out-group antagonisms, as is an important concept in psychology
and sociology (cf.\ Brewer (1979) \cite{Brewer1979}, Castano, Yzerbyt,
Bourguignon, and Seron (2002) \cite{Castano2002}, 
Kitts (2006) \cite{Kitts2006})
and as has also more recently been taken into account in economics 
models (cf., in an experimental context, e.g., Charness, Rigotti, and
Rustichini 
(2007) \cite{Charness2007}, Fehrler and Kosfeld
(2013) \cite{Fehrler2013}) and in social network theory (cf.\ Beasley
and Kleinberg (2010) \cite{Beasley2010}). 
Prime exemplars of 
opposition forces can be found in politics (e.g., democrats
vs.\ republicans; opposition parties vs.\ governing party in charge),
for example, 
or also on a more 
global 
societal level (e.g., 
punks or hippies/counterculture vs.\ mainstream culture). 
In the context of (DeGroot-like) opinion dynamics models, opposition
has, prominently, been discussed in Eger (2013) \cite{Eger2013} (but
see also our discussion in Section \ref{sec:related}), whose
modeling we relate to. 

In the model of Eger (2013) \cite{Eger2013}, there are \emph{two}
types of links between agents. One link type refers to whether agents
follow or oppose each other and the other link type denotes the
intensity of relationship and is given by a non-negative real number
$W_{ij}\in\real$. Belief updating is then performed via the operation
\begin{align}\label{eq:egerupdate}
  \mathbf{b}^k(t+1) = (\mathbf{W}^{(k)}\circ\mathbf{F}^{(k)})\mathbf{b}^k(t),
\end{align}
whereby the operator $\mathbf{W}^{(k)}\circ\mathbf{F}^{(k)}$ is
defined via
\begin{align*}
  \bigl((\mathbf{W}^{(k)}\circ\mathbf{F}^{(k)})(\mathbf{b})\bigr)_{i}=
  \sum_{j=1}^n W_{ij}^{(k)}F_{ij}^{(k)}(b_{j}),
\end{align*}
whereby $F_{ij}^{(k)}\in\set{F,D}$, where $F:S\goesto S$ is the
identity function  
(`agent $i$ follows agent $j$') and $D:S\goesto S$ is an opposition
function (`agent 
$i$ opposes/deviates from agent $j$'). In other words, in this model,
agents form their current beliefs by inverting (via $D$) or not (`via
$F$') the 
past belief signals of others and then taking a weighted arithmetic
average, as in standard DeGroot learning, 
of the so modified (or not) belief signals of their neighbors. As
becomes evident, endogenizing this model would require endogenizing two
variables, namely, $F_{ij}^{(k)}$ and $W_{ij}^{(k)}$, a task that is
beyond the scope of this section. Therefore, we take $F_{ij}^{(k)}$ as
exogenous and keep, as before, $W_{ij}^{(k)}$ as an endogenous variable,
formed, in the case that $F_{ij}=F$, by reference to an agent's past
performance.\footnote{If $F_{ij}=D$, it would make no sense, or be at
  least problematic, to posit that an agent would increase his
  intensity of relationship, relating to opposition behavior, in
  proportion to another 
agent's accuracy of predicting truth.}  

\begin{figure}
        \centering
        \begin{tikzpicture}[-,>=stealth',shorten >=1pt,auto,node distance=2.5cm,
  thick,main node/.style={circle,fill=blue!20,draw,font=\sffamily\Large\bfseries}]

  \begin{scope}
  \node[main node] (1) {1};
  \node[main node] (2) [below left of=1] {2};
  \node[main node] (3) [below right of=2] {3};
  \node[main node] (4) [below right of=1] {4};

  \node[main node] (5) [right= 4.75cm of 1] {5};
  \node[main node] (6) [below left of=5] {6};
  \node[main node] (7) [below right of=6] {7};

  \tikzset{Friend/.style   = {
                                 double          = green,
                                 double distance = 1pt}}
  \tikzset{Enemy/.style   = {
                                 double          = red,
                                 double distance = 1pt}}

  \draw[Friend](1) to (2); 
  \draw[Friend](1) to (3);
  \draw[Friend](1) to (4);
  \draw[Friend](2) to (3);
  \draw[Friend](2) to (4);
  \draw[Friend](3) to (4);

  \draw[Friend](5) to (6);
  \draw[Friend](5) to (7);
  \draw[Friend](6) to (7);

  \draw[Enemy](1) to (5);
  \draw[Enemy](1) to (6);
  \draw[Enemy](3) to (6);
  \draw[Enemy](3) to (7);
  \draw[Enemy](4) to (6);
  \begin{pgfonlayer}{background}
    \filldraw [line width=4mm,join=round,black!10]
      (1.north  -| 2.west)  rectangle (3.south  -| 4.east)
      (5.north -| 6.west) rectangle (7.south -| 7.east);
  \end{pgfonlayer}
\end{scope}
\end{tikzpicture}
        \caption{
        Illustration of an opposition bipartite operator
        $\mathbf{F}$ (agent nodes as blue circles). For convenience,
        $D$ relationships are indicated 
        in red, and $F$ relationships in green. Here, we draw the
        network of agents as undirected, although we generally allow
        directed links between agents.}      
        \label{fig:bip}
\end{figure}
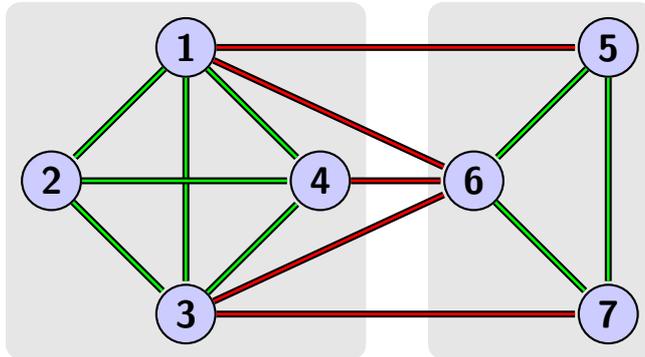
Hence, we consider the following situation. Denote by
$\struct{A}\subseteq [n]$ and $\struct{B}\subseteq[n]$ the two groups
of agents that oppose each other. We posit that $\mathbf{F}$ is
\emph{opposition bipartite} (cf.\ Figure \ref{fig:bip}): for all agents
$i,i'\in\struct{A}$ it holds that $F_{ii'}=F_{i'i}=F$ (analogously for
$\struct{B}$) and for all agents $i\in\struct{A}$ and
$i'\in\struct{B}$ it holds that $F_{ii'}=F_{i'i}=D$, which simply
means that agents within the two groups follow each other whereas agents
across the two groups oppose each other.\footnote{This specificiation
  is not self-evident; intra-group antagonisms, based, e.g., on
  personalized differences between members of the same group, might
  plausibly be allowable.} Next, we assume that, regarding weight
adjustments, members of both groups \emph{ignore} members of the other
group (a sin or bias of omission), taking into account only members of
their own group, that is, 
\begin{align}\label{eq:adjusteger}
  W_{ij}^{(k+1)} = 
      \begin{cases}
        {W}_{ij}^{(k)}+\delta\cdot T(\abs{N({\mathbf{b}^k(\tau)},\mu_k)}) & \text{if }
        \norm{b_{j}^k(\tau)-\mu_k}<\eta \text{ and } G(i)=G(j),\\
        {W}_{ij}^{(k)}& \text{otherwise},
      \end{cases}
\end{align}
where $G(i)$ denotes the group of agent $i$, which is either
$\struct{A}$ or $\struct{B}$; for simplicity, assume $T(\cdot)=1$,
here and in the remainder of this section. Finally, we assume that
agents $i$ of group 
$\struct{A}$ initially assign uniform intensity of relationship
$W_{ij}^{(1)}=
b$ to each member $j$ of
group $\struct{B}$ and members of group $\struct{B}$ do analogously,
assigning $W_{ij}^{(1)}=
c$, where
$b$ and $c$ are positive constants. We
also 
assume that these levels stay fixed over topics, that is,
$W_{ij}^{(k)}=W_{ij}^{(1)}$ whenever $G(i)\neq G(j)$. When
$G(i)=G(j)$, as said, we let weights be formed according to
\eqref{eq:adjusteger}. Finally, we always assume that weight matrices
$\mathbf{W}^{(k)}$ are row-stochastic. We now discuss the so specified
model, with endogenous weight (or intensity) formation for at least a
subset of 
agents, in the following example. For opposition function $D$, we let
$D$ be soft opposition on $S=\real$ (see Eger (2013) \cite{Eger2013} for
details) with the functional form $D(x)=-x$. We 
first outline 
the
following proposition from Eger (2013) \cite{Eger2013}, which gives
conditions for convergence of $\mathbf{W}\circ\mathbf{F}$, where we omit,
here and in the following, reference to topics $X_k$ for notational
convenience.  
\begin{proposition}\label{prop:convergenceToPol}
  Let $D$ be soft opposition on $S=\real$. 
  Then, 
  $\mathbf{W}\circ\mathbf{F}$ is 
  affine-linear with representation $(\mathbf{A},\mathbf{0})$. Then,
  if $\mathbf{F}$ is opposition bipartite,
  $\lambda=1$ is an eigenvalue of $\mathbf{A}$. If $\lambda=1$ is the
  only eigenvalue of $\mathbf{A}$ on the unit circle and if
  $\lambda=1$ 
  has algebraic multiplicity of $1$, then 
  $\lim_{t\goesto\infty}(\mathbf{W}\circ\mathbf{F})^t\mathbf{b}(0)=\mathbf{p}$
  for some polarization opinion vector $\mathbf{p}$ (that depends on
  $\mathbf{b}(0)$) and all initial opinion vectors $\mathbf{b}(0)\in
  S^n$. 
\end{proposition}
  In the proposition, a polarization opinion (or belief) vector is any
  vector $\mathbf{p}$ consisting of two beliefs $a,b\in S$ such that
  $D(a)=b$ and $D(b)=a$. For our specification of $D$, this would mean
  that $a=-b$. 
  Hence, the proposition says that if $D$ is soft opposition, then
  $\mathbf{W}\circ\mathbf{F}$ is representable by a matrix
  $\mathbf{A}$, and if in addition $\mathbf{F}$ is opposition
  bipartite (as we assume), then convergence of
  $\mathbf{W}\circ\mathbf{F}$ depends on the eigenvalues of
  $\mathbf{A}$. 
  In the following example, we will make reference to the
  proposition. 
\begin{example}\label{example:opposition}
  Let
  $n_1=\abs{\struct{A}}$ and $n_2=\abs{\struct{B}}$ with
  $n_1+n_2=n$. 
  Before (partly) endogenizing $\mathbf{W}$,  
  assume first that $\mathbf{W}$ and $\mathbf{F}$ have the
  following form,
  \begin{align}\label{eq:formWF}
    \mathbf{W} = \begin{pmatrix}
    \mathbf{W}_{\mathcal{A},\mathcal{A}}
    & \mathbf{W}_{\mathcal{A},\mathcal{B}}\\
    \mathbf{W}_{\mathcal{B},\mathcal{A}} & \mathbf{W}_{\mathcal{B},\mathcal{B}}
    \end{pmatrix}
    \quad
    \mathbf{F} = 
    \begin{pmatrix}
    \mathbf{F}_{\mathcal{A},\mathcal{A}}
    & \mathbf{F}_{\mathcal{A},\mathcal{B}}\\
    \mathbf{F}_{\mathcal{B},\mathcal{A}} & \mathbf{F}_{\mathcal{B},\mathcal{B}}
    \end{pmatrix}
  \end{align}
%
  where $[\mathbf{W}_{\mathcal{A},\mathcal{A}}]_{ij}=a$,
  $[\mathbf{W}_{\mathcal{A},\mathcal{B}}]_{ij}=b$,
  $[\mathbf{W}_{\mathcal{B},\mathcal{A}}]_{ij}=c$,
  $[\mathbf{W}_{\mathcal{B},\mathcal{B}}]_{ij}=d$, and 
  $[\mathbf{F}_{\mathcal{A},\mathcal{A}}]_{ij}=[\mathbf{F}_{\mathcal{B},\mathcal{B}}]_{ij}=F$,
  $[\mathbf{F}_{\mathcal{A},\mathcal{B}}]_{ij}=[\mathbf{F}_{\mathcal{B},\mathcal{A}}]_{ij}=D$;
  matrices $\mathbf{W}_{\mathcal{A},\mathcal{A}}$ and
  $\mathbf{F}_{\mathcal{A},\mathcal{A}}$ are of size $n_1\times n_1$,
  $\mathbf{W}_{\mathcal{A},\mathcal{B}}$ and
  $\mathbf{F}_{\mathcal{A},\mathcal{B}}$ of size $n_1\times n_2$,
  etc. 
  Hence,
  agents in $\struct{A}$ follow each other, assigning weight
  $a$ to each other, and agents in $\struct{B}$ assign weight $d$ to
  each other; across the two sets, agents oppose each other, with
  weights $b$ and $c$, respectively, as already indicated above. 
  Moreover, for simplicity, as the given specification posits, we
  assume that weights are uniform within 
  groups and opposition weights are also uniformly distributed. 
  Since, as 
  Proposition \ref{prop:convergenceToPol} tells, the so defined
  $\mathbf{W}\circ\mathbf{F}$ allows an 
  (affine-)linear representation, this is given by, in this setup, 
  \begin{align}\label{eq:A}
    \mathbf{A} = 
    \begin{pmatrix}
       \mathbf{W}_{\mathcal{A},\mathcal{A}}
    & -\mathbf{W}_{\mathcal{A},\mathcal{B}}\\
    -\mathbf{W}_{\mathcal{B},\mathcal{A}} & \mathbf{W}_{\mathcal{B},\mathcal{B}}
    \end{pmatrix},
  \end{align}
  as one can 
  verify (cf.\ Eger (2013) \cite{Eger2013}). 
  Furthermore, if
  $(\mathbf{W}\circ\mathbf{F})^t\mathbf{b}(0)=\mathbf{A}^t\mathbf{b}(0)$
  converges to a polarization vector (e.g., under the conditions of
  Proposition \ref{prop:convergenceToPol}), then the one limiting belief is
  given by $\sum_{j=1}^n s_jb_j(0)$ and the other is given by
  $-\sum_{j=1}^n s_jb_j(0)$, where
  $\mathbf{s}=(s_1,\ldots,s_n)^\intercal$ is 
  the unique
  eigenvector of $\mathbf{A}^\intercal$ satisfying
  $\sum_{j=1}^n\length{s_j}=1$ and corresponding to eigenvalue
  $\lambda=1$ of $\mathbf{A}^\intercal$ (cf.\ Eger (2013)
  \cite{Eger2013}, Remark 6.4). The vector $\mathbf{s}$ is 
  then a (generalized) social influence vector (cf.\ the concept of
  eigenvector centrality, e.g., Bonacich (1972) \cite{Bonacich1972}),
  with $\length{s_i}$ 
  denoting the social influence (proper) of agent $i$ and
  $\text{sgn}(s_i)$ his group membership. Since, by our
  specification of $\mathbf{W}\circ\mathbf{F}$, agents in group 
  $\struct{A}$ must have the same social influence (by homogeneity of
  these agents due to the uniform weight structure) as well as agents
  in group $\struct{B}$, we must 
  accordingly have that
  $\mathbf{s}=(\underbrace{x,\ldots,x}_{n_1},\underbrace{y,\ldots,y}_{n_2})^\intercal$
  for some $x,y\in\real$. 
  Then, $y$ (or $\length{y}$) measures social influence of members of
  group $\struct{B}$ and accordingly for $\struct{A}$. 
  Hence, from
  $\mathbf{A}^\intercal\mathbf{s}=\mathbf{s}$, we find (1)
  $n_1ax-n_2cy=x$, (2) $-n_1bx+n_2dy=y$, and (3) $n_1x-n_2y=1$ (from
  the unit condition on $\mathbf{s}$). From this, it follows that
  \begin{align}\label{eq:closed}
    y = \frac{b}{n_2(d-b)-1},\quad\text{and}\quad x=(1+n_2y)a-n_2cy. 
  \end{align}
  The case of $y$ may serve as an illustration. 
  Computing the comparative statics of $\length{y}$ with respect to
  $b$ and $d$, we first find that
  since $n_2(d-b)\le
  n_2d\le 1$, it holds that $y\le 0$. Hence,
  $\length{y}=\frac{b}{1-n_2(d-b)}$ and then,
  \begin{align*}
  \pardiv{\length{y}}{b}=\frac{1-n_2d}{\bigl(1-n_2(d-b)\bigr)^2}\ge
  0,\quad 
  \pardiv{\length{y}}{d}=\frac{n_2b}{\bigl(1-n_2(d-b)\bigr)^2}>
  0
  \end{align*}
  such that 
  an increase in $d$ leads to an increase
  in the absolute value of $y$, as we would expect --- if the weight
  that members of group $\struct{B}$ place upon each other increases,
  their social influence, measured in absolute value, 
  increases. Moreover, $\length{y}$ also increases in $b$ --- the more
  members 
  of group $\struct{A}$ want to disassociate from members of group
  $\struct{B}$, the more does group $\struct{B}$'s social influence
  increase, in absolute value. We exemplify in Figure \ref{fig:pol}
    (left).  
  
  Hence, under our current assumptions, limiting polarization beliefs
  of agents are given by $\sum_{j\in[n]}s_jb_j(0)=\sum_{j\in\struct{A}}xb_j(0)+\sum_{j\in\struct{B}}yb_j(0)$ and
  $-\sum_{j\in[n]}s_jb_j(0)=-\sum_{j\in\struct{A}}xb_j(0)-\sum_{j\in\struct{B}}yb_j(0)$,
  respectively. Let us, for the moment,  
  assume that all agents are $\epsilon$-intelligent with $\epsilon=0$,
  that is, all agents precisely receive truths for topics, as initial
  beliefs. Then, limiting 
  beliefs are, thus,  
  \begin{align*}
    b_{\mathcal{A}}^k(\infty)=\sum_{j\in[n]}s_jb_j(0)=\mu_k\underbrace{\Bigl(n_1x+n_2y\Bigr)}_{=\mathsf{c}},\quad
    \text{and}\quad  b_{\mathcal{B}}^k(\infty) = 
    -\sum_{j\in[n]}s_jb_j(0)=\mu_k\Bigl(\underbrace{-\bigl(n_1x+n_2y\bigr)}_{=-\mathsf{c}}\Bigr), 
  \end{align*}
  respectively, where closed-form solutions of $x$ and $y$ are given
  in Equation \eqref{eq:closed}. In Figure \ref{fig:pol} (right), we
  plot, for $c=\frac{1}{2n}$ and $d=\frac{1-n_1c}{n_2}$ fixed, the
  coefficient 
  $\mathsf{c}=n_1x+n_2y$ of limiting 
  beliefs (and its negative, as coefficient of the other limiting
  belief), as a  
  function of $b$ (and, hence, also of $a$ since
  $a=\frac{1-n_2b}{n_1}$); note that this 
  coefficient denotes the `scaling' of truth in the limiting beliefs,
  whence, if it is equal to $1$, (some) agents exactly reach truth. We
  observe the following: if $b$ is very low (compared with $c$), i.e.,
  agents 
  in group $\struct{A}$ care little about agents in group $\struct{B}$
  (at least relatively) --- that is, opposition from $\struct{A}$ to
  $\struct{B}$ is (relatively) low --- then $\mathsf{c}$ is very close
  to $1$, which means that agents in group $\struct{A}$ have limiting
  beliefs very close to truth, while agents in group $\struct{B}$ hold
  limiting beliefs that are very close to $-\mu_k$, the `opposite' of
  truth. As $b$ increases, $\mathsf{c}$ becomes smaller, approaching
  zero as $b=c$. In other words, if opposition `force' is equal between
  groups $\struct{A}$ and $\struct{B}$ --- in the sense that $b=c$ ---
  then both groups reach limiting beliefs of zero, no matter what
  truth is. As group $\struct{A}$ begins to oppose group $\struct{B}$
  more heavily than $\struct{B}$ opposes $\struct{A}$, that is, $b>c$,
  group $\struct{A}$ goes further away from truth, toward opposite
  levels of truth in that $\mathsf{c}$ becomes negative, while group
  $\struct{B}$ begins to approach truth. 
  \begin{figure*}[!ht]
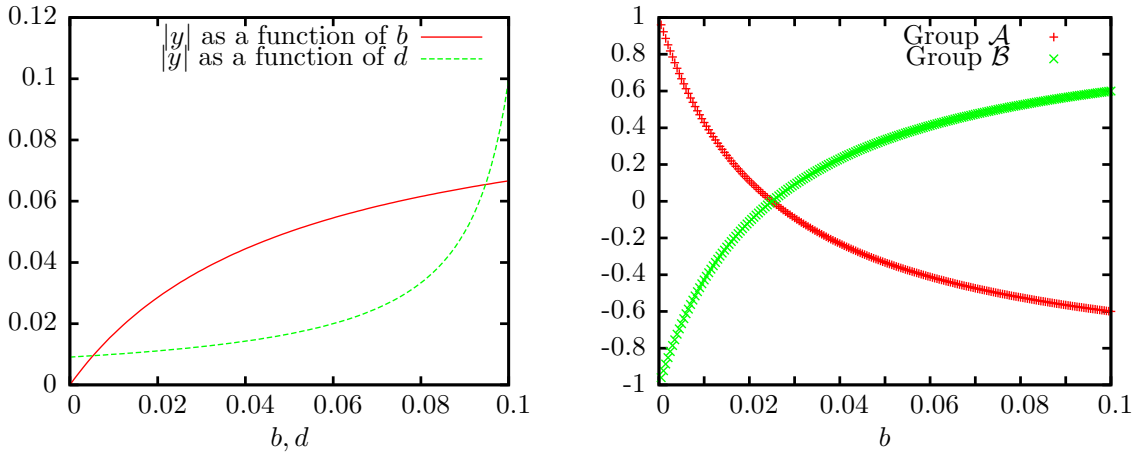

  \centering
  \begin{subfigure}{0.49\textwidth}
        \centering  
        \resizebox{1.0\textwidth}{!}{
        \input{plots/y.tex}}
    \end{subfigure}
  \begin{subfigure}{0.49\textwidth}
    \centering
    \resizebox{1.0\textwidth}{!}{
      \input{plots/polarization.tex}}
    \end{subfigure}
    \caption
{ 
  Both graphs: $n=n_1+n_2=10+10=20$ agents. 
  Left: Social influence, $\length{y}$, of agent $i$ of group
  $\struct{B}$ increases both as 
  a function 
  of $d$ ($b$ fixed), `within-group trust' of members of group
  $\struct{B}$, and 
  $b$ ($d$ fixed), the importance assigned to members of group
  $\struct{B}$ via agents of 
  group $\struct{A}$. Note that $n_2b\le 1$ and $n_2d\le 1$ (by
  row-stochasticity of weight matrices $\mathbf{W}$), which
  implies, in our case, $b,d\le \frac{1}{10}$.  
  Right: $c=\frac{1}{2n}=0.025$, $d=\frac{3}{2n}$ fixed. 
  Coefficient 
  $\mathsf{c}=(n_1x+n_2y)$ of $\mu_k$ (red) and $-\mathsf{c}$ (green)
  as a function 
  of $b$. Description in 
  text. 
  }
\label{fig:pol}
\end{figure*}

  Now, concerning the question whether the conditions on the
  eigenvalues of matrix $\mathbf{A}$, stated in Proposition
  \ref{prop:convergenceToPol}, are satisfied --- that is, do agents in
  fact converge to a polarization? --- we mention that exactly
  determining the spectrum of $\mathbf{A}$ is difficult in the current 
  situation, for general $n_1$ and $n_2$, and $a,b,c,d$. For $n_1=1$
  and $n_2$ arbitrary (and, by symmetry hence also for $n_2=1$ and
  $n_1$ arbitrary), we find, in the appendix, that  
  $\mathbf{A}$
  has exactly one eigenvalue, namely $\lambda=1$, on the unit circle
  and whose algebraic multiplicity is $1$. Thus, in this case, by
  Proposition \ref{prop:convergenceToPol}, 
  beliefs under $\mathbf{W}\circ\mathbf{F}$ indeed converge to a
  polarization, as we have sketched it, and limiting beliefs have the
  indicated form. We strongly suspect that this is true for arbitrary
  $n_1$ and $n_2$, but leave the derivation open. 

  Finally, when would we expect $\mathbf{W}\circ\mathbf{F}$ to have
  the form \eqref{eq:formWF}, taking the form of $\mathbf{F}$ as
  exogenous? The structure of $\mathbf{W}$ holds, for instance, when 
  $\mathbf{W}^{(1)}$ has the form indicated in \eqref{eq:formWF},
  agents adjust weights (to members of their own group) based on
  initial beliefs, $\tau=0$, and, e.g., $\epsilon=0$ (agents' initial
  beliefs are exactly $\mu_k$); then all $\mathbf{W}^{(k)}$ have the
  form as given in \eqref{eq:formWF}. Form \eqref{eq:formWF} also
  arises, in the limit, as $k$ becomes large, when $\tau=0$ and
  initial beliefs are stochastically centered around truth and each
  agent has the same variance (namely, agents then tend to assign
  uniform weights to those they take into consideration in adjusting
  weights; uniform weights for outgroup members have been assumed
  exogenous by us, anyways). In fact, the simulations shown in Figure
  \ref{fig:polSim}, 
  for the latter case, show good agreement with
  the analytical predictions for the situation when all agents are
  $\epsilon$-intelligent, for $\epsilon=0$, even for small $k$,
  indicating that $\mathbf{W}^{(k)}$ has a form close
  to \eqref{eq:formWF} quickly, when agents are stochastically
  intelligent with identical variances. 
\end{example}

\begin{figure*}[!ht]
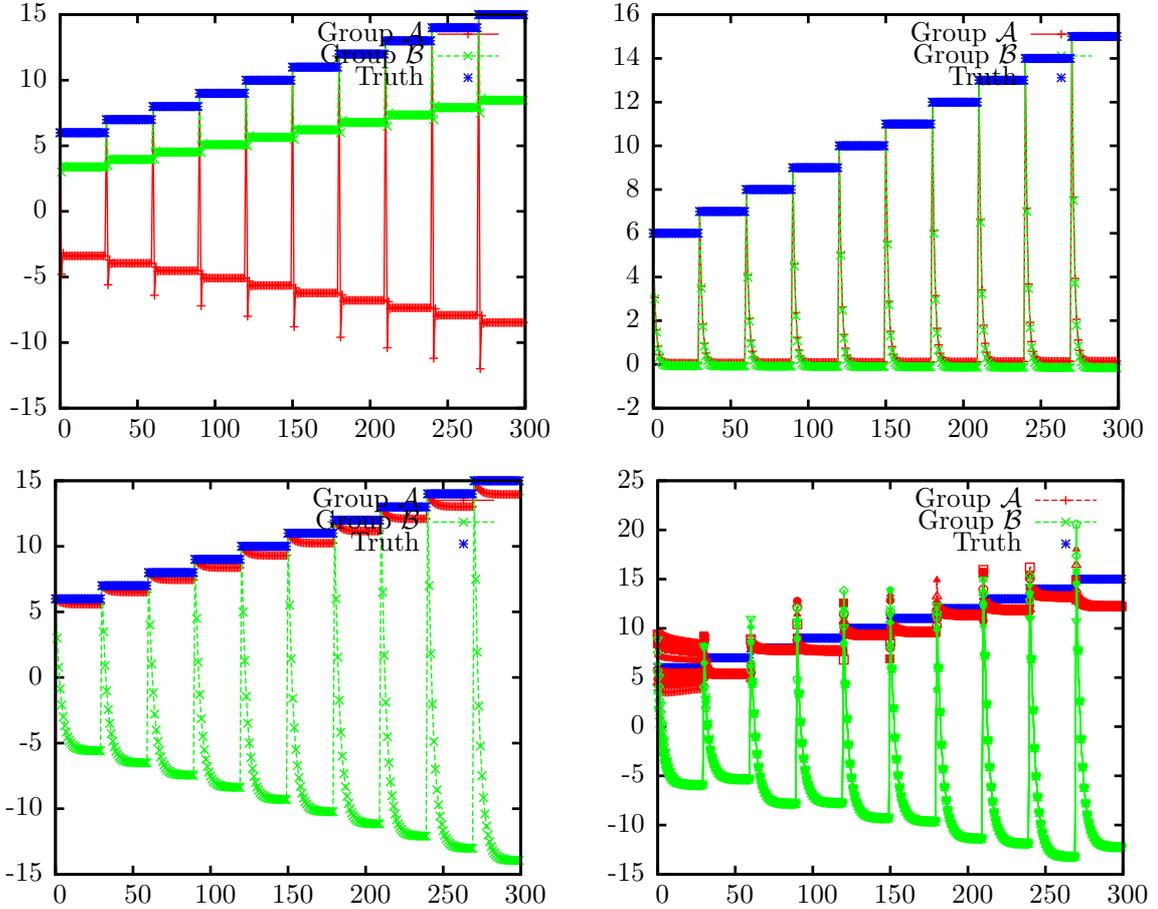

  \centering
      \begin{subfigure}{0.49\textwidth}
        \centering  
        \resizebox{1.0\textwidth}{!}{
        \input{plots/opposition0-09.tex}}
    \end{subfigure}%
    \begin{subfigure}{0.49\textwidth}
        \centering 
        \resizebox{1.0\textwidth}{!}{
        \input{plots/opposition0-0245.tex}}
    \end{subfigure}\\%
    \begin{subfigure}{0.49\textwidth}
        \centering  
        \resizebox{1.0\textwidth}{!}{
        \input{plots/opposition0-0009.tex}}
    \end{subfigure}
    \begin{subfigure}{0.49\textwidth}
        \centering  
        \resizebox{1.0\textwidth}{!}{
        \input{plots/opposition0-0009r.tex}}
    \end{subfigure}
    \caption
{ 
  Throughout: $n=n_1+n_2=10+10=20$ agents, $c=\frac{1}{2n}=0.025$,
  $d=\frac{3}{2n}$ fixed. Discussion of $10$ topics; 
  $t=0,1,2,\ldots,20$ discussion periods shown, for each topic. Truth
  $\mu_k=k+5$, for 
  $k=1,\ldots,10$ (black lines). In red and green: 
  members of groups $\struct{A}$ and $\struct{B}$, respectively.   
  Top left: All agents receive initial beliefs $b_i^k(0)=\mu_k$,
  $b=0.09$ (cf.\ Figure \ref{fig:pol} (right)). Top right:
  $b_i^k(0)=\mu_k$, 
  $b=0.0245$. Bottom left: $b_i^k(0)=\mu_k$,
  $b=0.0009$. Bottom right: $b_i^k(0)=\mathtt{N}(\mu_k,4)$,
  $b=0.0009$. 
  }
    \label{fig:polSim}
\end{figure*}
\begin{remark}
  In Example \ref{example:opposition}, we have outlined conditions
  under which we expect, due to polarization, at most one group of
  agents to become wise for topics. The conditions that we have
  highlighted --- e.g., $\epsilon$-intelligence, for $\epsilon=0$, or
  initial beliefs stochastically centered around truth whereby all
  agents have identical variances --- might appear quite special. We
  believe that similar polarization results hold for much more general
  conditions, but those outlined have the benefit of being
  analytically tractable more easily while still indicating results,
  as we think, of a
  general nature. 
\end{remark}
\begin{remark}
  To summarize the importance of results indicated in Example
  \ref{example:opposition}, note that in this section, agents have
  been influenced by two `polar' forces. On the one hand, they were
  attracted by truth by their adherence to principles that
  potentially lead them closer to truth --- e.g., weight adjustment to
  those agents in their in-group that have been truthful in the past. On
  the other hand, agents had --- exogenously --- specified antagonisms
  to members of another group (a sin or bias of commission), their
  outgroup, which drew them toward 
  beliefs that are different from those held by their adversaries. The
  message from Example \ref{example:opposition} is clear in this
  context: the group that has (relatively) stronger incentives to
  disassociate from negative referents tendentially will drift away from truth
  considerably, while the group with (relatively) weaker such
  incentives may still become wise (under appropriate initial conditions
  on beliefs), which is an intuitive result since, for the former
  group, disassociation seems to be (relatively) more critical than
  truth. 
\end{remark}


\section{Conformity}\label{sec:conformity}
Buechel, Hellmann, and Kl\"{o}{\ss}ner (2013) \cite{Buchel2013} and
Buechel, Hellmann, and Kl\"{o}{\ss}ner
(2012) \cite{Buchel2012}\footnote{Henceforth, we only relate to the
more recent version of their paper, unless the difference becomes
important.} study 
a DeGroot-like opinion dynamics model under \emph{conformity}, 
that is, where individuals are not
only \emph{informationally} socially influenced by others 
but also \emph{normatively} in that they are
motivated to state opinions that tend to fit the group norm, possibly,
in order 
to get ``utility gain[s] by simply making the same choice as one's
reference group'' (cf.\ Zafar 2011 \cite{Zafar2011}, p.\ 774). A
classical example of such conforming behavior is documented in the
famous study of Asch (1955) \cite{Asch1955} where subjects wrongly
judged the length of a stick after some other, supposedly neutral,
participants had given the same wrong judgment. More examples and
relevant theoretical and empirical literature, e.g., Deutsch and
Gerard (1955) \cite{Deutsch1955}, Jones (1984)
\cite{Jones1984}, etc., on conforming behavior among human subjects are
directly 
provided in Buechel, Hellmann, and Kl\"{o}{\ss}ner (2013)
\cite{Buchel2013}. As we have indicated in the introduction, we may
perceive of conformity to a reference opinion, in our context, as a bias
toward  
the beliefs of 
one's reference group. 

Mathematically, agents in the named model update their beliefs
\emph{informationally} 
according to the 
following rule,
\begin{align}\label{eq:buechel}
  \mathbf{b}(t+1) =
  \mathbf{D}\mathbf{b}(t)+(\mathbf{W}-\mathbf{D})\mathbf{s}(t), 
\end{align} 
where $\mathbf{s}(t)\in S^n$ denotes the vector of \emph{stated
  opinions} or beliefs (whose formation, as assumed, underlies
normative social pressure, as we indicate below), 
$\mathbf{b}(t)\in S^n$ denotes the vector of \emph{true beliefs},
$\mathbf{W}$ is the social network (or, `learning matrix')
as in the standard DeGroot model, and $\mathbf{D}$ denotes its
diagonal. Updating rule \eqref{eq:buechel} says that agents form their
current beliefs by taking a weighted arithmetic average of their past
true beliefs and others' stated beliefs. Then, as concerns \emph{normative
social influence}, agents are assumed to choose stated beliefs
$s_i(t)$ 
by reference to the 
utility maximization 
problem 
\begin{align}\label{eq:buechel_util}
  u_i(\mathbf{s};\mathbf{b}) = -(1-\delta_i)(s_i-b_i)^2-\delta_i(s_i-q_i)^2, 
\end{align}
whereby the term $(s_i-b_i)^2$ represents an agent's preference for
honesty (misrepresenting true opinions may cause cognitive discomfort, 
cf.\ Festinger (1957) \cite{Festinger1957}) and the term $(s_i-q_i)^2$
represents 
preference for conforming 
to a reference opinion
$q_i$. 
The parameter $\delta_i\in(-1,+1)$ displays the relative importance of
the 
preference for conformity in relation to the preference for
honesty. If $\delta_i<0$, then agents have preference for
\emph{counter-conformity} in that their reference group serves as a negative
referent. Now, consider that at the end of each (opinion updating) round
$t=0,1,2,\ldots$, agents play a normal form game
$([n],S^n,u_i(\cdot;b_i(t)))$. Let
$\mathbf{q}=(q_1,\ldots,q_n)^\intercal$ and let
$\mathbf{q}(t)=\mathbf{Q}\mathbf{s}(t)$ where $\mathbf{Q}$ is an
$n\times n$ matrix that indicates how reference opinions are formed
from stated opinions; we assume that $Q_{ii}=0$ for all $i\in[n]$ such
that agents do not take into account their own stated opinion in
reference opinion formation\footnote{They know better anyway, by knowledge of
their true opinions.} and we also assume that $\mathbf{Q}$ is
row-stochastic. 
The next proposition says that the normal form
game has a unique Nash equilibrium.
\begin{proposition}\label{prop:buechel1}
  Denote by
  $\mathbf{\Delta}$ 
  the diagonal matrix with
  $\Delta_{ii}=\delta_i$. Then the normal form game
  $([n],S^n,u(\cdot;\mathbf{b}(t)))$, for
  $u(\cdot;\mathbf{b}(t))=(u_1(\cdot;\mathbf{b}(t)),\ldots,u_n(\cdot;\mathbf{b}(t)))$, 
  has a unique Nash equilibrium, which  
  is given by
  \begin{align*}
    \mathbf{s}^*(t) = \inv{(\mathbf{I}_n-\mathbf{\Delta}
      \mathbf{Q})}(\mathbf{I}_n-\mathbf{\Delta})\mathbf{b}(t) =
    \tilde{\mathbf{Q}}\mathbf{b}(t). 
  \end{align*}
\end{proposition}
We prove Proposition \ref{prop:buechel1} in the appendix. The 
proposition is a (straightforward) extension of the corresponding
proposition, Proposition 1, in Buechel, Hellmann, and Kl\"{o}{\ss}ner
(2012) \cite{Buchel2012} in that they choose the particular
$\mathbf{Q}$ with $Q_{ij}=\frac{W_{ij}}{1-W_{ii}}$ ($Q_{ii}=0$). In the revised
version of their paper, the named authors also specify an iterative
process that explains how agents reach the Nash equilibrium
$\mathbf{s}^*(t)$ but we omit the recapitulation of this idea, because
it is rather technical and does not provide further insight at this
point. 

Hence, simply assuming that agents play the Nash equilibrium
$\mathbf{s}^*(t)$ at the end of each period $t$ (such that, for $t+1$,
$\mathbf{b}(t)$ \emph{and} $\mathbf{s}(t)$ are available), beliefs
evolve 
according to, combining \eqref{eq:buechel} with $\mathbf{s}^*(t)$,
\begin{align*}
  \mathbf{b}^k(t+1) = \mathbf{M}^{(k)}\mathbf{b}^k(t),
\end{align*}
where 
\begin{align}\label{eq:M}
  \mathbf{M}^{(k)} = \mathbf{D}^{(k)}+(\mathbf{W}^{(k)}-\mathbf{D}^{(k)})\tilde{\mathbf{Q}}^{(k)}=\mathbf{D}^{(k)}+(\mathbf{W}^{(k)}-\mathbf{D}^{(k)})\inv{(\mathbf{I}-\mathbf{\Delta}^{(k)}\mathbf{Q}^{(k)})}(\mathbf{I}-\mathbf{\Delta}^{(k)}),
\end{align}
and where we also index matrices by topic indices. As becomes
obvious, this model has now many variables that can potentially be
endogenized, namely, $\mathbf{W}^{(k)}$, $\mathbf{\Delta}^{(k)}$,
which summarizes the conformity parameters, and $\mathbf{Q}^{(k)}$,
which summarizes how agents form reference opinions. In the following,
we take $\mathbf{\Delta}^{(k)}$ as exogenously given and constant
across 
topics; the elements $[\mathbf{\Delta}^{(k)}]_{ii}=\delta_i$ may then
be perceived as `personality traits' of individuals. For
$\mathbf{W}^{(k)}$, we assume the same endogenous weight formation as
before, where weight increments depend upon past performance. The
matrix $\mathbf{Q}^{(k)}$, we take as arbitrary exogenous variable first,
satisfying row-stochasticity and $Q_{ii}=0$ as above, and specialize
then in the examples.  

Our first proposition paths the way for a convergence result in our
situation. It says that the property of having a positive column
propagates from $\mathbf{W}^{(k)}$ to $\mathbf{M}^{(k)}$ if no agent
is counter-conforming. 
\begin{proposition}\label{prop:column}
  Let $k\ge 1$ be arbitrary. 
  Let $\delta_i\ge 0$ for all $i\in[n]$ such that agents never
  counter-conform. Then, if $\mathbf{W}^{(k)}$ has a 
  positive column, then so does $\mathbf{M}^{(k)}$.
\end{proposition}
\begin{proof}
  By the proof of Proposition \ref{prop:buechel1}, given in the
  appendix, the inverse of $\mathbf{I}_n-\mathbf{\Delta}^{(k)}\mathbf{Q}^{(k)}$
  always exists (as long as $\length{\delta_i}<1$, which we assume
  throughout) and is given by $\sum_{r=0}^\infty
  \bigl(\mathbf{\Delta}^{(k)}\mathbf{Q}^{(k)}\bigr)^r$. Since
  $\delta_i\ge 0$ and since $\mathbf{Q}^{(k)}$ is assumed to be 
  row-stochastic, the latter sum is a sum of non-negative matrices and
  therefore, the infinite sum yields a matrix with non-negative
  entries. Moreover, since $\mathbf{A}^0=\mathbf{I}_n$ for any
  arbitrary matrix $\mathbf{A}$, all diagonals of $\sum_{r=0}^\infty
  \bigl(\mathbf{\Delta}^{(k)}\mathbf{Q}^{(k)}\bigr)^r$ are hence
  strictly positive ($\mathbf{I}_n$ has strictly positive diagonals
  and the remaining summands are all non-negative). Moreover, since
  $\mathbf{P}:=\mathbf{I}_n-\mathbf{\Delta}^{(k)}$ is a diagonal matrix with each
  entry $P_{ii}\in(0,1]$, 
    \begin{align*}
      \tilde{\mathbf{P}}=
      \inv{\Bigl(\mathbf{I}_n-\mathbf{\Delta}^{(k)}\mathbf{Q}^{(k)}\Bigr)}\bigl(\mathbf{I}_n-\mathbf{\Delta}^{(k)}\bigr)
    \end{align*}
    accordingly also has diagonal entries that are all strictly
    positive. Next, since $\mathbf{W}^{(k)}$ has a strictly positive
    column $j$ by assumption, $\mathbf{W}^{(k)}-\mathbf{D}^{(k)}$ has a
    strictly positive column $j$, except for 
    element $j$ 
    of that column,
    which is zero. Hence, multiplying,
    $\mathbf{W}^{(k)}-\mathbf{D}^{(k)}$, a non-negative matrix by
    assumption, with $\tilde{\mathbf{P}}$ 
    results in a matrix that also has a strictly positive column $j$,
    except possibly for its diagonal. But since $\mathbf{D}^{(k)}$ has
    a positive entry $[\mathbf{D}^{(k)}]_{jj}$ (since column $j$ of
    $\mathbf{W}^{(k)}$ is positive by assumption),
    \begin{align*}
      \mathbf{D}^{(k)}+(\mathbf{W}^{(k)}-\mathbf{D}^{(k)})\inv{(\mathbf{I}-\mathbf{\Delta}^{(k)}\mathbf{Q}^{(k)})}(\mathbf{I}-\mathbf{\Delta}^{(k)})
    \end{align*}
    has a positive column $j$. The latter matrix is, by definition
    (Eq. \ref{eq:M}), precisely the matrix $\mathbf{M}^{(k)}$. 
\end{proof}
As a 
corollary, 
we have our first convergence (to consensus)
result, which 
provides both an alternative to the convergence result provided in
Buechel, Hellmann, and Kl\"{o}{\ss}ner (2013) \cite{Buchel2013}, and a
generalization as well as a strengthening. It is more general since it
considers arbitrary  
$\mathbf{Q}$ rather than the peculiar choice that the named authors
consider. It provides an alternative since it says that conformity and
a positive column of $\mathbf{W}^{(k)}$ are sufficient conditions for
convergence, while the proposition in Buechel, Hellmann,
and Kl\"{o}{\ss}ner (2013) \cite{Buchel2013} states that conformity
and a positive \emph{diagonal} of $\mathbf{W}^{(k)}$ are sufficient
conditions for 
convergence. Finally, it is a strengthening because it states
convergence to \emph{consensus} rather than merely convergence. Before
proving the theorem, we need the following lemma which says that the
rows of $\mathbf{M}^{(k)}$ sum to $1$ and which we prove in the
appendix. 
\begin{lemma}\label{lemma:one}
  The matrix $\mathbf{M}^{(k)}$ defined in \eqref{eq:M} satisfies
  \begin{align*}
  \mathbf{M}^{(k)}\one = \one
  \end{align*}
  for any row-stochastic $\mathbf{Q}$. 
\end{lemma}
\begin{corollary}\label{cor:consensusBuechel}
  Let $k\ge 1$ be arbitrary. Assume that $\delta_i\ge 0$ for all
  $i\in[n]$. Then, if $\mathbf{W}^{(k)}$ has a positive column, then
  $\mathbf{M}^{(k)}$ induces a consensus. 
\end{corollary}
\begin{proof}
  First, if $\delta_i\ge 0$, then $\mathbf{M}^{(k)}$ is a non-negative
  matrix by the proof of Proposition \ref{prop:column}. Moreover, by
  Lemma \ref{lemma:one}, $\mathbf{M}^{(k)}$ is then also
  row-stochastic. Finally, by Proposition \ref{prop:column}, if
  $\mathbf{W}^{(k)}$ has a positive column, then so does
  $\mathbf{M}^{(k)}$. A row-stochastic matrix with positive column
  induces a consensus by Theorem \ref{theorem:degroot}. 
\end{proof}
It might be worthwhile, in future considerations, to study in more
depth which properties 
$\mathbf{M}^{(k)}$ inherits from $\mathbf{W}^{(k)}$, and under which
conditions. As mentioned, Buechel, Hellmann, and Kl\"{o}{\ss}ner
(2013) \cite{Buchel2013}
demonstrate that $\mathbf{M}^{(k)}$ inherits a positive diagonal from
$\mathbf{W}^{(k)}$ (under their particular choice of $\mathbf{Q}$ and
under conformity) as well as the general social network structure (see
their discussion in their Section 4), while we show that the property of
positive columns also propagates, for arbitrary $\mathbf{Q}$. 

For now, we contend ourselves with the fact that
Corollary \ref{cor:consensusBuechel} implies that, as in the standard
DeGroot model, agents almost always reach a consensus --- that
is, for almost all topics --- even under conformity
($\delta_i\ge 0$), under very mild conditions.
\begin{proposition}\label{prop:consensusBuechel}
  Let $\eta\ge 0$, agents' tolerance, be fixed. Let $\tau=0$ (resp.\
  $\tau=\infty$). Let $T(\cdot)>0$. 
  As in Proposition \ref{prop:consensus}, let $r$ be the earliest time
  point that some agent is $\eta$-intelligent (resp. $\eta$-wise) for
  topic $X_r$. Then, 
  under the conformity model presented above, with $\delta_i\ge 0$ for
  all $i\in[n]$, 
  (a) agents reach a consensus for all topics $X_k$ with $k>r$,
  independent of their initial beliefs. (b) For topics $1,\ldots,r$,
  agents' reaching a 
  consensus 
  depends on $\mathbf{W}^{(1)}$ and on $\mathbf{Q}^{(1)}$ to
  $\mathbf{Q}^{(r)}$ as well as on $\mathbf{\Delta}^{(1)}$ to
  $\mathbf{\Delta}^{(r)}$. 
\end{proposition}
\begin{proof}
  (a) As in the corresponding proof of
  Proposition \ref{prop:consensus}, $\mathbf{W}^{(r+1)}$ has a
  positive column and so do, in general, have all matrices
  $\mathbf{W}^{(k)}$, for $k>r$. Hence, by
  Corollary \ref{cor:consensusBuechel}, agents reach a consensus for
  topics $X_k$, for $k>r$, under
  the conformity model, as long as $\delta_i\ge 0$ for all
  $i\in[n]$. (b) Of course, whether or not agents reach a consensus
  for $X_1$ to $X_r$ depends on the parameters of the model. 
\end{proof} 
As mentioned before, if there exist agents whose initial beliefs have
positive probability of lying within an $\eta$-interval around truths,
then $r$, as defined in Proposition \ref{prop:consensusBuechel}, is a
finite number almost surely. For standard parametrizations (e.g., all
agents have positive probability of being truthful, for all topics),
$r$ is very low --- typically $r=1$ --- with probability that goes to
$1$ in $n$, population size (cf.\ Remark \ref{rem:distR}). 
\begin{example}\label{example:counter}
        If agents are counter-conforming,
        Proposition \ref{prop:consensusBuechel} may be false in that
        beliefs may even diverge, rather than lead to a
        consensus. Consider, for instance, 
        \begin{align*}
  \mathbf{W}^{(r+1)} =
  \frac{1}{1+2\delta} 
  \begin{pmatrix}
  1+\delta & \delta & 0 \\
  \delta & 1+\delta & 0 \\
  \delta & \delta & 1 \\
  \end{pmatrix},
  \end{align*} 
  which would be the resulting weight matrix if $\tau=0$ and agent
        $1$'s and $2$'s initial beliefs were
        in an $\eta$-radius around truth, for the first
        time, for topic $X_r$. For convenience, assume that 
        \begin{align*}
  \mathbf{Q}^{(r+1)} = 
  \begin{pmatrix}
  0 & \frac{1}{2} & \frac{1}{2}\\
  \frac{1}{2} & 0 & \frac{1}{2}\\
  \frac{1}{2} & \frac{1}{2} & 0\\
  \end{pmatrix},\quad\mathbf{\Delta}=
  \begin{pmatrix}
  a & 0 & 0 \\ 0 & b & 0 \\ 0 & 0 & c
  \end{pmatrix},
  \end{align*}
  where $-1<a,b,c<1$. Then, sample belief dynamics for this setting are
        sketched in Figure \ref{fig:counter}. As the graphs
        illustrate, under
        counter-conformity, agents may want to diassociate from others
        in a manner strong enough to induce divergence. 
        \begin{figure}
        \centering
        \includegraphics[scale=0.25]{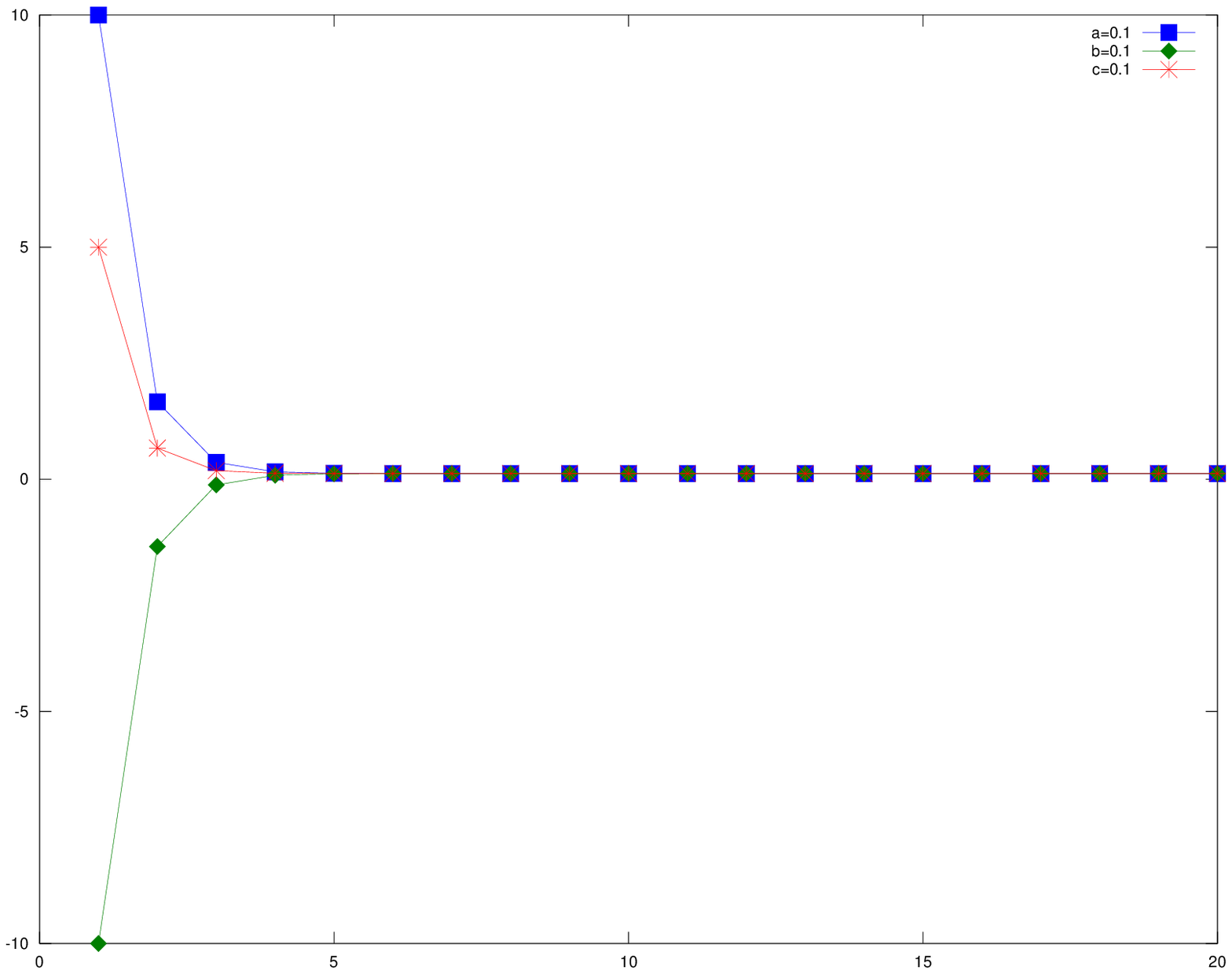}
        \includegraphics[scale=0.25]{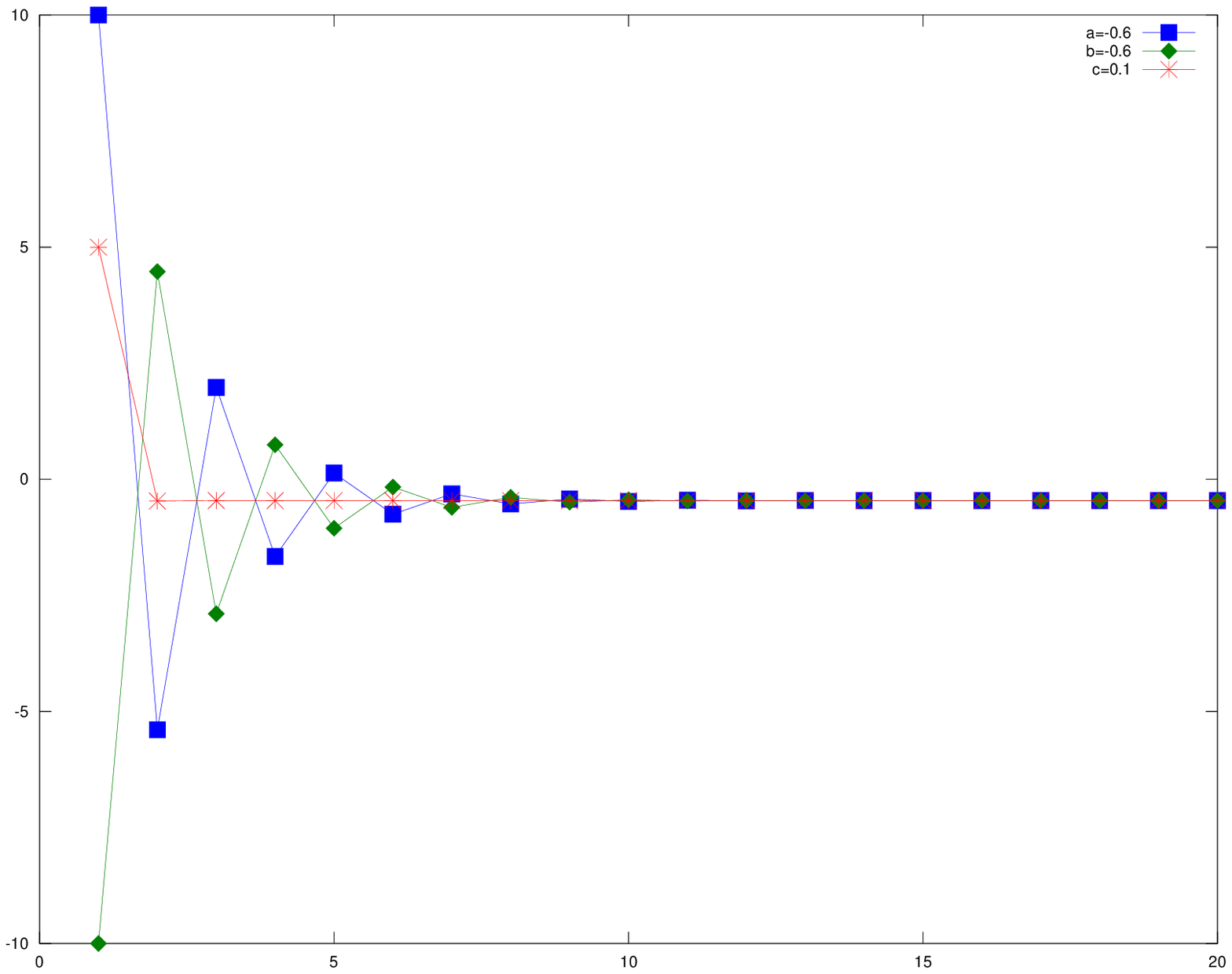}
        \includegraphics[scale=0.25]{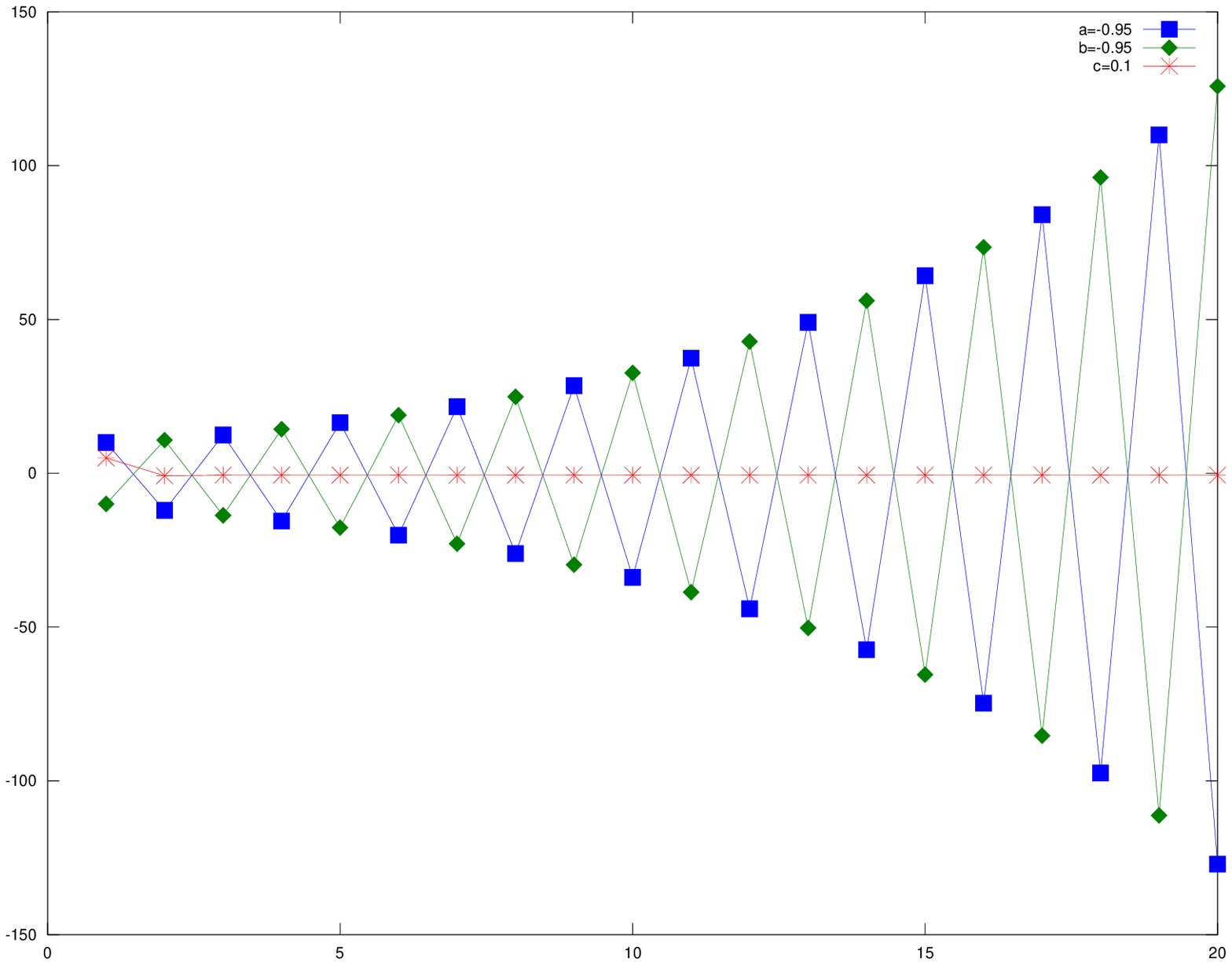}
        \caption{
        Belief dynamics for topic $X_{r+1}$ as sketched in 
        Example \ref{example:counter}. We use $\delta=5$ for the
        weight adjustment increment. Left:
        $a=b=0.1$. Middle: $a=b=-0.6$. Right: $a=b=-0.95$. Throughout
        $c=0.1$. 
        }      
        \label{fig:counter}
        \end{figure}
\end{example} 
Next, we discuss social influence of agents as a function of conformity
parameters and the structure of $\mathbf{W}^{(k)}$. In particular, we
show that even agents with an `empty record of past successes' can be
influential in the endogenous conformity model if conformity
parameters $\delta_i$ and matrix $\mathbf{Q}$ are appropriately
specified. 
\begin{proposition}\label{prop:influential}
  Let $\delta_i>0$ for all $i\in[n]$. 
  Assume that $\mathbf{Q}$ is strictly positive on all off-diagonal
  entries. Then, if $\mathbf{W}$ has at least two positive columns,
  $\mathbf{M}$ is strictly positive everywhere.  
\end{proposition}
\begin{proof}
  Again, $\mathbf{M}$ is 
  \begin{align*}
  \mathbf{M}
  = \mathbf{D}+(\mathbf{W}-\mathbf{D})\inv{(\mathbf{I}_n-\mathbf{\Delta}\mathbf{Q})}(\mathbf{I}_n-\mathbf{\Delta}). 
  \end{align*}
  We have
  $\mathbf{R}:=\inv{(\mathbf{I}_n-\mathbf{\Delta}\mathbf{Q})}=\sum_{r=0}^\infty
  (\mathbf{\Delta}\mathbf{Q})^r$ is strictly positive in each entry
  since $\mathbf{Q}$ is positive on all off-diagonals and since
  $\delta_i>0$, whence $(\mathbf{\Delta}\mathbf{Q})^1$ is positive on all
  off-diagonals and note that
  $(\mathbf{\Delta}\mathbf{Q})^0=\mathbf{I}_n$ has positive diagonals;
  the remaining terms $(\mathbf{\Delta}\mathbf{Q})^r$, for $r\ge 2$,
  are non-negative. 
  Hence, also
  $\tilde{\mathbf{R}}:=\mathbf{R}(\mathbf{I}_n-\mathbf{\Delta})>0$
  entrywise, since the 
  diagonals of $(\mathbf{I}_n-\mathbf{\Delta})$ are positive. Hence,
  since $\mathbf{W}$ has at least two positive columns,
  $(\mathbf{W}-\mathbf{D})\tilde{\mathbf{R}}$ is also positive 
  and then also
  $\mathbf{M}= \mathbf{D}+(\mathbf{W}-\mathbf{D})\tilde{\mathbf{R}}$. 
\end{proof}
\begin{remark}
  Thus, if $\mathbf{Q}$ is strictly positive in each entry (other than
  the diagonals) and all agents are (strictly) conforming and
  $\mathbf{W}$ has at least two positive columns, then $\mathbf{M}$ is
  strictly positive in each entry. This means that $\mathbf{M}^t$ is
  strictly positive in each entry, for all powers $t\ge 1$. This also
  means that $\lim_{t\goesto\infty}\mathbf{M}^t$, which exists
  for a row-stochastic $\mathbf{M}$ that is strongly connected and
  aperiodic (as an $\mathbf{M}$ with strictly positive entries is),
  is positive in each entry. But this means that, under conformity, if
  agents form their reference opinions $\mathbf{q}$, to which they
  strive to conform,  
  with respect to all other agents (that is, $\mathbf{Q}$ is strictly
  positive, except for the diagonals), then each agent $i$ has strictly
  positive social
  influence on the limiting beliefs, even if $i$ has never been
  truthful in the past. The amount of influence (even non-truthful)
  agents have on limiting beliefs depends then both on past
  performance \emph{and} on the conformity parameters $\delta_i$. We
  illustrate in the next example. 
\end{remark}
\begin{example}\label{example:socialInfluence}
  Assume the following situation. There are $n=3$ agents and
  $\mathbf{W}^{(k)}$ has the form
  \begin{align*}
  \mathbf{W}^{(k)} = 
  \begin{pmatrix}
  \frac{1}{2} & \frac{1}{2} & 0 \\
  \frac{1}{2} & \frac{1}{2} & 0 \\
  \frac{1}{2} & \frac{1}{2} & 0 \\
  \end{pmatrix}.
  \end{align*} 
  Such a $\mathbf{W}^{(k)}$ may arise, for instance, for large $k$,
  when $\tau=0$ and
  when agents $1$ and $2$ have initial beliefs centered around truth,
  with identical variances, 
  and agent $3$ has never been truthful, for instance, because
  $\text{Pr}[b_3^k(0)\in B_{k,\eta}]=0$, for all $k\ge 1$. Assume that
  \begin{align*}
  \mathbf{Q}^{(k)} = 
  \begin{pmatrix}
  0 & \frac{1}{2} & \frac{1}{2}\\
  \frac{1}{2} & 0 & \frac{1}{2}\\
  \frac{1}{2} & \frac{1}{2} & 0\\
  \end{pmatrix},
  \end{align*}
  which means that everyone weighs everyone else uniformly in
  forming reference opinions, 
  and assume that $\mathbf{Q}^{(k)}$ is constant across topics. Finally, let
  $\delta_1=\delta_2=a$ and $\delta_3=b$, where $0<a,b<1$. Consider
  the social influence 
  of agents --- that is, their
  influence on limiting beliefs as a function of initial
  beliefs. Since agents $1$ and $2$ are identical, they must have the
  same social influence, which we denote by $x\ge 0$, and let agent $3$
  have influence $y\ge 0$, such that $2x+y=1$. In the appendix, we
  show that $y$ has the form
  \begin{align*}
  y = \frac{a(1-b)}{4-ab-3a}. 
  \end{align*}
  Computing the comparative statics with respect to $a$ and $b$, we
  find
  \begin{align*}
  \pardiv{y}{a}
  = \frac{(1-b)\Bigl(4-ab+3b\Bigr)}{\Bigl(4-ab-3a\Bigr)^2}>0,\quad 
  \pardiv{y}{b}
  = \frac{a(a-7)}{\Bigl(4-ab-3a\Bigr)^2}<0.
  \end{align*}
  Thus, influence of agent $3$ decreases in own
  conformity, $b$, and increases in conformity of the stochastically
  intelligent agents, $a$. Moreover, we find
  \begin{align*}
  \lim_{a\goesto 1}y = 1,\quad \lim_{b\goesto 1}y = 0,
  \end{align*}
  such that agent $3$, who has zero probability of knowing truth, may
  have arbitrarily large social influence on limiting beliefs, as long as
  agents $1$ and $2$ exhibit arbitrarily large conformity, and agent
  $3$'s social influence may also vanish, as his own conformity
  becomes arbitrarily large. In Figure \ref{fig:socInf}, we plot $y$ as a
  function of $a$ (for fixed levels of $b$) and $b$ (for fixed levels
  of $a$). 
  \begin{figure*}[!ht]
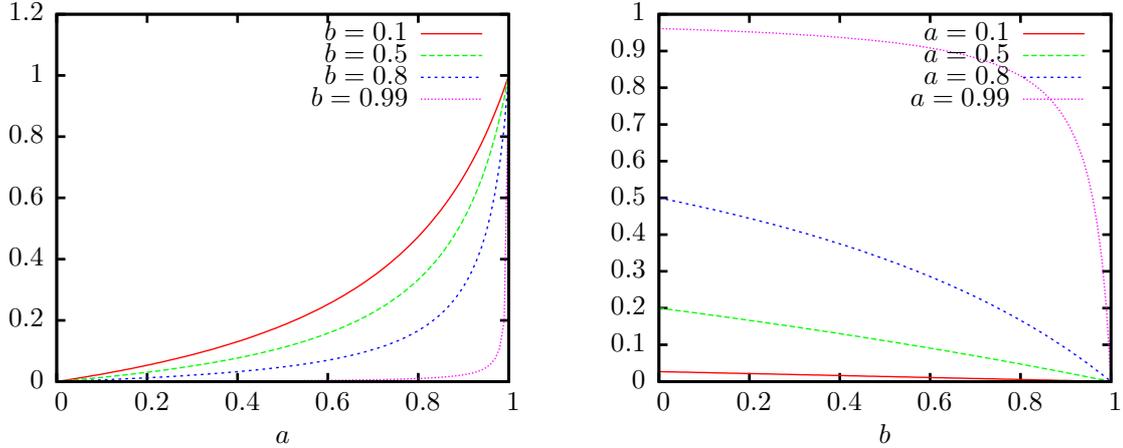

  \centering
  \begin{subfigure}{0.49\textwidth}
    \centering
    \resizebox{1.0\textwidth}{!}{
      \input{plots/socialInfluence.tex}}
    \end{subfigure}
    \begin{subfigure}{0.49\textwidth}
        \centering  
        \resizebox{1.0\textwidth}{!}{
        \input{plots/socialInfluence_yb.tex}}
    \end{subfigure}
    \caption
{ 
Social influence $y$ of agent $3$ as a function of $a$, conformity of
  agents $1$ and $2$, and $b$, own conformity. 
  }
\label{fig:socInf}
\end{figure*}
Note that this result, namely, that an agent with zero probability of 
being truthful may become arbitrarily influential, could not have been
possible in the standard model with endogenous weight formation as we
have sketched, since such an agent's social influence would always
be 
zero (or converge to zero) by the results developed in
Section \ref{sec:standard}, as $k$ becomes large. 
This
result would also not have 
been possible had there been only one positive column of matrix
$\mathbf{W}^{(k)}$ (and the others all zero) since in this case the corresponding agent, even if
he were excessively conforming and would thus state an opinion
arbitrarily close to that of his reference group, would both ignore
his own stated opinion (because he knows his true opinion) and that of
others (because all other columns are zero, by assumption) such that
the agent corresponding to the positive column would always have
social influence of $1$. In other words, we require at least two
conforming intelligent agents in order for a non-intelligent to be
able to become 
influential, under our current model specification. 

To summarize, our current example shows that if conformity parameters
are `adequately' specified, an agent with zero probability of being
close to 
truth may become arbitrarily influential, which, again, means that
society may be arbitrarily far off from truth in our setup. 
In particular, under conformity, society may be drawn away from truth
by agents without any past successes.\footnote{Which might be an
explanation of why
even `blatantly' and repetitively false propaganda may work.}
\end{example}

But, of course, a crucial aspect in the current example has also been
matrix $\mathbf{Q}$, which determines how agents form reference
opinions, and which agents a particular agent strives to conform to, 
and which we have assumed strictly positive. Of course, it might not
be implausible to assume that $\mathbf{Q}$ depends, in particular, on
past 
performance. This is our final example. 
\begin{example}
  Let
  $[\mathbf{Q}^{(k)}]_{ij}=\frac{[\mathbf{W}^{(k)}]_{ij}}{1-[\mathbf{W}^{(k)}]_{ij}}$
  for $i\neq j$ and let $[\mathbf{Q}^{(k)}]_{ii}=0$ such that
  $\mathbf{Q}$ is formed in an analogous way as
  $\mathbf{W}$.\footnote{As mentioned, this is the specification
  discussed in Buechel, Hellmann, and Kl\"{o}{\ss}ner
  (2013) \cite{Buchel2013}.} In 
  particular, in this setup, agents want to conform to those who have
  performed well in the past. Assume that there are two types of
  agents, whereby one type has initial beliefs centered around truth
  as in \eqref{eq:centered} and the other type has zero probability of
  ever 
  being truthful, 
  $\text{Pr}[b_i^k(0)\in B_{k,\eta}]=0$ for all $k$ and all
  $i\in \struct{N}_2$ where $\struct{N}_2\subseteq[n]$ is the set of
  agents of type two, and by $\struct{N}_1$ we denote the set of agents of
  type one. Then, if $\tau=0$ and as $k$ becomes large,
  $\mathbf{W}^{(k)}$ has the 
  structure where in each row $i$, 
  $W_{ij}^{(k)}\approx \frac{1}{C\sigma_j^2}$ for
  $j\in\struct{N}_1$ ($C$ is a normalization constant) and
  $W_{ij}^{(k)}\approx 0$ for $j\in\struct{N}_2$, by the results developed in
  Section \ref{sec:standard}, and where $\sigma_j^2$ is the variance
  of agent $j$'s initial belief. The  
  informational social influence of agent $i$ is then also given by,
  roughly, 
  $w_i=\frac{1}{C\sigma_i^2}$, for $i\in\struct{N}_1$, and $w_i=0$,
  for $i\in\struct{N}_2$, respectively. Moreover, the
  combined informational as well as normative social influence of
  agents is given by
  \begin{align*}
  v_i = \frac{(1-\delta_i)\cdot w_i}{\sum_{j\in\struct{N}_1}(1-\delta_j)w_j}
  \end{align*}
  for agents $i\in\struct{N}_1$ and $v_i=0$ for agents
  $i\in\struct{N}_2$. These results follow directly from Theorem 1 and
  Corollary 1 given in Buechel, Hellmann, and Kl\"{o}{\ss}ner
  (2013) \cite{Buchel2013}, which precisely state that closed and
  strongly connected groups (which $\struct{N}_1$ is, at least for
  large $k$) have social
  influence $v_i$ as given and the `rest of the world', which group
  $\struct{N}_2$ forms, has $v_i=0$.

  This example shows that if $\mathbf{Q}$ follows the structure of
  $\mathbf{W}$, then, unlike in the previous example where
  $\mathbf{Q}$ was uniform (or at least strictly positive on the
  off-diagonals), agents 
  that 
  never know truth cannot be influential. It moreover shows that
  social influence decreases in conformity (for the agents in
  $\struct{N}_1$), as we have already observed in the previous
  example.  

  Investigating social influence in the general case, for arbitrary
  $\mathbf{Q}, \mathbf{W}$, and $\mathbf{\Delta}$, would be highly
  interesting, as it indicates to which degree agents who are
  never truthful can still be influential,
  and 
  scope for future work.\footnote{Proposition \ref{prop:influential}
  is an important step in this direction 
  already, since it says that if $\mathbf{Q}^{(k)}$ is strictly positive on
  all off-diagonals, $\mathbf{\Delta}^{(k)}$ is strictly positive in all
  diagonals, and $\mathbf{W}^{(k)}$ has two positive columns,
  then \emph{all} agents are influential.}
\end{example}

\section{Homophily}\label{sec:homophily}

In this section, we extend the standard endogenous weight adjustment
opinion dynamics model discussed in Section \ref{sec:standard} by
introduction of the concept of \emph{homophily}, according to which, as
McPherson, Smith-Lovin, and Cook (2001) \cite{McPherson2001} phrase
it, `similarity breeds connection', 
and which is a majorly accepted standard concept in modern
socio-economic research. In the opinion dynamics
literature, homophily has been modeled 
by positing 
that weights
(social ties)
between any two agents are functionally dependent on the agents' current
belief 
distance (cf.\ the Hegselmann and Krause models, Deffuant, Neau,
Amblard, and Weisbuch (2000) \cite{Deffuant2000}, Pan (2010) 
\cite{Pan2010}, etc.), that is, agents with more similar current beliefs place
greater (current) weight upon each other. Opinion updating is then
performed as in 
standard DeGroot 
learning, 
via weighted averages of peers' past beliefs,
where the weights are now 
endogenously 
formed by the 
homophily principle. As indicated in the introduction, we think of
homophily, in our context, as arising from biased reasoning,
where individuals overrate beliefs that are similar to their own (cf.\ Kunda
(1990) \cite{Kunda1990}).  

In the Hegselmann and Krause models, to which we
relate, agents set \emph{time-varying} weights according to the
following 
rule,\footnote{Note that, thus far, we have assumed weights to be only
varying across \emph{topics} and not in addition across
discussion \emph{rounds} 
(time) within a given topic.}  
\begin{align*}
  W_{ij}(t) &= 
  \begin{cases}
  \frac{1}{I_i(\mathbf{b}(t))} & \text{if } j\in I_i(\mathbf{b}(t)),\\ 
  0 & \text{else},
  \end{cases}
\end{align*}
where $I_i(\mathbf{b}(t))$ denotes the set of agents within an
$\eta_H$-radius, for $\eta_H\ge 0$, around agent $i$'s belief
$b_i(t)$ at time $t$, that is, $I_i(\mathbf{b}(t))=\set{j\in[n]\sd
  \norm{b_j(t)-b_i(t)}<\eta_H}$, and where $\eta_H$ is an
external parameter. One plausible integration of this setup in our
framework is to let agents \emph{increment} weights to other agents whenever
distance between their current beliefs is sufficiently small, that
is,{\footnote{Otherwise, if weights were not incremented but rather set
  in an `absolute manner', the continuity of weight relationships
  across topics could not be maintained.}$^,$\footnote{
  After each round $t$, we \emph{renormalize} weights in order for
  them to satisfy the row-stochasticity condition.}}
\begin{align}\label{eq:timevarying}
  W_{ij}^{(k)}(t+1) &= 
  \begin{cases}
    W_{ij}^{(k)}(t)+\delta_H & \text{if } j\in I_i(\mathbf{b}(t)),\\
    W_{ij}^{(k)}(t).
  \end{cases}
\end{align}
Then, after truth is revealed, agents again adjust weights according
to the `truth related' principles outlined in Section \ref{sec:modelECCS}. In particular,
we would now have,
\begin{align}\label{eq:adjusttruth2}
      W_{ij}^{(k+1)}(0) = 
      \begin{cases}
        \lim_{t\goesto\infty}{W}_{ij}^{(k)}(t)+\delta_T\cdot T(\abs{N({\mathbf{b}^k(\tau)},\mu_k)}) & \text{if }
        \norm{b_{j}^k(\tau)-\mu_k}< \eta_T,\\
        \lim_{t\goesto\infty}{W}_{ij}^{(k)}(t)& \text{otherwise},
      \end{cases}
    \end{align}
    for all $k\ge 1$, where we need to consider, for next topic's
    initial weights, the limit, as time goes to infinity, of the 
    time-varying weights for the
    previous topic, since weights are now also adjusted within topic
    periods. Note that, here, we also subscript $\eta$ --- the radius
    within which weights are adjusted --- and $\delta$ --- the weight
    increment --- by $T$ and $H$, respectively, depending on whether
      we relate to 
    adjusting/incrementing based on truth or based on homophily. 

    Also observe that \eqref{eq:adjusttruth2} is
    well-defined only if $\lim_{t\goesto\infty}W_{ij}^{(k)}(t)$
    exists, which we na\"{i}vely assume in the following but 
    the formal proof of which we leave open. 
    It is worthwhile mentioning that adjusting weights based on truth may
    microeconomically be
    justified precisely as we did in Section \ref{sec:justification}
    --- namely, it may follow from the tenet that agents have
    disutility from not knowing truth, whence, 
    by incrementing weights to agents who have been truthful in
    the past, they increase their likelihood of eventually becoming close to
    truth, provided that the assumptions they make (\emph{bona fides}, etc.)
    are satisfied. In contrast, we offer no explicit microeconomic foundation
    --- that is, based on utility functions and their explicit
    maximization --- 
    here of why agents would increment weights to other agents based
    on the homophily relation, 
      taking this behavior simply as a form of (exogenous) bias. 
      Moreover, similar as in the opposition model, it is
    appropriate, in the current setting, to think of agents as
    motivated by two contrarian 
    forces --- truth and homophily ---\footnote{Note that in the
      opposition model, the two forces were an agent's ingroup and his
    outgroup.} which may possibly act to the `same ends', but which we
    generally think of as of antipodal origin and direction. 

Concerning updating of beliefs, beliefs evolve according to
\begin{align}\label{eq:updatehomophily}
  b_{i}^k(t+1) = \sum_{j=1}^n W_{ij}^{(k)}(t)b_{j}^k(t),
\end{align}
as outlined in Section \ref{sec:modelECCS}, with the addition that we
now let weights $W_{ij}^{(k)}$ vary within discussion periods. Due to
the slightly greater complexity involved in belief dynamics, we
summarize the belief evolution process in the below schematic
form. 
\begin{algorithm}
    \begin{algorithmic}[1]
      \State let $\mathbf{W}^{(1)}(0)$ be (exogenously) given
      \For{$k=1,2,3,\ldots$}
        \State 
        let $b_i^k(0)$ denote initial beliefs for topic $X_k$ for all agents
    $i=1,\ldots,n$ 
        \For{$t=0,1,2,3,4,\ldots$}
         \State  adjust weights $\mathbf{W}^{(k)}(t)$ based on homophily,
    Eq. \eqref{eq:timevarying}; normalize weights
          \State update beliefs $b_{i}^k(t+1)$ for all agents $i$ via
    Eq. \eqref{eq:updatehomophily} 
      \EndFor
      \State adjust weights $\mathbf{W}^{(k+1)}(0)$ based on truth,
    Eq. \eqref{eq:adjusttruth2}; normalize weights
      \EndFor
    \end{algorithmic}
  \end{algorithm}

Providing general results for the belief dynamics
process currently under consideration is not so easy since weights do
not only vary by time now, but, in particular, by the current belief state
vector $\mathbf{b}^{(k)}(t)$. Lorenz (2005) \cite{Lorenz2005} gives
convergence results for this general setup, which we list in 
Appendix \ref{sec:appendix}, 
whose assumptions, however, do not apply to our situation. As a first
illustration, still, we show that, unlike in the standard model and
under conformity, even after some agent has been truthful for a topic
$X_r$, agents need not converge to a consensus, but may hold distinct
limiting beliefs about $X_{r+1}$; whether consensus obtains or not may
depend on the relative sizes of $\delta_T$, which we may think of as
`importance of truth', versus $\delta_H$, which we may think of as
`importance of homophily'. 
\begin{example}\label{example:homExample}
  Let there be $n=4$ agents. Assume that $\mathbf{W}^{(1)}(0)$ is the
  $n\times n$ identity matrix. Moreover, let $\tau=0$ and assume that
  agent $1$ receives initial signal $b^1_1(0)=\mu_1$
  and that agents $2$, $3$ and $4$ are not within an $\eta_T$-radius around
  truth, for topic $X_1$. Finally, assume that any two agents' initial
  beliefs $b_i^1(0)$ and $b_j^1(0)$ are at least
  a distance of $\eta_H$ away from each other for topic $X_1$ so
  that homophily plays no role for topic $X_1$. Then, at the beginning
  of topic $X_2$, agents adjust weights based on truth such that
  $\mathbf{W}^{(2)}(0)$ looks as follows
  \begin{align*}
    \mathbf{W}^{(2)}(0) = \frac{1}{1+\delta_T}\begin{pmatrix}
      1+\delta_T & 0 & 0 & 0 \\
      \delta_T & 1 & 0 & 0\\
      \delta_T & 0 & 1 & 0\\
      \delta_T & 0 & 0 & 1
      \end{pmatrix}.
  \end{align*}
  Assume that initial beliefs of agents for topic $X_2$ are
  $b_1^2(0)=b_2^2(0)$ and $b_3^2(0)=b_4^2(0)$, whereby initial beliefs of
  agents $1$ and $2$, on the one hand, and $3$ and $4$, on the other
  hand, are at a distance of at least $\eta_H$. 
  This specification means that agents $1$ and $2$, on the one hand,
  and $3$ and $4$, on the other, form distinct `homophily clusters',
  at least at time $t=0$, 
  for topic $X_2$. 
  In Figure
  \ref{fig:homExample}, we sketch belief dynamics for topic $X_2$ for
  different values of $\delta_H$, with $\delta_T=0.1$ and
  $\eta_H=\eta_T=0.2$ fixed. We see that, unlike in the standard
  DeGroot learning case in this setup (and also in the conformity
  model) and as already indicated, agents do not necessarily reach a
  consensus. If homophily is `too strong', that is, $\delta_H$ is
  `too large', agents polarize in this setting. As homophily becomes
  weaker, that is, $\delta_H$ becomes smaller, the beliefs of agents
  $3$ and $4$ move closer to the beliefs of agents $1$ and $2$, the
  former of which has been truthful for topic $X_1$. As $\delta_H$
  falls below a certain threshold, the agents reach a consensus. 
  \begin{figure}[!htb]
    \centering
    \includegraphics[scale=0.4]{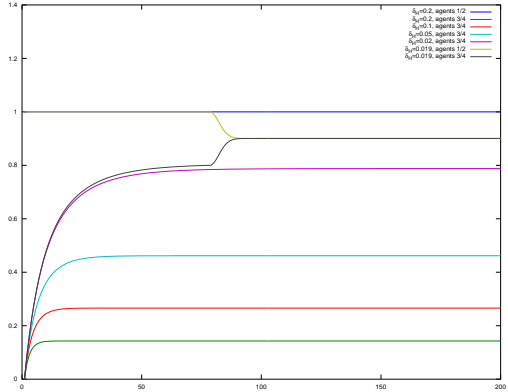}
    \caption{
      Belief dynamics for topic $X_2$ with setup as sketched in Example
      \ref{example:homExample}. Initial beliefs are
      $b_1^2(0)=b_2^2(0)=1$, $b_3^2(0)=b_4^2(0)=0$. As long as beliefs
      of agents $3$ and $4$ do not come within distance of
      $\eta_H=0.2$ to the beliefs of agents $1$ and $2$, the latters'
      beliefs evolve according to $b_1^2(t)=b_2^2(t)=1$ since both
      agents have the same initial beliefs and are not `disturbed' by
      agents $3$ and $4$. In contrast, beliefs of agents $3$ and $4$
      are affected by agent $1$'s beliefs since agent $1$ has been
      truthful for topic $X_1$, but their weight link to this agent
      vanishes as $t$ becomes large if $\delta_H$ is `large
      enough'. In general, we have $b_1^2(t)=b_2^2(t)$ and
      $b_3^2(t)=b_4^2(t)$ for all $t\ge 0$ so that it suffices to
      graph belief dynamics of agents $1$, on the one hand, and $3$,
      on the other.
    }
    \label{fig:homExample}
  \end{figure}
\end{example}
To sketch one (simple) result of a general nature, here, however,
consider, similarly as before, the situation when there are two groups
$\struct{N}_1$ and $\struct{N}_2$ 
of agents, whereby 
agents in $\struct{N}_1$ are $\epsilon$-intelligent, 
for a fixed $\epsilon\ge 0$,
and agents in
$\struct{N}_2$ have initial beliefs $b_i^k(0)$ with
$\text{Pr}[b_i^k(0)\in B_{k,\eta_T}]=0$, that is, agents in
$\struct{N}_2$ have initial beliefs
such that the probability that they `correspond to' truth is zero. In
the next 
proposition, we show that \emph{all} agents \emph{may} become
$\epsilon$-wise, even if $\delta_H>0$, 
in this situation provided that agents value truth `sufficiently
much' \emph{and} value relations based on homophily sufficiently
little. This result is not entirely trivial because, for instance, for 
our opposition model, arbitrarily small `opposition force' could
induce (at least some) agents to not converge to truth.  
The result says that homophily does not \emph{always} need to
interfere with wisdom. 
\begin{proposition}\label{prop:deltaT}
  Let $\eta_T\ge 0$, $\epsilon\in [0,\eta_T]$, and topic $X_k$ be
  fixed, for $k\ge 2$.\footnote{For topic $X_1$, there are 
    no weight adjustments based on truth, so we exclude this situation.}
  Then there 
  exist $\delta_T>0$ large enough and $\delta_H>0$ small enough such
  that all agents become 
  $\epsilon$-wise for $X_k$.\footnote{Of course, if $\delta_H$ were
    zero, any positive $\delta_T$ would satisfy the conditions of the
    proposition. In our setup, we assume, however, that homophily
    always plays a role, that is, $\delta_H>0$ for all `homophily
    increments' $\delta_H$.}
\end{proposition}
\begin{proof}
  Since the agents in $\struct{N}_1$ are $\epsilon$-intelligent, with
  $\epsilon\le \eta_T$, all agents adjust weights for these agents based on
  truth. By choosing $\delta_T$ large enough 
  and $\delta_H$ small enough (but positive), it can be ensured that
  beliefs $[\mathbf{b}^k(1)]_i$, 
  for all $i\in[n]$, are
  in the (open) $\epsilon$-interval around truth $\mu_k$ (the
  weights for the $\epsilon$-intelligent agents may become arbitrarily
  close to uniform provided $\delta_T$ is large enough and $\delta_H$
  is small enough and the weights
  for the agents in $\struct{N}_2$ may become arbitrarily close to
  zero). 
  Since 
  $\mathbf{W}^{(k)}(t)$ is row-stochastic, for every $t\ge 0$, and
  since the (open) 
  interval of radius $\epsilon$ around truth is a convex set, all
  belief vectors $\mathbf{b}^k(t)$, for $t\ge 1$, lie, component-wise,
  in $B_{k,\epsilon}$.  
\end{proof}
\begin{remark}
In the last proposition, $\delta_T=\delta_T(k)$ and
$\delta_H=\delta_H(k)$ may depend upon the
topic $X_k$ 
(e.g., in particular on the distribution of initial beliefs for this
topic). 
If we let, $\delta_T:=\max_{k\in\nn}\delta_T(k)$ and
$\delta_H:=\min_{k\in\nn}\delta_H(k)$,\footnote{This would require to
  ensure that the so defined $\delta_H$ is strictly positive (rather than
  zero) and that $\delta_T<\infty$.} then for 
this choice of $\delta_T$ and $\delta_H$, all agents will be 
$\epsilon$-wise for all topics $X_k$, for $k\ge 2$. 
\end{remark}
\begin{example}\label{example:deltaT}
We illustrate Proposition \ref{prop:deltaT} in Figure
\ref{fig:deltaT}, where we sketch belief dynamics for a sequence of
topics for fixed parametrizations and various choices of $\delta_T$
and $\delta_H$. In the figure, we simulate belief dynamics across
topics for $n=50$ agents, where $n_1=10$ agents are
$\epsilon$-intelligent and $n_2=40$ agents have 
initial 
distribution of beliefs such that 
their initial beliefs 
are never in an $\eta_T$
interval around truth; for the sake of concreteness, we let $\mu_k=0$,
for all $k\ge 1$, 
$\epsilon=0.25$, $\eta_T=0.25$, and for the agents in $\struct{N}_2$, we let
their initial beliefs be distributed according to the random uniform
distribution on the interval $[1,4]$. 

The graphs illustrate, first, that truth attracts all agents since
even the
beliefs of 
the agents in $\struct{N}_2$ move in the direction of truth, as time
progresses. However, as long as preference for homophily $\delta_H$ is
not small enough and preference for truth $\delta_T$ is not large
enough, the agents in $\struct{N}_2$ do not become $\epsilon$-wise for
topics. The graphs also illustrate the clustering of beliefs due to
the homophily relationship, a circumstance well-known from the
classical Hegselmann-Krause models. 
\begin{figure}[!htb]
  \centering
  \includegraphics[scale=0.35]{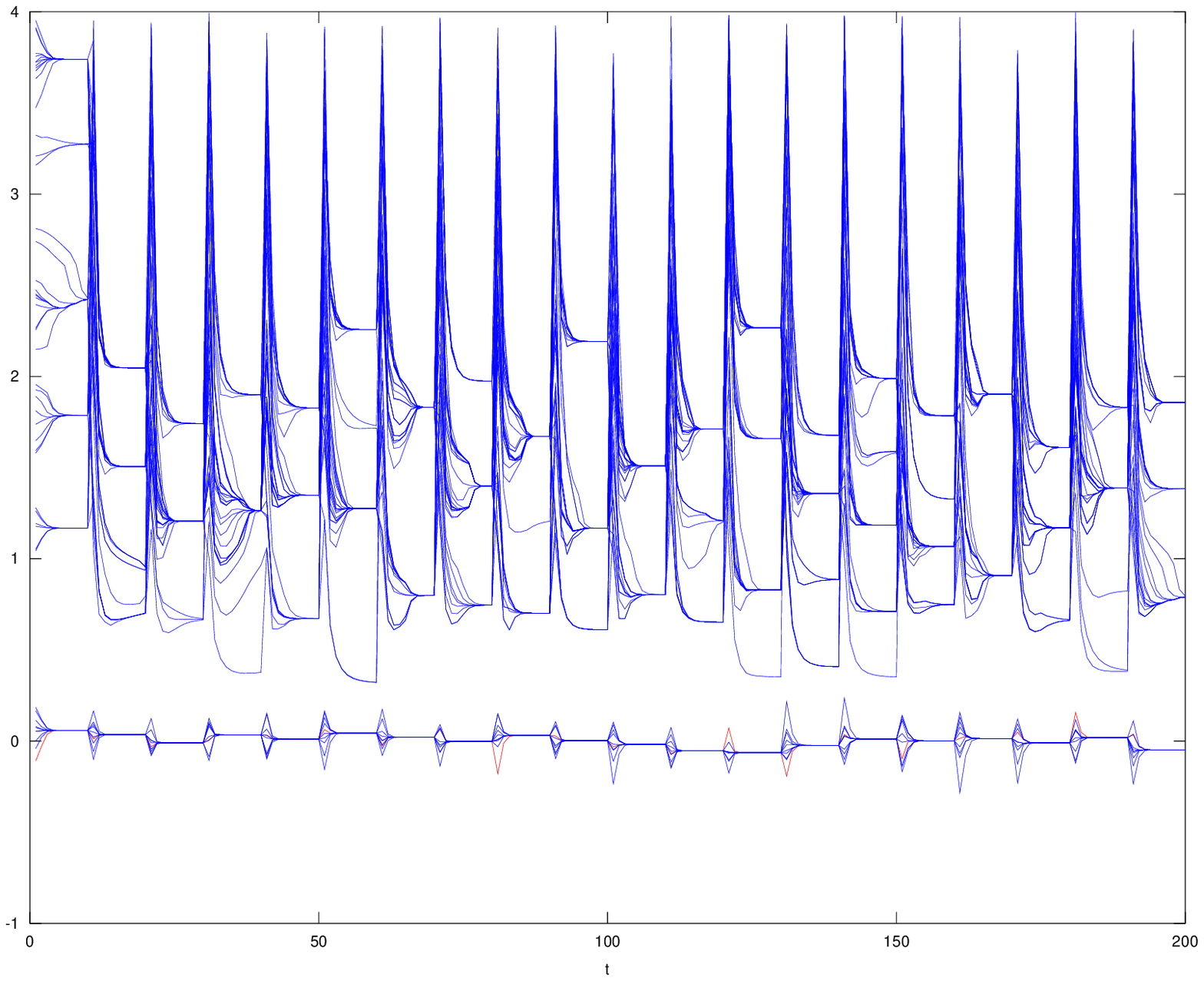}
  \includegraphics[scale=0.35]{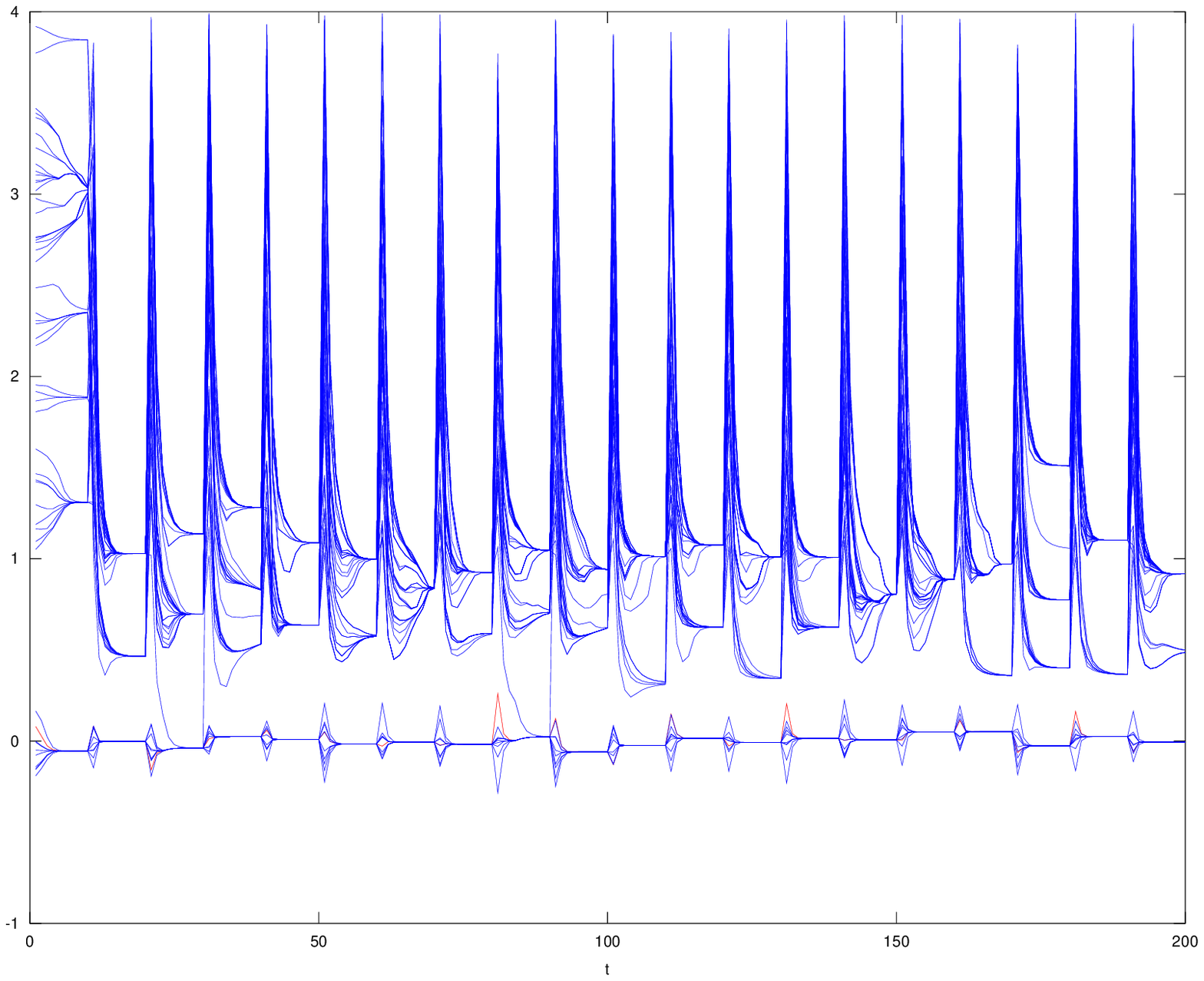}
  \includegraphics[scale=0.35]{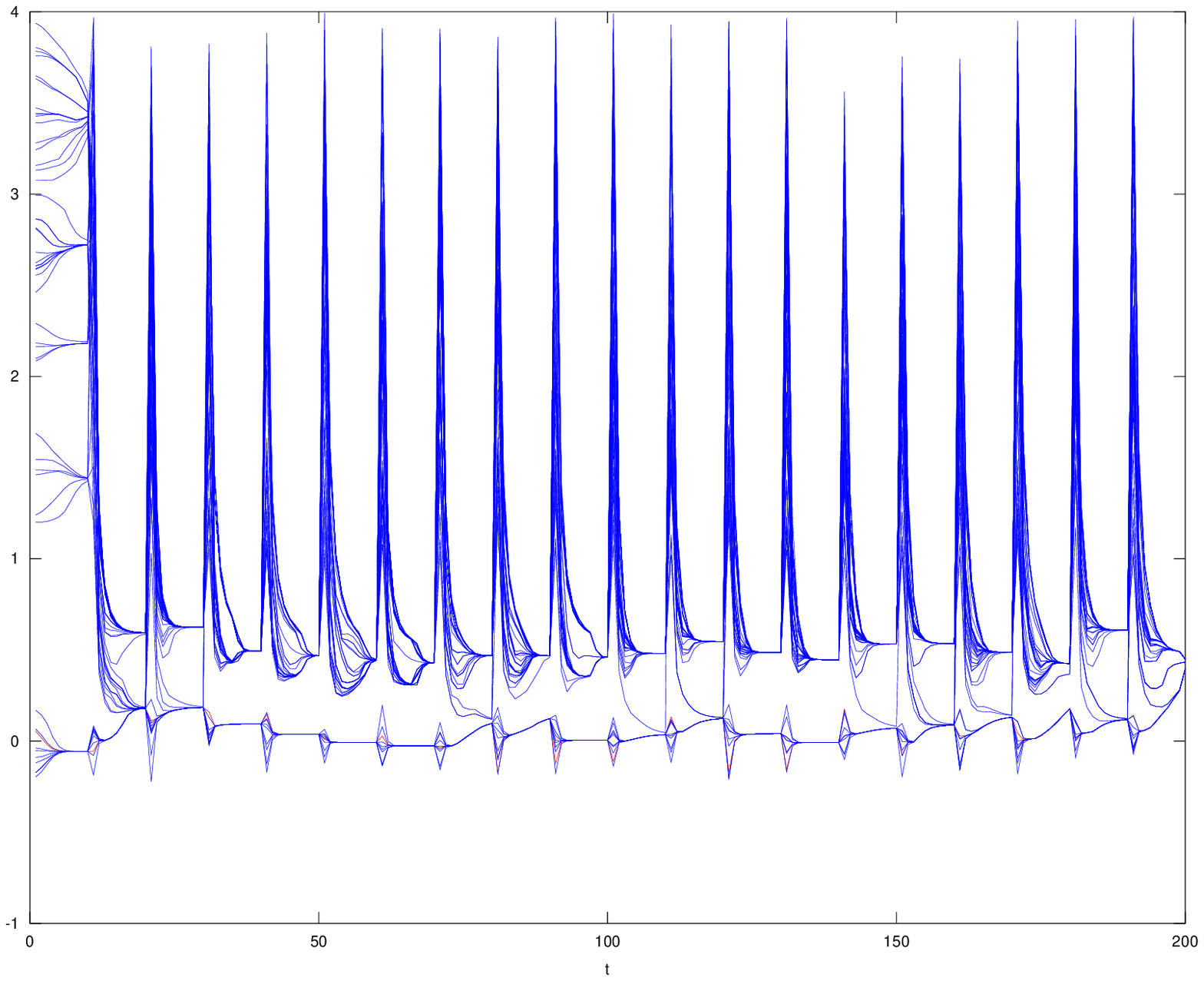}
  \includegraphics[scale=0.35]{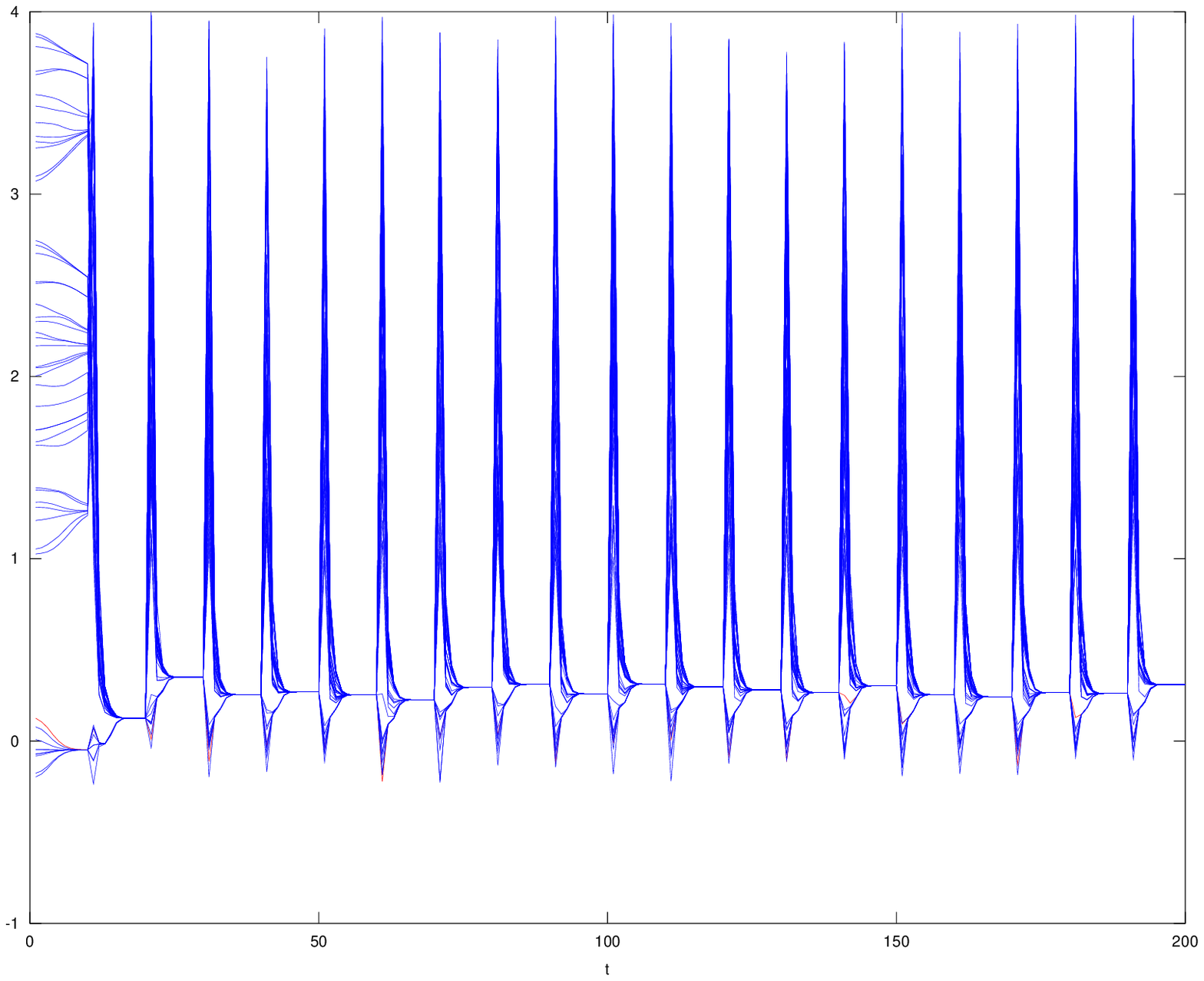}
  \includegraphics[scale=0.35]{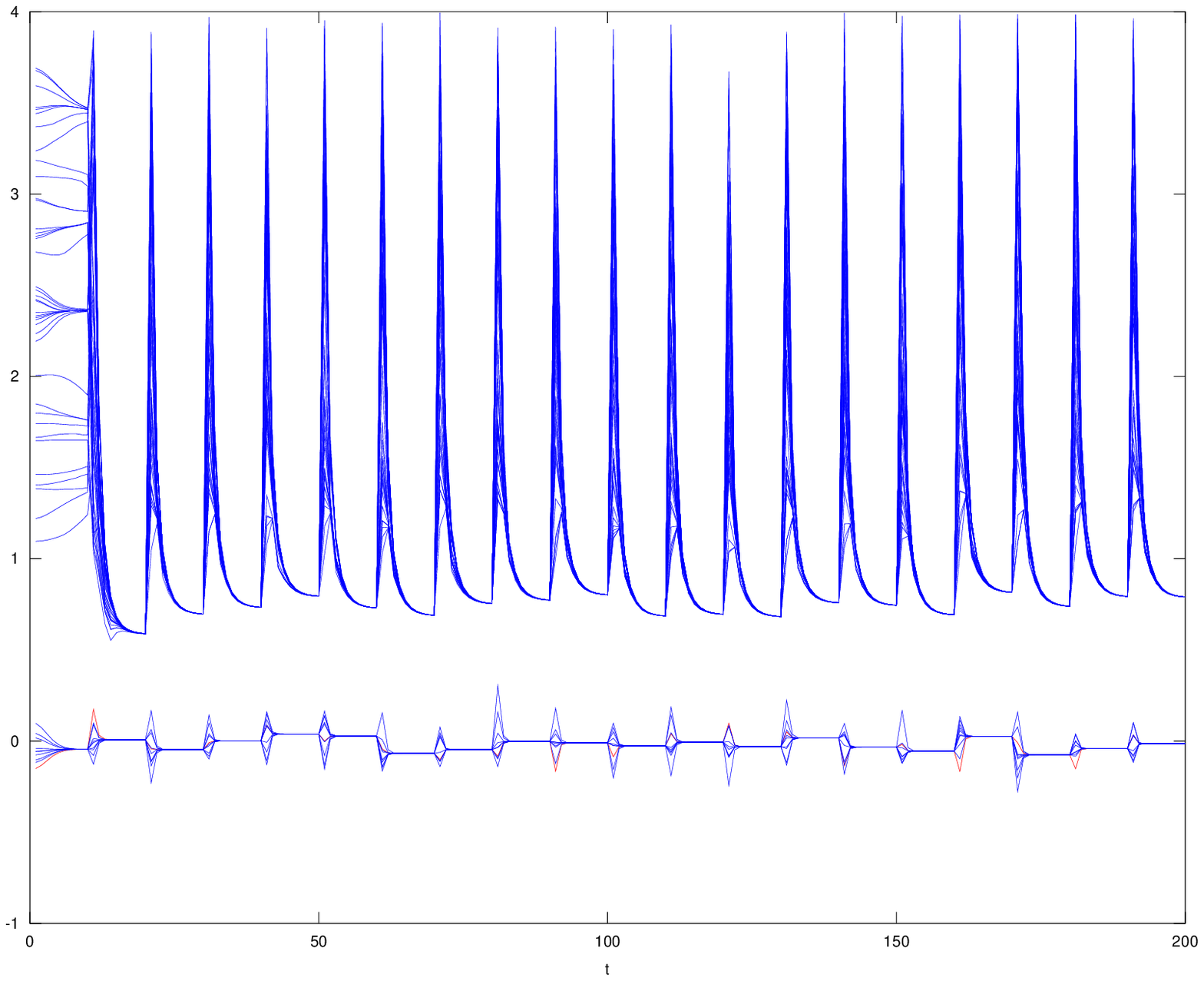}
  \caption{Parametrizations throughout: $\epsilon=0.25$, $\mu_k=0$ for
    all $k\ge 1$, $\eta_T=\eta_H=0.25$, $n=50$ agents,
    $\length{\struct{N}_1}=10$, $\length{\struct{N}_2}=40$. From top
    to bottom and left to 
  right: $(\delta_H,\delta_T)=(0.2,1.0)$;
  $(\delta_H,\delta_T)=(0.1,1.0)$; $(\delta_H,\delta_T)=(0.05,1.0)$;
  $(\delta_H,\delta_T)=(0.02,1.0)$;
  $(\delta_H,\delta_T)=(0.02,0.1)$. We show topics $X_k$, for
  $k=1,2,\ldots,20$ and, for each topic, discussion rounds
  $t=0,1,\ldots,10$.} 
  \label{fig:deltaT}
\end{figure}
\end{example}
\begin{remark}
  The graphs in Figure \ref{fig:deltaT} show much analogy with results
  of the
  original opinion dynamics model `under homophily' as developed in
  the work of Hegselmann and Krause, on which our current modeling is
  based (recall that the difference is that we \emph{increment} weights in
  case two agents' beliefs are similar, while they set weights uniformly
  in this case, and that we in addition introduce truth as an
  influential factor). In particular, in the graphs, we find that
  \begin{itemize}
    \item the opinion dynamics process always converges, and that 
    \item agents (or, rather, their beliefs) cluster into subgroups
      in which agents reach a consensus. 
  \end{itemize}
  Proving these apparently generally true observations is beyond the
  scope of our investigation here, and we leave it for future
  consideration.\footnote{See Krause (2000) \cite{Krause2000} for a
    starting point on how to prove the results in question.}
\end{remark}

  We close this section by presenting simulations on the role of the
  truth related radius $\eta_T$ and the homophily related radius
  $\eta_H$, respectively. Concerning $\eta_H$, we find in Figure
  \ref{fig:etaH} that a smaller $\eta_H$ (that is, based on homophily, agents
  trust/listen to only those with very similar beliefs) tends to
  produce a larger 
  degree of fragmentation of limiting belief spectra while larger
  $\eta_H$ (that is, based on homophily, agents even trust/listen to agents
  with rather 
  distinct beliefs) tends to promote global agreement among agents. 
  Interestingly, smaller $\eta_H$ also leads agents closer to truth
  (since the homophily relation applies to fewer agents). 
  Concerning
  $\eta_T$, in Figure \ref{fig:etaT}, we find that, overall, 
  an increase in $\eta_T$ increases the average distance of
  limiting beliefs to 
  truth since also beliefs that are remote from truth are taken into
  consideration in link weight adjustment. 
\begin{figure}[!htb]
  \centering
  \includegraphics[scale=0.35]{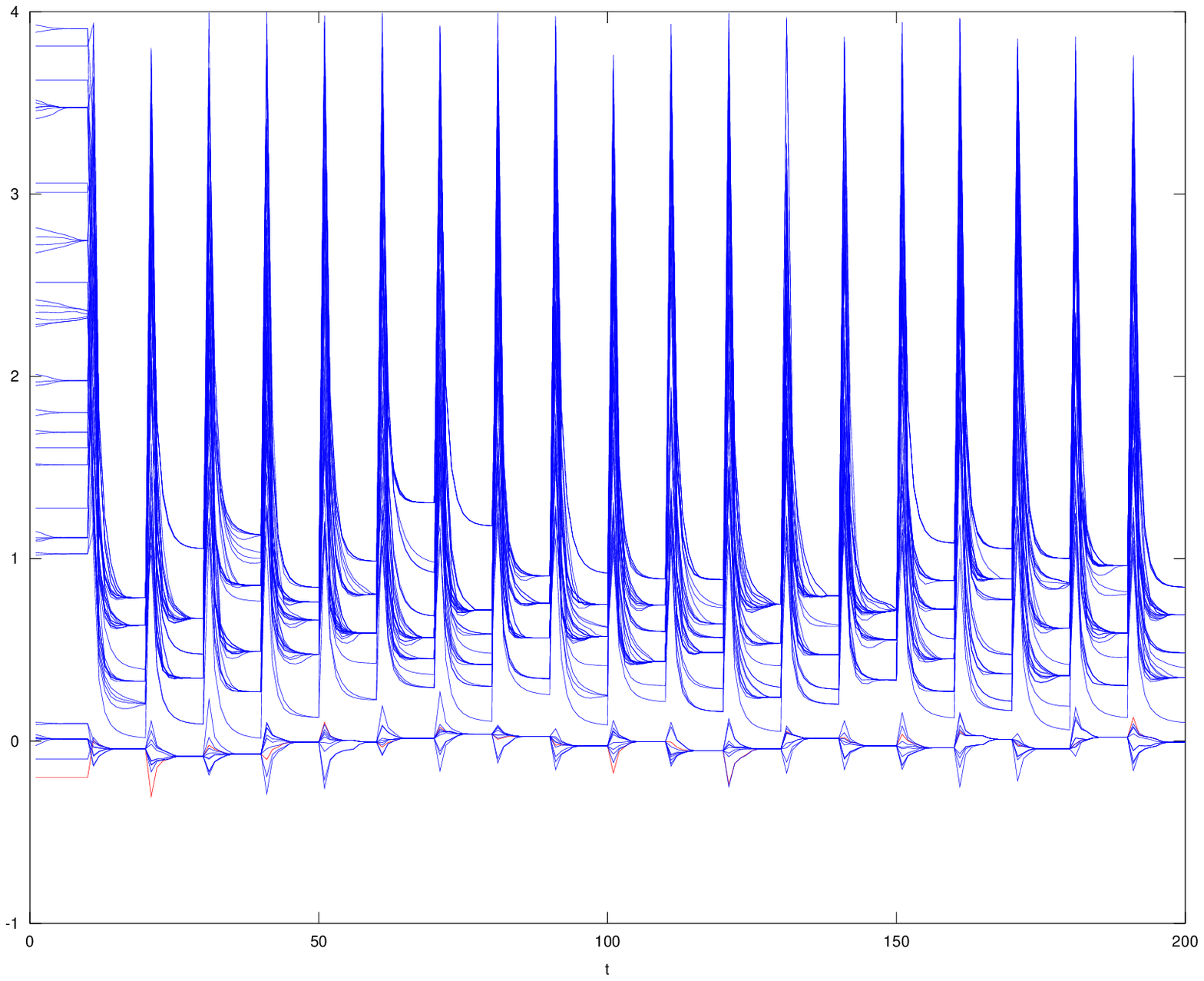}
  \includegraphics[scale=0.35]{plots/homophily/homEvolution_0p2_1p0_0p25_0p25.eps}
  \includegraphics[scale=0.35]{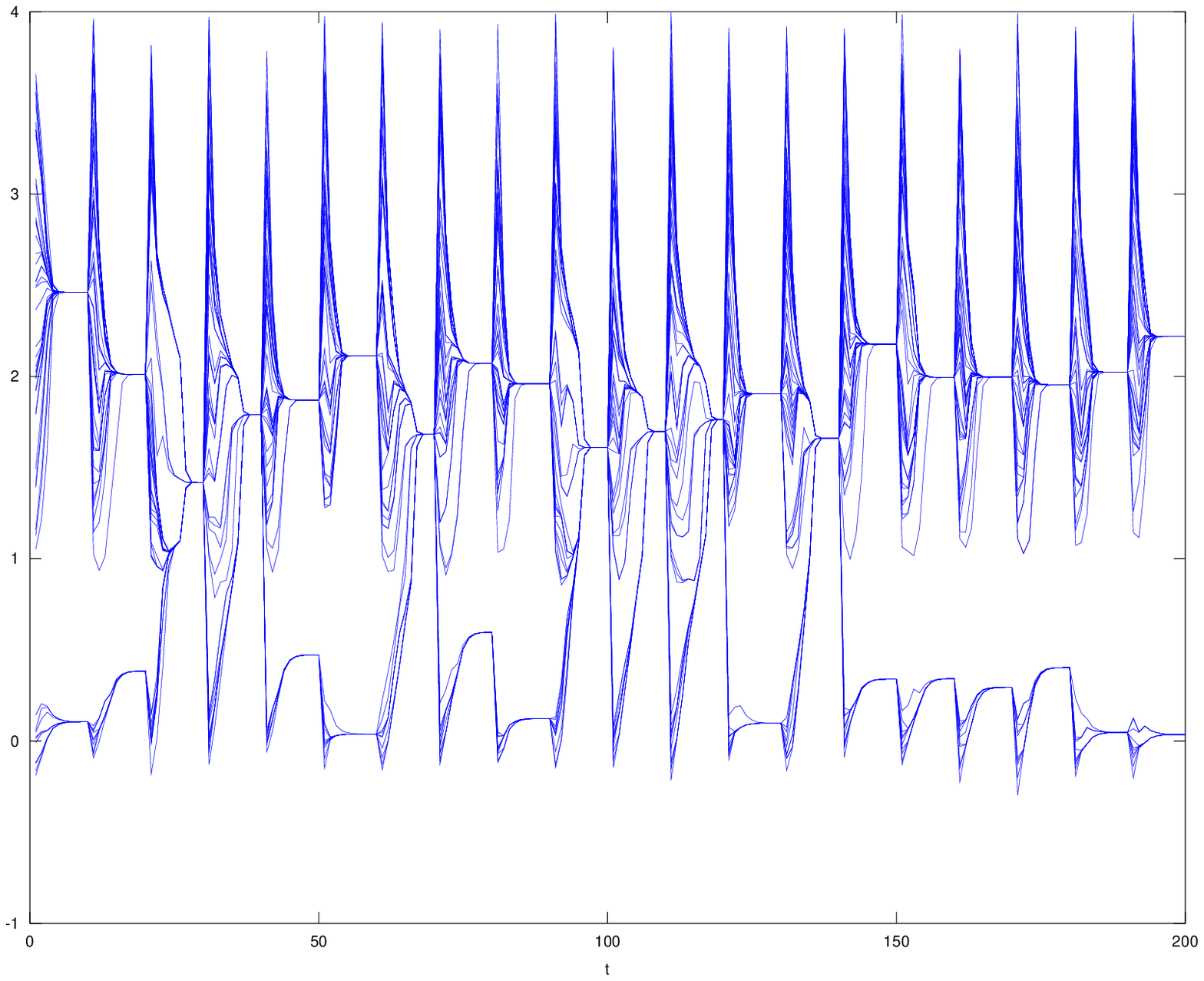} 
  \includegraphics[scale=0.35]{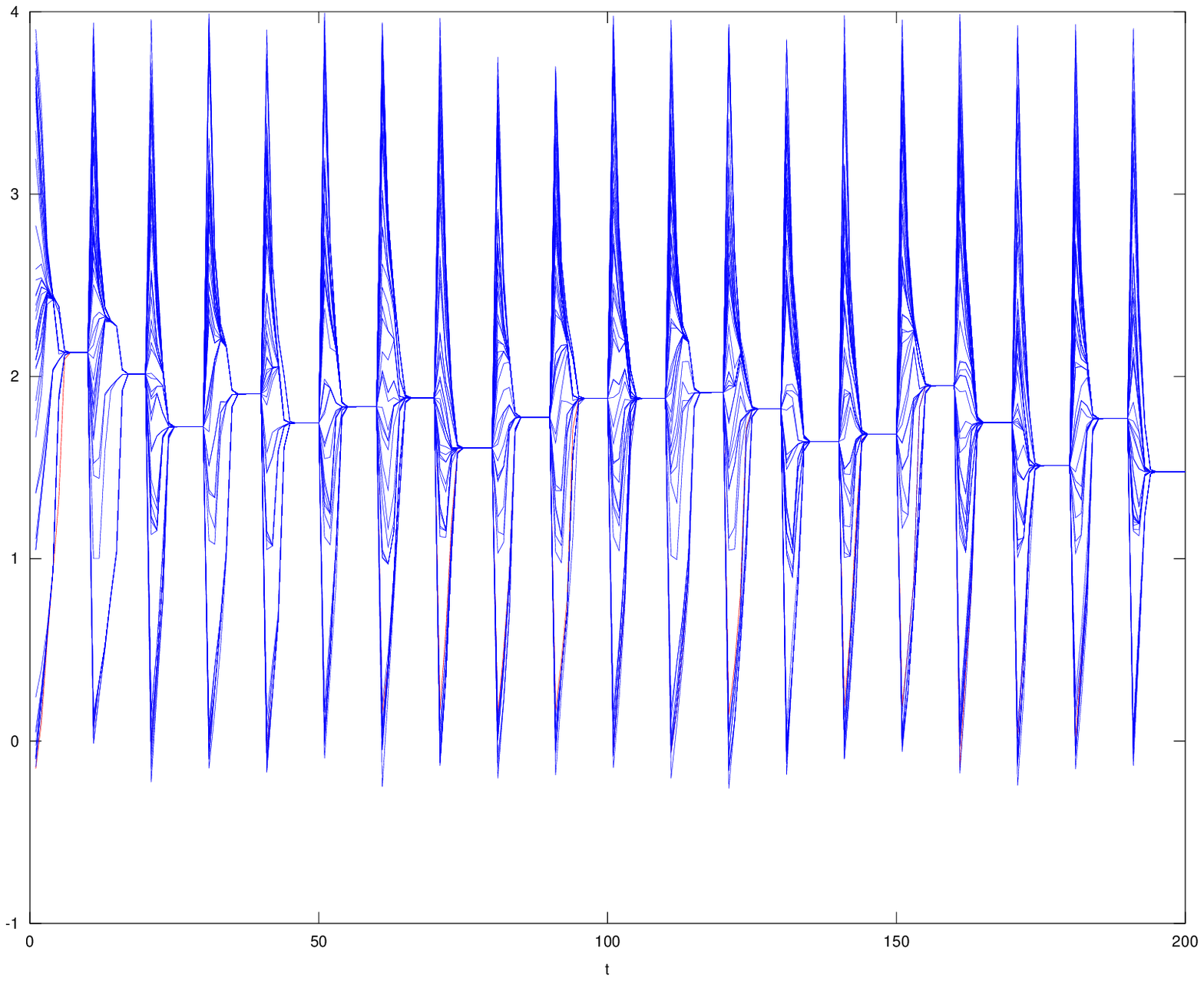}
  \caption{Parametrizations throughout: $\epsilon=0.25$, $\mu_k=0$ for
    all $k\ge 1$, $\eta_T=0.25$, $\delta_H=0.2$, $\delta_T=1.0$, $n=50$ agents,
    $\length{\struct{N}_1}=10$, $\length{\struct{N}_2}=40$. From top
    to bottom and left to 
  right: $\eta_H=0.05$, $\eta_H=0.25$, $\eta_H=1.10$, $\eta_H=1.50$. 
  We show topics $X_k$, for
  $k=1,2,\ldots,20$ and, for each topic, discussion rounds
  $t=0,1,\ldots,10$.} 
  \label{fig:etaH}
\end{figure}
\begin{figure}[!htb]
  \centering
  \includegraphics[scale=0.35]{plots/homophily/homEvolution_0p2_1p0_0p25_0p25.eps}
  \includegraphics[scale=0.35]{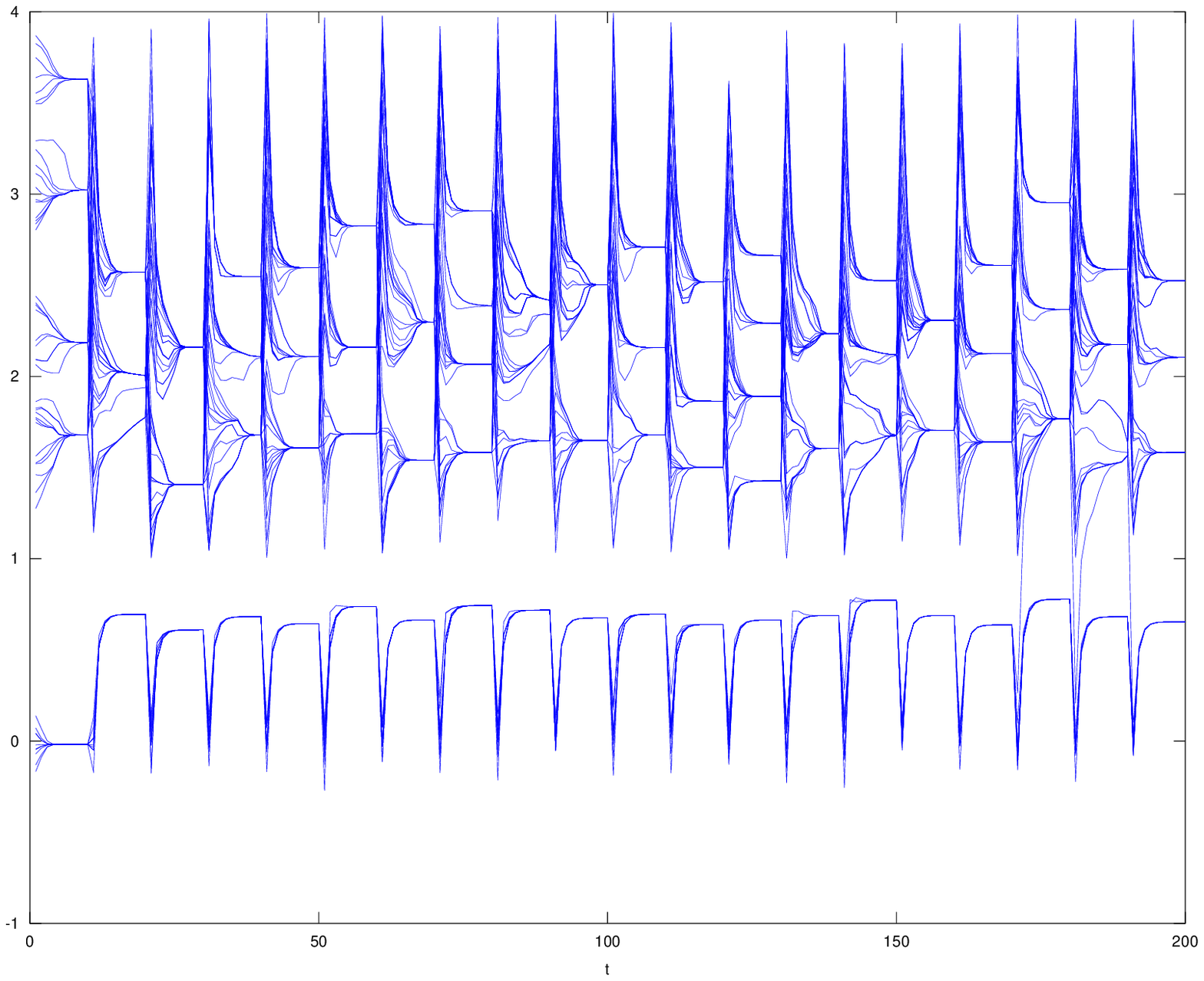}
  \caption{Parametrizations throughout: $\epsilon=0.25$, $\mu_k=0$ for
    all $k\ge 1$, $\eta_H=0.25$, $\delta_H=0.2$, $\delta_T=1.0$, $n=50$ agents,
    $\length{\struct{N}_1}=10$, $\length{\struct{N}_2}=40$. Left:
  $\eta_T=0.25$. Right: $\eta_T=2.50$. 
  We show topics $X_k$, for
  $k=1,2,\ldots,20$ and, for each topic, discussion rounds
  $t=0,1,\ldots,10$.} 
  \label{fig:etaT}
\end{figure}

In sum, in this section, we have shown that, under the `homophily
bias' and under the presence of agents with biased initial beliefs,
agents need neither become wise nor reach a consensus. If the
homophily relation is sufficiently `weak', wisdom may obtain
(Proposition \ref{prop:deltaT}), but if it sufficiently `strong',
agents' beliefs will generally cluster into distinct regions of the
belief spectrum. As in the conformity model (and also as under
opposition), even agents with zero probability of being close to truth
may influence others.


\section{Conclusion}\label{sec:conclusion}
As Acemoglu and Ozdaglar (2011) \cite{Acemoglu2011}, and many others, 
point out, the 
importance of the beliefs we hold 
cannot be overstated. For example, the demand
for a product depends on consumers' opinions and beliefs about the
quality of that product and majority opinions determine the political
agenda. Thus, beliefs also shape (our) behavior in that they lead us
to buy certain products and reject others or in that they 
 are causal factors for the 
implementation of laws and policies. On a more abstract level,
the set of norms and beliefs we hold determine, in the end, who we are
and substantiate our cultural foundations. In modern microeconomic
research, beliefs and opinions are thought to originate from \emph{social
learning} processes whereby individuals are situated in a 
network of peers
and update their opinions, e.g., via communication with
others. Rejecting the hypothesis that individuals are fully rational,
much recent research has assumed that people learn from others via
simple `rules of thumb', simply averaging peers' past beliefs to
arrive at new beliefs. Then, given that there exist `true states' for
the issues that individuals hold beliefs about, a natural question to
ask is
whether such agents, who commit the bias of not properly accounting
for the repetition of information they hear, can, in fact, still learn
these true states and, thus, become collectively `wise' (cf.\ Surowiecki (2004)
\cite{Surowieki2004}), successfully aggregating dispersed
information. 

In the current work, we have studied belief dynamics under an
\emph{endogenous} 
network formation process. In particular, we have assumed that agents
strengthen their ties to other agents based on the criterion of `past
performance' such that agents increment their trust weights to whoever has been
`close enough' to truth for a current topic. We have, moreover, assumed that
agents are \emph{multiply biased} in that they are not only susceptible to
persuasion bias --- the simplifying DeGroot learning rule --- but also
have biased 
initial beliefs (the possibly non-social, `intelligence-based' substrate of
beliefs), and commit several other sins of reasoning, such as being
biased toward members of their in-group and motivated to disassociate
from members of their out-group, being motivated to conform with
the beliefs of their reference group, or overrating beliefs that are
close to their own. Our 
goal 
has been to outline situations
under which collective failure (or at least, `failure of wisdom') can
obtain, even though the potential for wisdom --- dispersed correct 
information --- is assured. Thus, our work was also in part targeted
at the recent 
`optimism' concerning biased (`na\"{i}ve')
learning in 
social networks and crowd wisdom (e.g., Golub and Jackson (2010)
\cite{Golub2010}), which has also been challenged by 
experimental research (cf., e.g., Lorenz, Rauhut, Schweitzer, and
Helbing (2011) \cite{Lorenz2011}). 

As to our results, under the standard DeGroot learning model, we have
seen that wisdom can
fail if there are \emph{sufficiently many} agents with biased 
initial beliefs such that they, still, have positive probability of being
close to truth. 
The intuition behind this result is that 
even if the biased agents have small, but positive, probability of
`guessing' truth, then, if they are sufficiently many --- such
that 
many 
of them will still be close to truth --- the biased 
agents can, in total, receive large enough weight mass from all
agents, whence
they may become arbitrarily socially influential, 
leading
all of society to the expected value of a biased variable, away from
truth. This result may be thought of as based on the \emph{bona
  fides} bias, which says that agents do not give up the assumption
that their own (initial) beliefs are unbiased and that others' beliefs 
share this property with their own, even despite potential collective
failure,  
from which it may be motivated that agents continuously 
apply 
their
trust weight incrementing rule (to all agents).  
In the conformity
model, wisdom may fail even when the biased agents have zero
probability of being close to truth and when their number is small,
provided that the unbiased agents are sufficiently conforming. 
We might take this as an argument for why even `blatantly' and
repetitively false propaganda could work. A necessary condition for
this result is that agents want to conform to reference groups
including even 
biased and completely unknowing 
agents (which might be justified on grounds/biases such as
truth-unrelated 
prominence, e.g., due to political power or popularity). 
In
the opposition model, wisdom can fail even 
if all agents' initial beliefs are unbiased and, in addition, arbitrarily
close to truth, 
merely as a consequence of agents being
attracted by contrarian forces --- their in-groups, on the one hand,
which attract them toward truth, 
and their out-groups, on the other, from which they want to
disassociate. In the homophily model, wisdom can fail because agents
are, again, influenced by antagonistic forces --- truth, on the one
hand, and agents
with similar beliefs, on the other. Hence, biased agents' beliefs may
cluster, if they form a homogenous group, and unbiased agents' beliefs
may also cluster, so that some agents would become wise and others not.  

Concerning future research directions within our context, of course,
endogenizing several 
(more) 
of the parameters of the DeGroot learning models that we have
discussed might be of interest. In the current work, we have solely 
endogenized the social network, without explaining, for example, where
in-group/out-group antagonisms actually come from or how conformity
may develop and how reference groups evolve. The endogenizing of such
parameters would 
plausibly require psychological and socio-economic motivations that
are independent of 
the criterion of `past performance'. Moreover, in our model, we have
generally assumed that agents are \emph{homogenous} with respect to
many dimensions of attributes such as their truth tolerances $\eta$,
trust weight increments $\delta$, etc., and a heterogenous setup may
provide further insight. Finally, introducing strategic agents (cf.,
e.g., Anderlini, Gerardi and Lagunoff (2012) \cite{Anderlini2012}),
that potentially have 
incentives to deliberatively \emph{mislead} others, might be a
promising research direction to incorporate in our general setup of
social learning and collective wisdom/failure.


\begin{appendices}
\section{Proofs}\label{sec:appendix}
\subsection*{Standard model}
\begin{lemma}\label{lemma:identicalRows}
  If matrix $\mathbf{A}\in\real^{n\times n}$ has identical rows with
  row sum $s=\sum_{j=1}^n A_{ij}$, then
  $\mathbf{A}^t = s^{t-1}\mathbf{A}$ for any $t\ge 1$. 
\end{lemma}
\begin{proof}
  Follows by induction. 
\end{proof}
\subsubsection*{Wisdom of crowds under initial beliefs centered around
  truth}
The following are results from Golub and Jackson (2010)
\cite{Golub2010}. They state conditions under which a growing
population, parametrized by its size $n$, converges to truth $\mu$ under the
assumption that agents receive initial belief signals that are
centered around $\mu$ (as in \eqref{eq:centered}). The statement of
the below results is that agents become ($\epsilon$-)wise (for any
$\epsilon>0$) if and only if the \emph{influence} of the most
influential agent converges to zero as $n$ increases, whereby an
agent's influence is given by his 
\emph{social influence}, as we have discussed above and 
as we define below. In
undirected networks ($W_{ij}=W_{ji}$ for all $i,j\in[n]$) with uniform
weights, this condition is tantamount to all agents' relative degrees
(the number of links they have to other agents divided by the total
number of links in the network) converging to zero as $n$ becomes
large. Hence, in this setup, an obstacle to wisdom would be the
circumstance when each agent who newly enters society assigns, e.g.,
a constant fraction of his links to a particular agent, who would then
become 
excessively influential. 

\begin{remark}
  If a social network $\mathbf{W}$ induces a consensus, then limiting
  beliefs can be represented as
  $\mathbf{b}(\infty)=\mathbf{s}^\intercal\mathbf{b}(0)$, for a
  non-negative vector $\mathbf{s}$ with $\sum_{i=1}^n s_i=1$ which we call
  the \emph{social influence vector} and $s_i$ agent
  $i$'s \emph{influence}. 
  The influence vector is given as the unique normalized unit-vector
  $\mathbf{s}$ which satisfies
  $\mathbf{s}=\mathbf{W}^\intercal\mathbf{s}$ (i.e., $\mathbf{s}$ is
  the normalized 
  unit-eigenvector of $\mathbf{W}^\intercal$ corresponding to the
  eigenvalue $\lambda=1$). 
\end{remark}
Now, as in Golub and Jackson (2010) \cite{Golub2010}, we parametrize
social networks $\mathbf{W}$ by their population size $n$, which we
denote by $\mathbf{W}(n)$; we also
parametrize other quantities such as limiting beliefs of a set of
agents by population size $n$ (here and in the following, we omit
reference to topics $k$ for notational convenience). Moreover, we
denote a \emph{society} by the sequence
$\bigl(\mathbf{W}(n)\bigr)_{n\in\nn}$. We restate the following lemma
and the proposition from Golub and Jackson
(2010) \cite{Golub2010}, which they list as Lemma 1 and Proposition 2.   
\begin{lemma}[A law of large numbers]
  If $\bigl(\mathbf{s}(n)\bigr)_{n\in\nn}$ is any sequence of
  influence vectors, then
  \begin{align*}
    \mathbf{s}(n)^\intercal\mathbf{b}(0;n)\goesto \mu \quad\text{as
    }n\goesto\infty  
  \end{align*}
  (where convergence is in probability or almost surely) if and only
  if $s_1(n)\goesto 0$, where we assume, without loss of generality,
  that $s_1(n)\ge s_2(n)\ge\cdots\ge s_n(n)$. 
\end{lemma}
\begin{proposition}
  If $\bigl(\mathbf{W}(n)\bigr)_{n\in\nn}$ is a sequence of 
  networks, each inducing a consensus, then the underlying agents
  become ($\epsilon$-)wise (for any 
  $\epsilon>0$) as $n\goesto\infty$ if and only if the associated
  influence vectors are 
  such that $s_1(n)\goesto 0$ as $n\goesto\infty$. 
\end{proposition}
We now argue informally that the proposition entails convergence to
truth in the situation where agents' initial beliefs are centered
around truth as in \eqref{eq:centered} and in our setup of 
endogenous weight formation.\footnote{We assume that $k$ is so large
  that each network $\mathbf{W}^{(k)}(n)$ always induces a
  consensus. Note that, if agents are stochastically intelligent, a
  consensus is reached quickly (and increasingly fast in the number of
  agents $n$), by the results developed in Section \ref{sec:standard}.} 
Namely, we first argue that an agent's influence $s_i$ is
directly inversely proportional to his variance $\sigma_i^2$. Although
a proof thereof would require technical sophistication, the claim
appears very 
intuitive 
since influence $s_i$ captures weight mass assigned to an agent by
other agents (in addition to these agents' influence; cf.\ DeMarzo,
Vayanos and Zwiebel (2003) \cite{Demarzo2003}) and, in our
setup, the weight mass that an agent receives is directly inversely
proportional to his variance $\sigma_i^2$ (more intelligent agents
receive weight increases more often). Next, consider networks
$\bigl(\mathbf{W}(n)\bigr)_{n\in\nn}$ where, for all $n\in\nn$,
agents' 
variances $\sigma_i^2$ satisfy $\sigma_i^2\ge\bar{\sigma}^2>0$ for
some lower bound 
$\bar{\sigma}^2>0$.\footnote{Should there be no lower bound on
  the most 
  intelligent agent's variance, then this agent may become excessively
influential but his initial beliefs also become arbitrarily accurate,
so that society becomes ($\epsilon$-)wise simply because one agent is
arbitrarily well-informed.} Then, as $n\goesto\infty$, the influence of the most
influential agent certainly goes to zero since the number of agents
increases (all of which are influential in the sense that they receive
weight mass from others) while the expected weight mass that the most
influential 
agent receives is bounded.

\subsubsection*{Varying weights on own beliefs}
\begin{proof}[Proof of Proposition \ref{prop:demarzo}]
  Since $\mathbf{W}=\mathbf{W}^{(k)}$ converges for all initial belief
  vectors $\mathbf{b}(0)$, there exists a matrix $\mathbf{W}^\infty$
  such that $\lim_{t\goesto\infty}\mathbf{W}^t=\mathbf{W}^\infty$. To
  prove the proposition, show that
  $\prod_{s=0}^{t-1}\mathbf{W}(\lambda_s)$ converges to
  $\mathbf{W}^\infty$ as $t\goesto\infty$, whereby
  $\mathbf{W}(\lambda)=\Bigl((1-\lambda)\mathbf{I}+\lambda\mathbf{W}\Bigr)$
  and where
  $\mathbf{b}(t)=\Bigl(\prod_{s=0}^{t-1}\mathbf{W}(\lambda_s)\Bigr)\mathbf{b}(0)$
  according to \eqref{eq:degrootupdatedemarzo}. Proceed exactly as in
  DeMarzo, Vayanos, and Zwiebel (2003) \cite{Demarzo2003}.

  Define the random variable $\Lambda_t$ to be equal to $1$ with
  probability $\lambda_t$ and zero otherwise. Assume also that
  $\Lambda_t$ are independent over time. Define the random matrix
  $\mathbf{Z}_t$ by
  $\mathbf{Z}_t=\prod_{s=0}^{t-1}\mathbf{W}(\Lambda_s)=\mathbf{W}^{\sum_{s=0}^{t-1}\Lambda_s}$. Then
  $\Exp[\mathbf{Z}_t]=\prod_{s=0}^{t-1}\mathbf{W}(\lambda_s)$. By the
  Borel-Cantelli lemma, if $\sum_{t=0}^\infty
  \text{Pr}[\Lambda_t=1]=\sum_{t=0}^\infty\lambda_t=\infty$, then 
  \begin{align*}
    \text{Pr}[\sum_{t=0}^\infty \Lambda_t=\infty] =
    \text{Pr}[\Lambda_t=1 \text{ infinitely often}] 
     = 1.
  \end{align*}
  Since the matrix $\mathbf{W}^t$ is bounded uniformly in $t$, the
  dominated convergence theorem implies that
  \begin{align*}
    \lim_{t\goesto\infty} \prod_{s=0}^{t-1}\mathbf{W}(\lambda_s) =
    \lim_{t\goesto\infty} \Exp[\mathbf{Z}_t] =
    \lim_{t\goesto\infty}\Exp[\mathbf{W}^{\sum_{s=0}^{t-1}\Lambda_s}]
    = \mathbf{W}^\infty. 
  \end{align*}
\end{proof}

\subsection*{Opposition}
\begin{lemma}
  Consider any $n\times n$ matrix $\mathbf{A}$ of the form 
\begin{align}\label{eq:formA'}
  \mathbf{A} = 
  \begin{pmatrix}
    \beta & \alpha & \ldots & \alpha\\
    \alpha & \beta & \ldots & \alpha\\
    \vdots & \ldots & \ddots & \vdots\\
    \alpha & \alpha & \ldots & \beta
  \end{pmatrix}
\end{align}
 with $\alpha,\beta\in\real$. 
 The eigenvalues of matrix $\mathbf{A}$ are given by
 $\lambda_1=\beta+(n-1)\alpha$ and
 $\lambda_2=\cdots=\lambda_n=\beta-\alpha$. 
\end{lemma}
\begin{proof}
  We first consider the determinant of
  $\mathbf{A}=\mathbf{A}(n)$. Subtracting the second row from the
  first, we
  find
  $\det(\mathbf{A}(n))=(\beta-\alpha)\det(\mathbf{A}(n-1))+(\beta-\alpha)\det(\mathbf{B}(n-1))$,
  where $\mathbf{B}(n)$ is the $n\times n$ matrix with
  $[\mathbf{B}(n)]_{ij}=[\mathbf{A}(n)]_{ij}$ for all $i,j$ with
  $(i,j)\neq (1,1)$; for $(i,j)=(1,1)$, we have
  $[\mathbf{B}(n)]_{ij}=\alpha$. Proceeding analogously as for
  $\mathbf{A}(n)$, we find 
  $\det(\mathbf{B}(n))=(\beta-\alpha)^{n-1}\alpha$. Therefore,
  \begin{align*}
    \det(\mathbf{A}(n)) =
    (\beta-\alpha)^{n-1}\bigl(\beta+(n-1)\alpha\bigr). 
  \end{align*}
  Now, consider the characteristic polynomial of $\mathbf{A}(n)$; it
  is $\chi(\lambda)=\det(\mathbf{A}-\lambda \mathbf{I}_n)$. Note that
  $\mathbf{A}-\lambda \mathbf{I}_n$ is a matrix of the form
  \eqref{eq:formA'}. Hence, its determinant is given by 
  \begin{align*}
    \chi(\lambda) =
    \Bigl((\beta-\alpha)-\lambda\Bigr)^{n-1}\Bigl(\beta+\alpha(n-1)-\lambda\Bigr).  
  \end{align*}
  This concludes the proof. 
\end{proof}
\begin{lemma}\label{lemma:A}
  Consider any matrix of the form \eqref{eq:A} with $a,b,c,d>0$
  and such that $\sum_{j=1}^n \length{A_{ij}}=1$ for all
  $i=1,\ldots,n$. 
  Let $n_1=1$. Then, 
  the characteristic polynomial of $\mathbf{A}$ is
  given by
  \begin{align*}
    \chi(\lambda)=\det(\mathbf{A}-\lambda \mathbf{I}_n) =
    (-\lambda)^{n-2}\Bigl(\lambda^2-(a+(n-1)d)\lambda+(n-1)(ad-bc)\Bigr) = 
    (-\lambda)^{n-2}(\lambda-1)(\lambda-q), 
  \end{align*}
  where $q=(n-1)(ad-bc)=a-1+(n-1)d$. 
\end{lemma}
\begin{proof}
  Expanding the determinant along the last row (and subtracting the
  second-to-last row from the last), we find that the
  determinant $\det(\mathbf{B}_n)$ of
  $\mathbf{B}_n=\mathbf{A}_n-\lambda\mathbf{I}_n$, with
  $\mathbf{A}_n=\mathbf{A}$, is given by
  \begin{align*}
    -\lambda \det(\mathbf{B}_{n-1})-\lambda\det(\mathbf{C}_{n-1})
  \end{align*}
  whereby $\mathbf{C}_n=\mathbf{A}_n-\lambda\mathbf{I}_n$, except for
  the entry in row $n$ and column $n$, which is
  $[\mathbf{C}_n]_{nn}=A_{nn}$. The determinant of $\mathbf{C}_n$ can
  easily be found to be $(-\lambda)^{n-2}\cdot
  \Bigl((a-\lambda)d-bc\Bigr)$. Then solving $\det(\mathbf{A}_n)$
  inductively 
  leads to the required solution. Finally, the factorization of the
  quadratic polynomial results from the fact that $\mathbf{A}$ has one
  eigenvalue of $1$, as can readily be checked. 
\end{proof}
From Lemma \ref{lemma:A}, we can infer that matrix $\mathbf{A}$ from
\eqref{eq:A} has $n-2$ eigenvalues $0$, one eigenvalue of $1$, and one
eigenvalue $q$, which is a real eigenvalue. Moreover, all eigenvalues
of $\mathbf{A}$ are bounded from above by $1$ (cf.\ Eger (2013)
\cite{Eger2013}, Proposition 6.3).  
Assume that $q$ were $-1$. Then 
\begin{align*}
  a-1+(n-1)d = q=- 1 \quad\iff\quad a+(n-1)d = 0 \quad\iff\quad
  a=-(n-1)d, 
\end{align*}
whence $a$ is
negative, which contradicts $a>0$. Thus, assume 
$q$ 
were $+1$. Then
\begin{align*}
  a+(n-1)d = 2,
\end{align*}
which contradicts
$a+(n-1)d=a+n_2d<1+1=2$ 
(since both $b$ and $c$ are positive and recall the row sum
restrictions $n_1a+n_2b=1$, etc.). Therefore, $\lambda=1$ is the
only eigenvalue of $\mathbf{A}$ on the unit circle and it has
algebraic multiplicity of $1$.

\subsection*{Conformity}
\begin{lemma}\label{lemma:invert}
  Consider $\mathbf{I}_n-\mathbf{A}$ for an $n\times n$ matrix
  $\mathbf{A}$. If $\lim_{k\goesto\infty}\mathbf{A}^k=\mathbf{0}$, then
  $\mathbf{I}_n-\mathbf{A}$ is invertible and its inverse is given by
  the \emph{Neumann series}
  \begin{align*}
    \inv{(\mathbf{I}_n-\mathbf{A})} = \sum_{k=0}^\infty\mathbf{A}^k.
  \end{align*}
\end{lemma}
\begin{proof}
  See Meyer (2000) \cite{Meyer2000}, p.618. 
\end{proof}
\begin{proof}[Proof of Proposition \ref{prop:buechel1}]
  Our proof follows along the lines of the proof of the corresponding
  proposition 
  of Buechel, Hellmann, and Kl\"{o}{\ss}ner (2012) \cite{Buchel2012}. 

  The best response $s_i^*$ of player $i$ to the strategies $s_{-i}$
  of the other players is given by the first order conditions,
  \begin{align*}
    \pardiv{u_i(s_i,s_{-i};b_i)}{s_i}|_{s_i=s_{i}^*} =
    -2(1-\delta_i)(s_i^*-b_i)-2\delta_i\Bigl(s_i^*-\sum_{j\neq
      i}Q_{ij}s_j\Bigr) = 0
  \end{align*}
  for all $i\in[n]$. 
  Note that the best response is unique. A strategy profile
  $\mathbf{s}^*\in S^n$ is a Nash equilibrium if and only if $s_i^*$
  is a best response to $\mathbf{s}^*_{-i}$. Thus, Nash equilibria
  $\mathbf{s}^*\in S^n$ satisfy:
  \begin{align*}
    (\mathbf{I}_n-\mathbf{\Delta})(\mathbf{s}^*-\mathbf{b})+\mathbf{\Delta}(\mathbf{s}^*-\mathbf{Q}\mathbf{s}^*)
    =
    (\mathbf{I}_n-\mathbf{\Delta})(\mathbf{s}^*-\mathbf{b})+\mathbf{\Delta}(\mathbf{I}_n-\mathbf{Q})\mathbf{s}^*
    = \mathbf{0}. 
  \end{align*}
  Rewriting leads to 
  \begin{align*}
    \mathbf{s}^* =
    \inv{(\mathbf{I}_n-\mathbf{\Delta}\mathbf{Q})}(\mathbf{I}_n-\mathbf{\Delta})\mathbf{b},
  \end{align*}
  which is well-defined since 
  $\mathbf{I}_n-\mathbf{\Delta}\mathbf{Q}$ is invertible by Lemma
  \ref{lemma:invert}. Namely, we have 
  \begin{align*}
    \norm{\mathbf{\Delta Q}}^k \le
    \norm{\mathbf{\Delta}}^k\norm{\mathbf{Q}}^k \le
    \Bigl(\underbrace{\max_{i\in[n]}{\abs{\delta_i}}}_{=:\delta_{\text{max}}}\Bigr)^k\norm{\mathbf{Q}}^k  
  \end{align*}
  for any matrix norm $\norm{\cdot}$. Hence, 
  \begin{align*}
    0\le \lim_{k\goesto\infty} \norm{\mathbf{\Delta Q}}^k\le
    \lim_{k\goesto\infty}{(\delta_{\text{max}})}^k\norm{\mathbf{Q}}^k = 0,
  \end{align*}
  since $\length{\delta_i}<1$ by assumption, for all $i\in[n]$, and
  $\norm{\mathbf{Q}}^k$ 
  is bounded since $\mathbf{Q}$ is row-stochastic. Therefore,
  $\lim_{k\goesto\infty} (\mathbf{\Delta Q})^k=\mathbf{0}$. 
\end{proof}
\begin{proof}[Proof of Lemma \ref{lemma:one}]
  Consider $\mathbf{M}\one$ (which is $\mathbf{M}\cdot\one$), which is
  \begin{align*}
    \mathbf{D}\one +
    (\mathbf{W}-\mathbf{D})\inv{(\mathbf{I}_n-\mathbf{\Delta}\mathbf{Q})}
    (\mathbf{I}_n-\mathbf{\Delta})\one. 
  \end{align*}
  It suffices to show that 
  \begin{align*}
    \mathbf{R}\one:=\inv{(\mathbf{I}_n-\mathbf{\Delta}\mathbf{Q})}
    (\mathbf{I}_n-\mathbf{\Delta})\one = \one
  \end{align*}
  because of row-stochasticity of $\mathbf{W}$, which entails that
  $\mathbf{W}\one = \one$. 

  Now, we have
  $\inv{(\mathbf{I}_n-\mathbf{\Delta}\mathbf{Q})}=\sum_{r=0}^\infty(\mathbf{\Delta}\mathbf{Q})^r$
  by row-stochasticity of $\mathbf{Q}$ and since
  $\length{\delta_i}<1$. Hence
  \begin{align*}
    \mathbf{R}\one &=
    \sum_{r=0}^\infty(\mathbf{\Delta}\mathbf{Q})^r(\mathbf{I}_n-\mathbf{\Delta})\one
    =
    (\mathbf{I}_n-\mathbf{\Delta})\one+\sum_{r=1}^\infty(\mathbf{\Delta}\mathbf{Q})^{r-1}[\mathbf{\Delta}\mathbf{Q}\one-\mathbf{\Delta}\mathbf{Q}\mathbf{\Delta}\one]
    \\
    &=
    (\mathbf{I}_n-\mathbf{\Delta})\one+\sum_{r=1}^\infty(\mathbf{\Delta}\mathbf{Q})^{r-1}[\mathbf{\Delta}\one-\mathbf{\Delta}\mathbf{Q}\mathbf{\Delta}\one]
    =
    (\mathbf{I}_n-\mathbf{\Delta})\one+\inv{(\mathbf{I}_n-\mathbf{\Delta}\mathbf{Q})}[\mathbf{I}_n-\mathbf{\Delta}\mathbf{Q}]\mathbf{\Delta}\one\\
    &= (\mathbf{I}_n-\mathbf{\Delta})\one+\mathbf{\Delta}\one = \one. 
  \end{align*}
\end{proof}
\begin{proposition}
  In the situation of Example \ref{example:socialInfluence}, the
  social influence weights $x,x$ and $y$ of agents $1,2$ and $3$ are
  given by
  \begin{align*}
    x=\frac{2(1-a)}{4-ab-3a},\quad\text{and}\quad y =
    \frac{a(1-b)}{4-ab-3a}.  
  \end{align*}
\end{proposition}
\begin{proof}
  The social influence weights can be found by computing $\mathbf{M}$
  and then solving $\mathbf{M}^\intercal \mathbf{x}=\mathbf{x}$ where
  $\mathbf{x}=(x,x,y)^\intercal$. The computation, though cumbersome,
  is straightforward. 
\end{proof}
\subsection*{Homophily}
The following theorem is the `stabilization theorem' of Lorenz (2005)
\cite{Lorenz2005}. It discusses convergence of the opinion dynamics
process $\mathbf{b}(t+1)=\mathbf{W}(\mathbf{b}(t),t)\mathbf{b}(t)$,
where weight matrix $\mathbf{W}$ may depend on time $t$ and the
current vector of beliefs $\mathbf{b}(t)$, as in the homophily model
we have sketched. We abbreviate the theorem to fit our needs. 
\begin{theorem}[Lorenz (2005) \cite{Lorenz2005}]\label{theorem:lorenz} 
  Let $(\mathbf{W}(t))_{t\in\nn}$ be a sequence of row-stochastic
  matrices. If each matrix $\mathbf{W}(t)$ satisfies 
  \begin{itemize}
    \item[(1)] $[\mathbf{W}(t)]_{ii}>0$ for all $i\in [n]$ (`each agent has
      a little bit of self-confidence'),
    \item[(2)]  $[\mathbf{W}(t)]_{ij}>0\iff  [\mathbf{W}(t)]_{ji}>0$
      (`confidence is mutual'),
    \item[(3)] there exists $\kappa>0$ (that is independent of $t$)
      such that the smallest positive 
      entry of $\mathbf{W}(t)$ is greater than $\kappa$ (`positive
      weights do not converge to zero'),
  \end{itemize}
  then $\lim_{t\goesto\infty}\mathbf{b}(t)$ exists, that is, the
  belief dynamics process converges. 
\end{theorem}
While Theorem \ref{theorem:lorenz} applies, in particular, to the
Hegselmann and Krause models, on which our homophily model rests, it
does not apply to the latter. This is easy to see: while condition (1)
in the theorem 
on $\mathbf{W}(t)$ is satisfied in our case (due to $\delta_H>0$ and
$\norm{b_i^k(t)-b_i^k(t)}=0<\eta_H$ for any positive $\eta_H$), both
conditions (2) and (3) may be violated in our modeling. Condition (2)
may be violated because of truth related weight adjustment, which is
generally asymmetric (agent $i$
may have been true for a topic $X_k$, while $j$ may not have been
true so that $j$ increases his weight for $i$ while $i$ does not
increase his weight for $j$); and condition (3) may be violated
because a positive link weight between two agents may converge to zero
in our model, e.g., when an agent $i$ has known truth for a topic, so
that another agent $j$ increases his link weight for $i$ (based on
truth), but $i$ and $j$'s beliefs are sufficiently distinct such that
homophily, toward other agents, causes the link weight
$[\mathbf{W}(t)]_{ji}$ to drop to zero, as $t\goesto\infty$. 
 
\section{Experiment} \label{sec:appendixb}
Below, we list details on the experiment indicated in the
introduction. 
In total, $n=119$ subjects, all from \texttt{Amazon Mechanical Turk},
participated in the experiment; not all subjects answered all
questions.\footnote{The data set is available upon request.} We set a
time limit for answering the $16$ `common 
knowledge' questions of $3$
minutes and reimbursed subjects with $60$ US cents if they completed
and submitted the questions (this required them to press the `submit'
button rather than to indeed answer all questions), which corresponds 
to an hourly wage of $12$ USD. Obviously, this was an attractive wage,
since all requested slots ($119$) were filled within approximately one
hour. On average, individuals took $2$ minutes and $25$ seconds to
answer all $16$ questions, including reading the instructions and
optionally providing feedback, although some subjects complained that
time limits were too narrow. Below, we summarize the instructions, the
questions, and give histograms of the distributions of answers (Figure
\ref{fig:a-2}) as well
as of the `logarithmically scaled' data --- for some questions,
individuals beliefs' seemed to be lognormally distributed, so we
provide these histograms in Figure \ref{fig:a-1}. We note that we ---
very slightly --- adjusted the data when it very obviously seemed to
be corrupted. For example, one person gave as average daily
temperature in Miami in July the number $8856347$, which cannot
plausibly be correct; similarly, two people answered the question
concerning the age of homo sapiens sapiens as $80$ years, which
constitutes most probably a misunderstanding of the question. 

From the histograms in Figures \ref{fig:a-2} and \ref{fig:a-1}, we
observe that people's beliefs appear to be centered around truth only
occasionally. In particular, for example, 
the histogram for the question concerning the average height of an
adult male US American appears to be consistent with independent
normal distributions, centered around truth, as underlying subjects'
beliefs. For the question 
regarding the number of official languages of the European union, the
population density of Beijing, and the distance from earth to moon, independent
lognormal distributions appear as plausible. As we have already
discussed in the introduction, neither the mean nor the median are
very reliable quantities for the true values of questions, as
Table \ref{table:appendixB} illustrates.

\bigskip

\noindent
\fbox{\parbox{\linewidth}{
\textbf{Instructions}\\
Give truthful estimates on 16 questions such as ``When did the first settlers arrive in America?''. Don't look them up, don't google them. We're interested in your honest estimate/guess, not in your ability to use search engines. If you don't know the correct answer, please try to provide your best guess. Please answer all 16 questions.\\
A valid answer to the above question might be ``in 1620'' (if this is what you think the correct answer is). Certainly don't take longer than 20 seconds to answer any one question.\\
If your answer requires a unit such as ``pounds'', ``miles'', or
``kilometers'', please indicate it, for the sake of clarity. }}
\begin{table}[!htb]
  {\small
  \begin{tabular}{|c|l|r|}\hline
    (1)&the average daily temperature in Miami in July (in Fahrenheit or Celsius)?&$87.8$F\\
    (2)&the population size of New York city, as of 2012?&$8,336,697$\\
    (3)&the current level of the Dow Jones stock market index?&$15,658.36$\\
    (4)&the number of official languages in the European Union?&$24$\\
    (5)&the age of modern humans (homo sapiens sapiens) on earth? &\\
    & In
    other words, since how long do (modern) humans exist on earth?& $200,000$\\
    (6)&the year the first world war started?& $1914$\\
    (7)&the number of McDonald's restaurants in the US?& $12,804$\\
    (8)&the number of people per square mile (or square kilometer) in
    China's capital Beijing?& $3,300/$sqm\\
    (9)&the how many-th US president was Bill Clinton?& $42$nd\\
    (10)&the average height of an adult male in the US as of 2012? (in
    feet and inches or centimeters)& $177.8$cm\\
    (11)&the distance from earth to the moon (in miles or
    kilometres)?& $238,610$m\\
    (12)&the number of states the United States consists of?& $50$\\
    (13)&the fraction of population in the US that is left-handed in
    percent?& $10.5\%$\\
    (14)&the average life-span of an African elephant (in years) in
    the wild?& $56$\\
    (15)&$17-4\times 2$?& $9$\\
    (16)&the diameter of the sun in miles or kilometers?& $857,490$m\\
    \hline
  \end{tabular}
  }
  \caption{Questions, to be preceded by `What do you think is ...', and
    `true' answers. Most `true' answers are taken from 
    \texttt{wikipedia} or similar resources.}
\end{table}
\begin{figure*}[!htb]
  \includegraphics[scale=0.4]{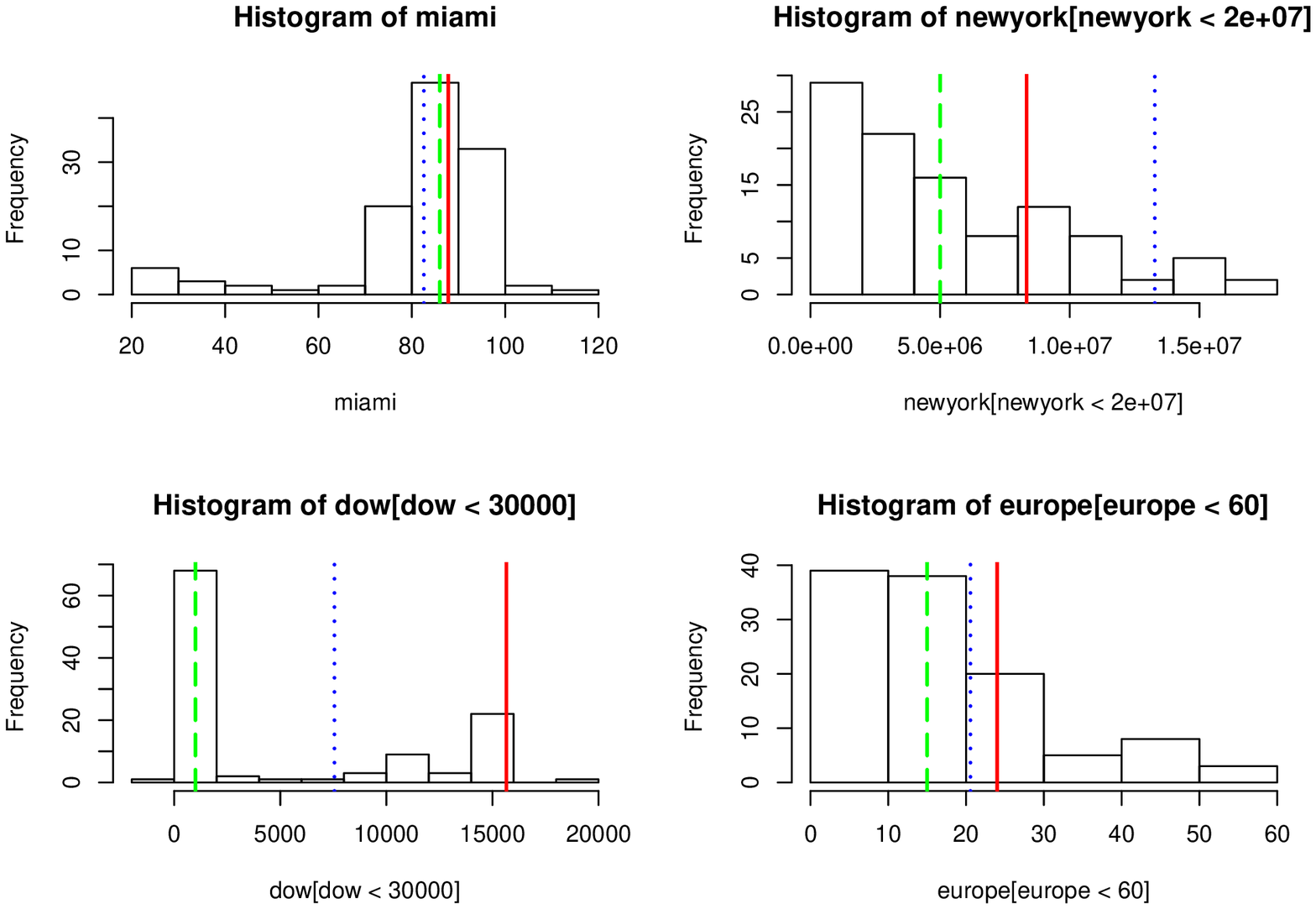}
  \includegraphics[scale=0.4]{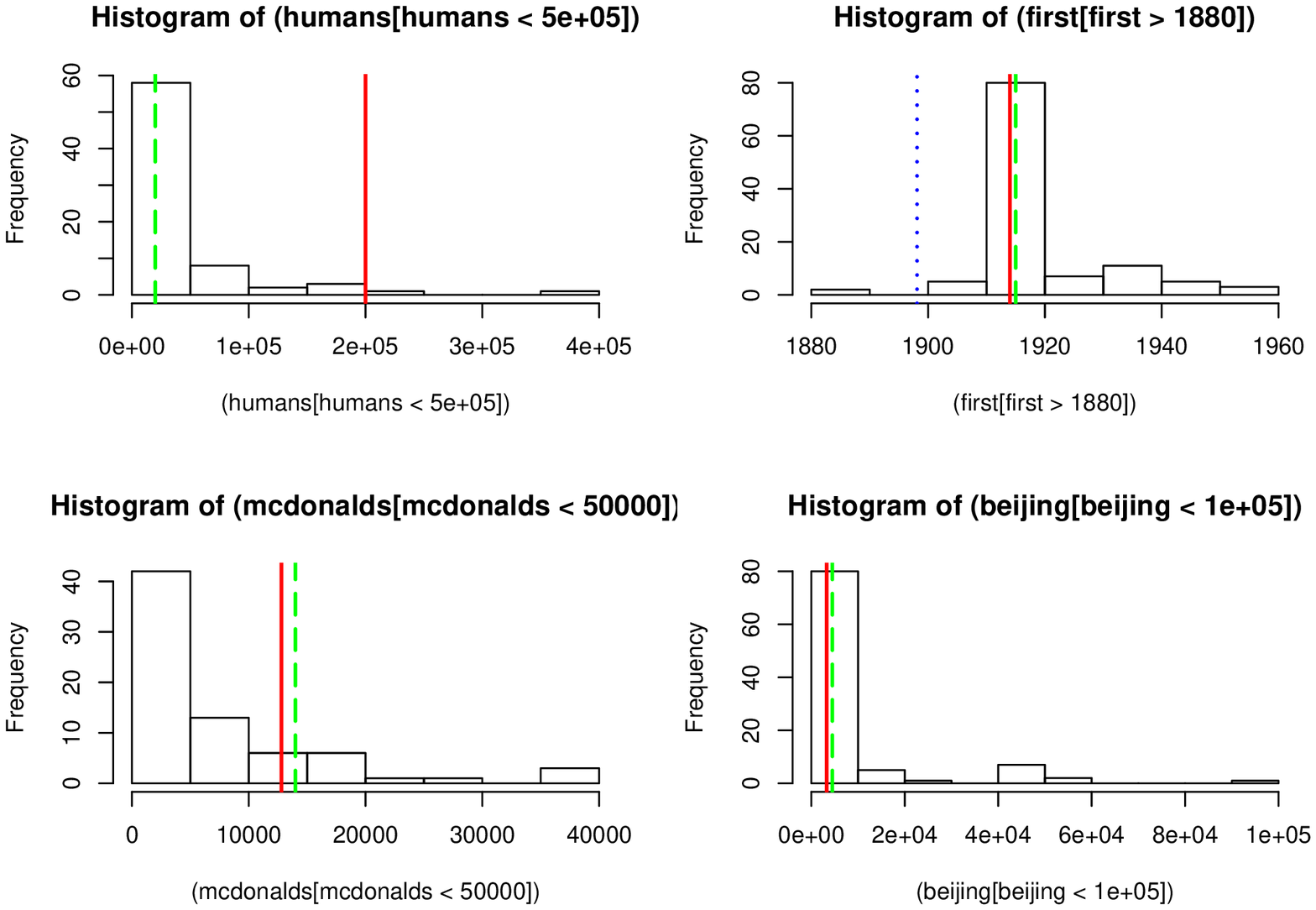}
  \includegraphics[scale=0.4]{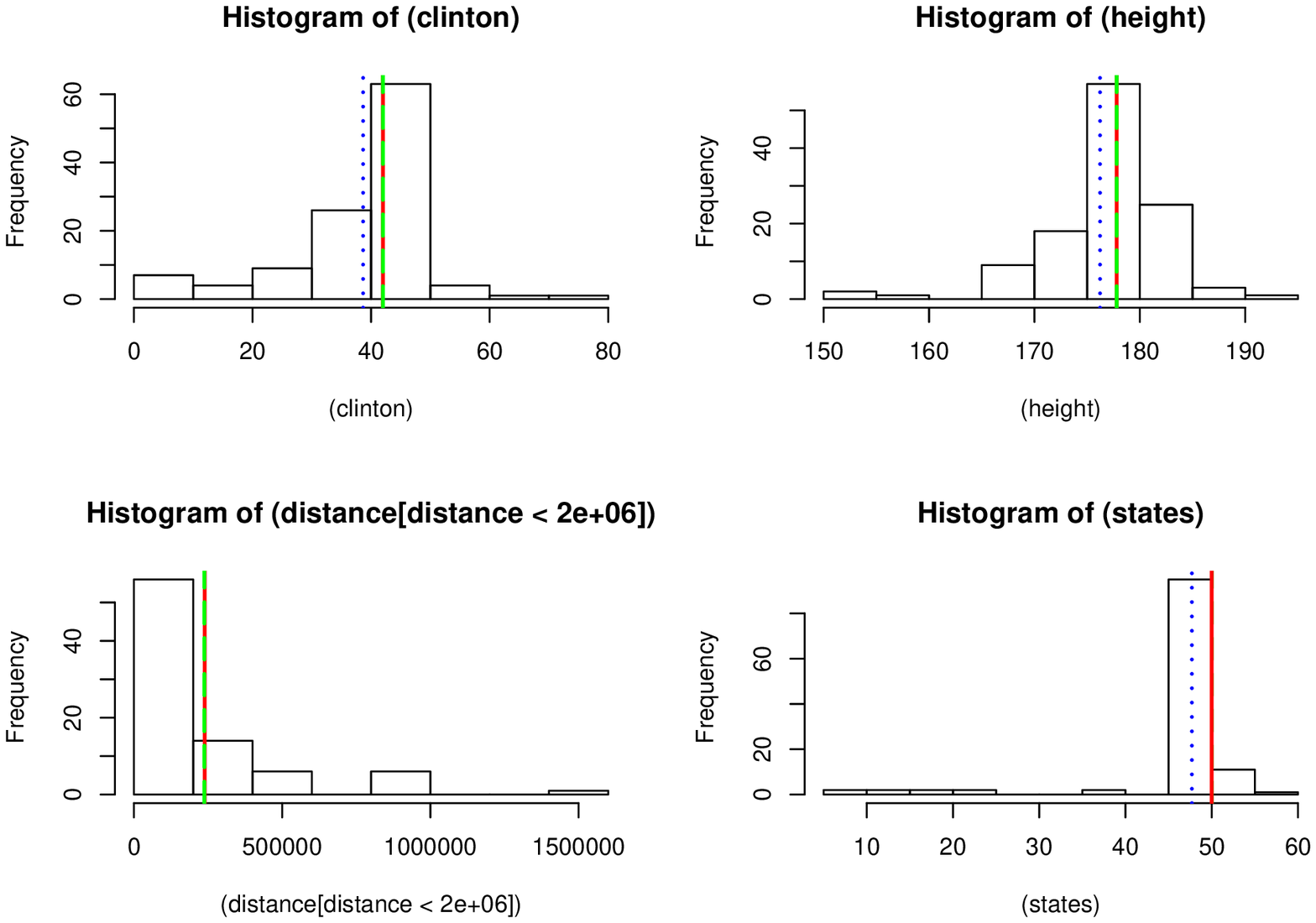}
  \includegraphics[scale=0.4]{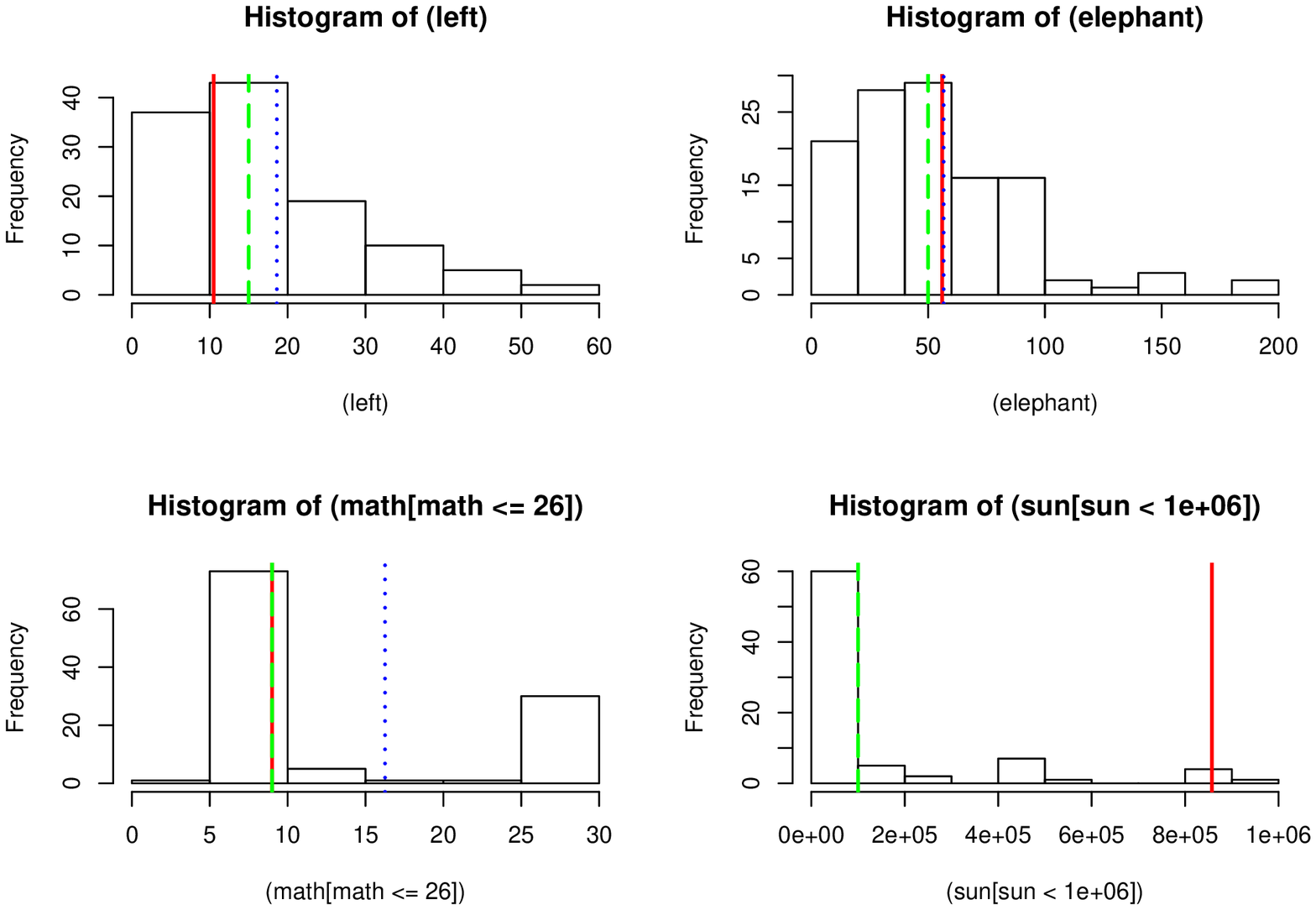}
  \caption{Histograms of answers to questions (1) to (16), top to
    bottom and left to right. Mean (dotted blue), median (dashed green), and
    truth (solid red) indicated.}
  \label{fig:a-2}
\end{figure*}
\begin{figure*}[!htb]
  \includegraphics[scale=0.4]{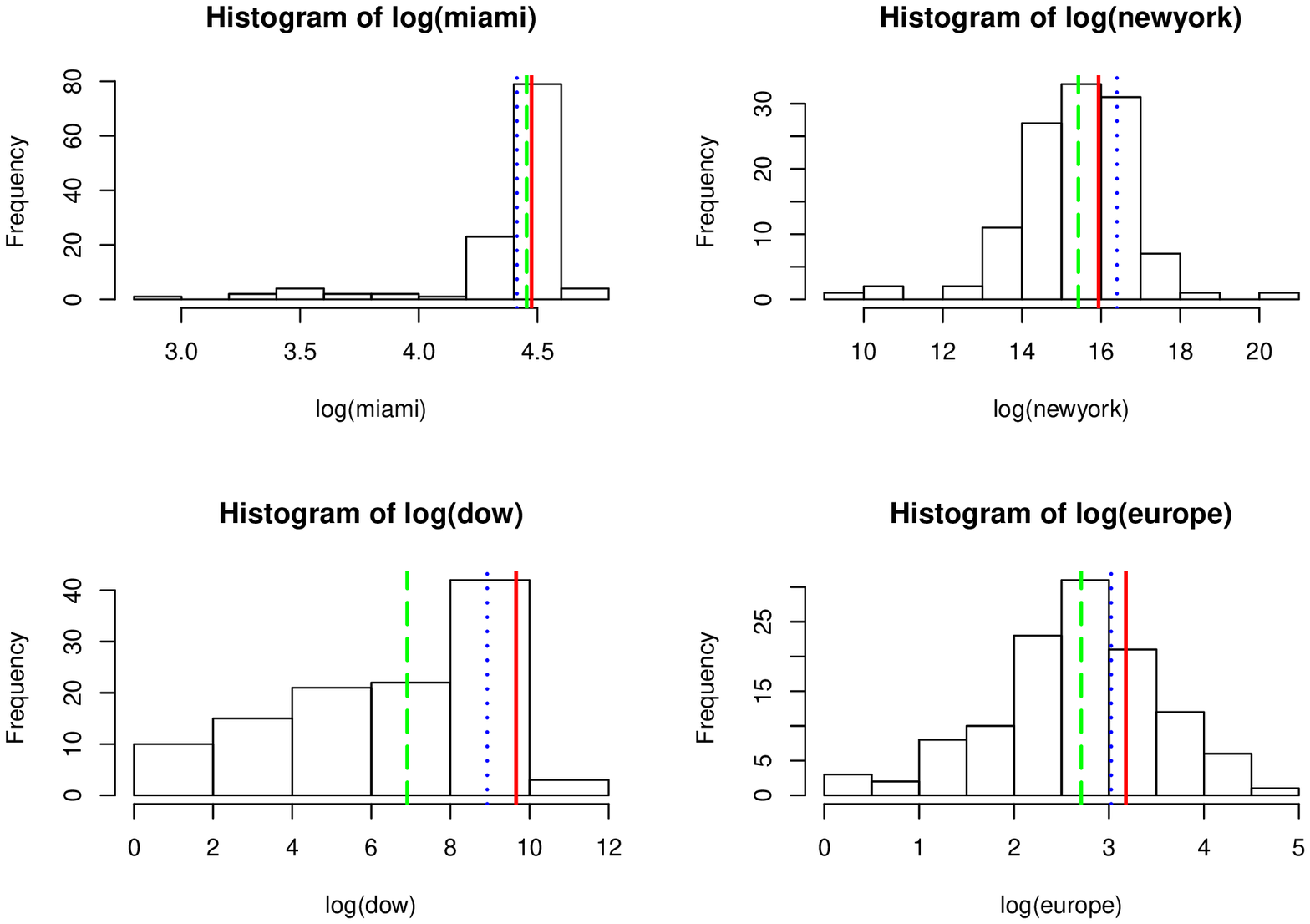}
  \includegraphics[scale=0.4]{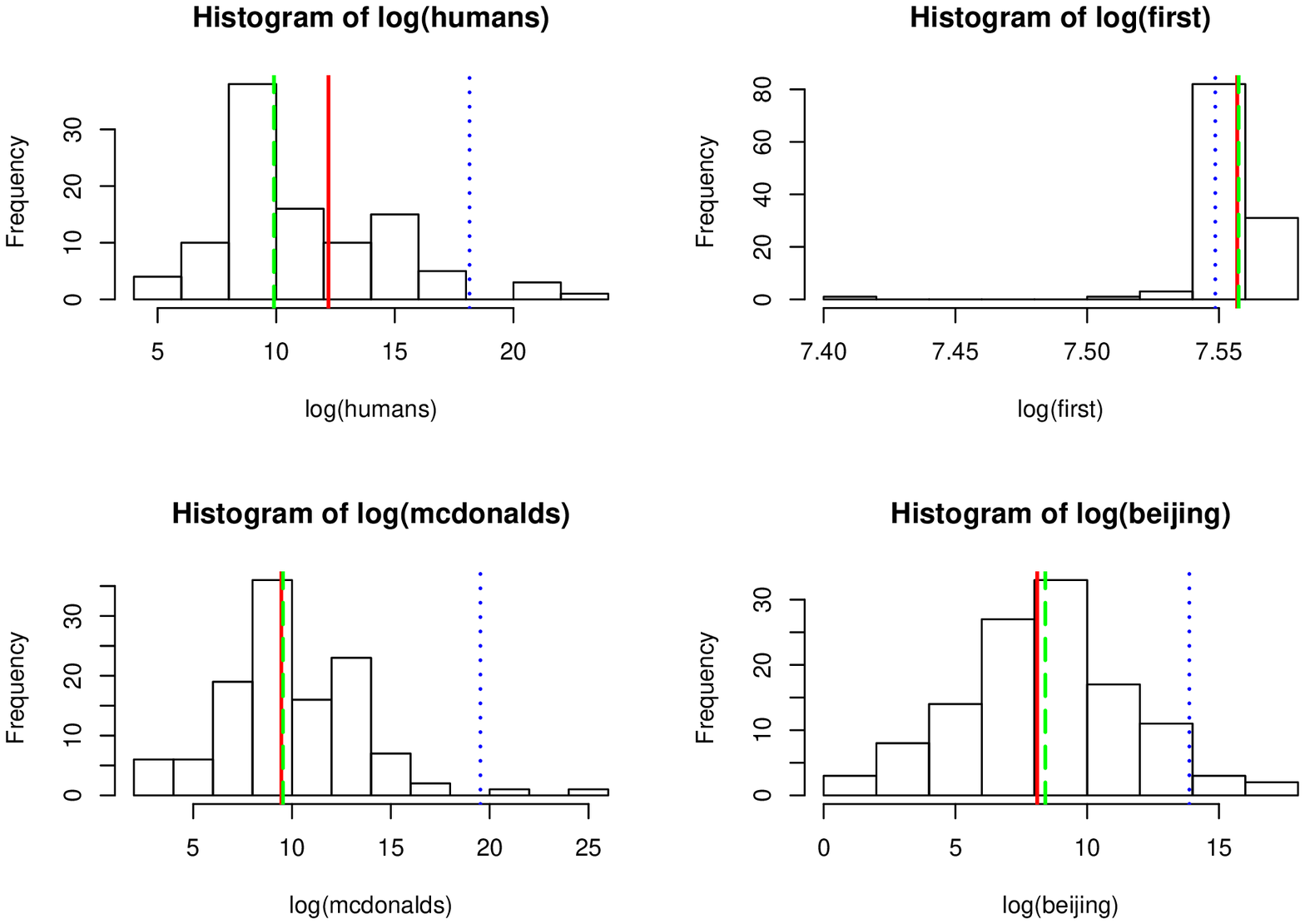}
  \includegraphics[scale=0.4]{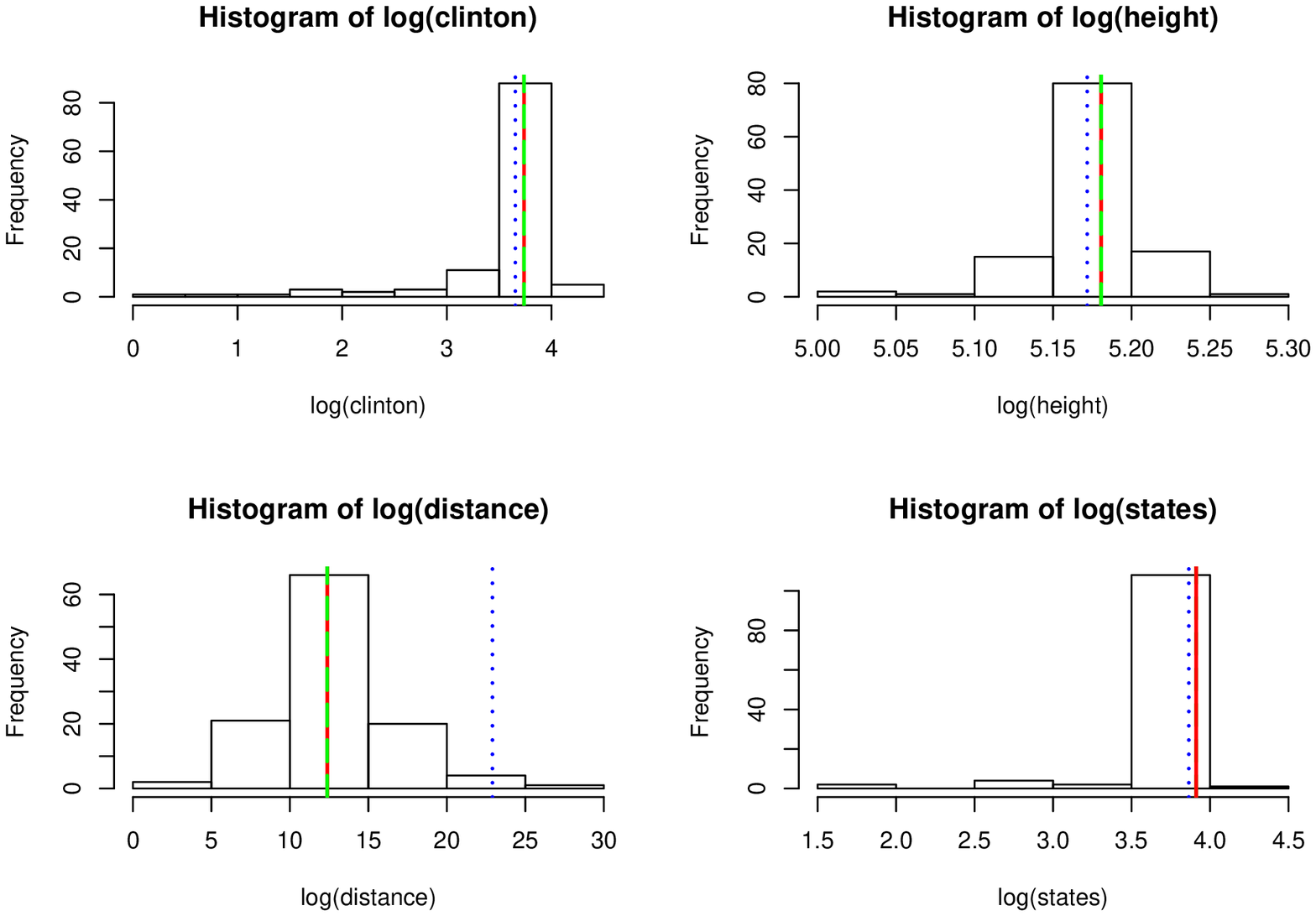}
  \includegraphics[scale=0.4]{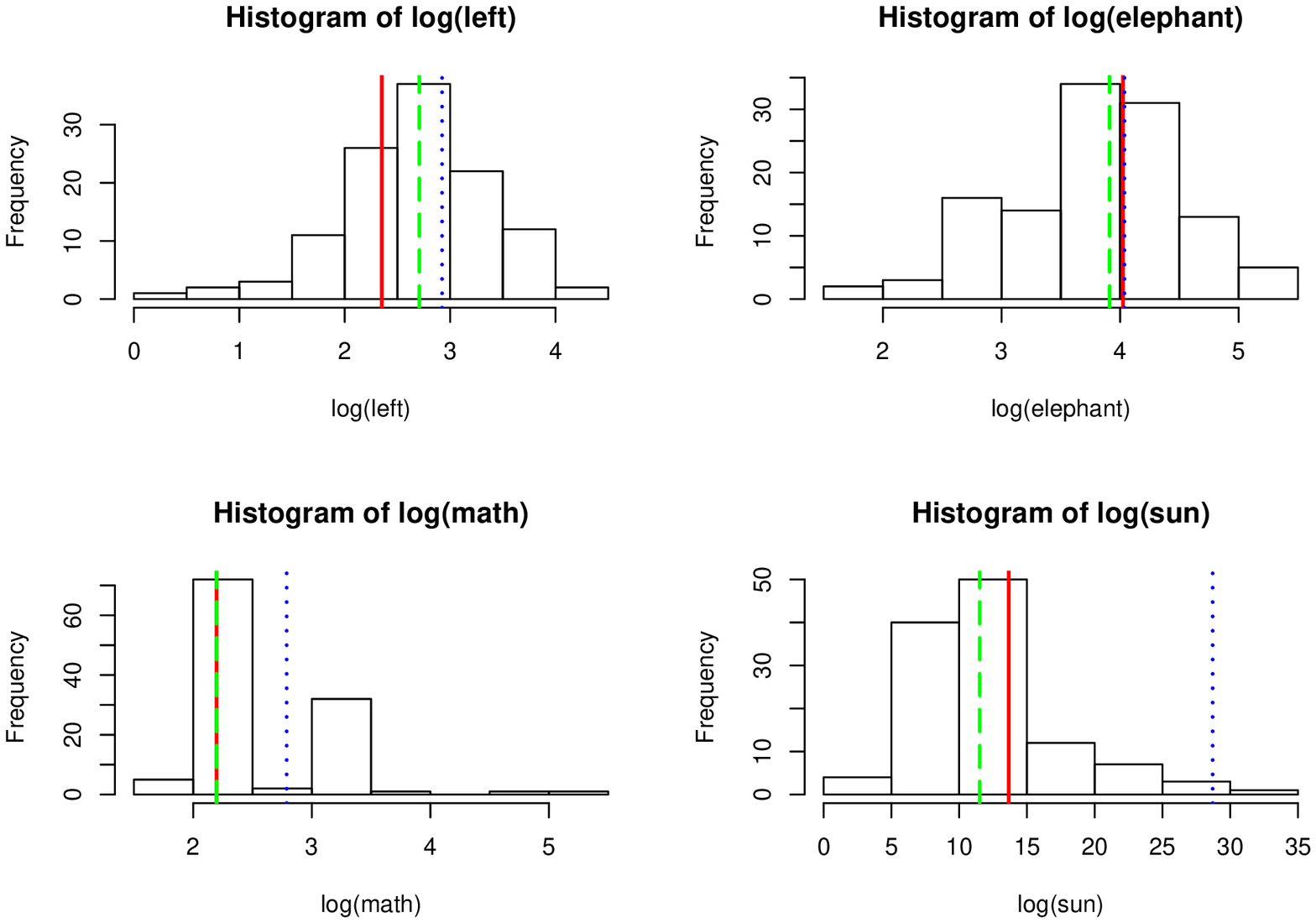}
  \caption{Histograms of $\log$(answers) to questions (1) to (16), top to
    bottom and left to right. Mean (dotted blue), median (dashed green), and
    truth (solid red) indicated.}
  \label{fig:a-1}
\end{figure*}
\begin{table}[!htb]
  \centering
  \begin{tabular}{c||cccc|cccc}
    & \multicolumn{4}{c|}{Median} & \multicolumn{4}{c}{Mean} \\\hline
    & $1\%$ & $2\%$ & $5\%$ & $10\%$ & $1\%$ & $2\%$ & $5\%$ &
    $10\%$\\ \hline\hline
  (1)&   &   &x   &x   &   &   &   &x\\
  (2)&   &   &   &   &   &   &   &\\
  (3)&   &   &   &   &   &   &   &\\
  (4)&   &   &   &   &   &   &   &\\
  (5)&   &   &   &   &   &   &   &\\
  (6)&x   &x   &x   &x   &x   &x   &x   &x\\
  (7)&   &   &   &x   &   &   &   &\\
  (8)&   &   &   &   &   &   &   &\\
  (9)&x   &x   &x   &x   &   &   &   &x\\
  (10)&x   &x   &x   &x   &x   &x   &x   &x\\
  (11)&x   &x   &x   &x   &   &   &   &\\
  (12)&x   &x   &x   &x   &   &   &x   &x\\
  (13)&   &   &   &   &   &   &   &\\
  (14)&   &   &   &   &   &x   &x   &x\\
  (15)&x   &x   &x   &x   &   &   &   &\\
  (16)&   &   &   &   &   &   &   &\\ \hline
      &$6$   & $6$ & $7$ & $8$ & $2$ & $3$ & $4$ & $6$
  \end{tabular}
  \caption{Question numbers and indication whether (x) or not median
    or mean are within the indicated intervals around truth.}
  \label{table:appendixB}
\end{table}

\end{appendices}




\bibliographystyle{plain}
\bibliography{bibliography}

\begin{thebibliography}{10}

\bibitem{AcemogluForthcoming}
Daron Acemoglu, Giacomo Como, Fabio Fagnani, and Asuman Ozdaglar.
\newblock Opinion fluctuations and disagreement in social networks.
\newblock {\em Mathematics of Operations Research}, forthcoming.

\bibitem{Acemoglu2008}
Daron Acemoglu, Munzer~A. Dahleh, Ilan Lobel, and Asuman Ozdaglar.
\newblock Bayesian learning in social networks.
\newblock {\em Review of Economic Studies}, 78:1201--1236, 2011.

\bibitem{Acemoglu2011}
Daron Acemoglu and Asuman Ozdaglar.
\newblock Opinion dynamics and learning in social networks.
\newblock {\em Dynamic Games and Applications}, 1:3--49, 2011.

\bibitem{Acemoglu2010a}
Daron Acemoglu, Asuman Ozdaglar, and Ali ParandehGheibi.
\newblock Spread of (mis)information in social networks.
\newblock {\em Games and Economic Behavior}, 70:194--227, 2010.

\bibitem{Anderlini2012}
Luca Anderlini, Dino Gerardi, and Roger Lagunoff.
\newblock Communication and learning.
\newblock {\em Review of Economic Studies}, 79:419--450, 2012.

\bibitem{Asch1955}
Solomon Asch.
\newblock Opinions and social pressure.
\newblock {\em Scientific American}, 27:31--35, 1955.

\bibitem{Bahrami2012}
Bahador Bahrami, Karsten Olsen, Dan Bang, Andreas Roepstorff, Geraint Rees, and
  Chris Frith.
\newblock What failure in collective decision-making tells us about
  metacognition.
\newblock {\em Philsophical Transactions of the Royal Society B},
  367:1350--1365, 2012.

\bibitem{Banerjee1992}
Abhijit~V. Banerjee.
\newblock A simple model of herd behavior.
\newblock {\em Quarterly Journal of Economics}, 107:797--817, 1992.

\bibitem{Banerjee2004}
Abhijit~V. Banerjee and Drew Fudenberg.
\newblock Word-of-mouth learning.
\newblock {\em Games and Economic Behavior}, 46:1--22, 2004.

\bibitem{Beasley2010}
David Beasley and Dan Kleinberg.
\newblock {\em Networks, Crowds, and Markets: Reasoning about a Highly
  Connected World}.
\newblock Cambridge University Press, Cambridge, UK, 2010.

\bibitem{Bonacich1972}
Philip Bonacich.
\newblock Factoring and weighting approaches to status scores and clique
  identification.
\newblock {\em The Journal of Mathematical Sociology}, 2:113--120, 1972.

\bibitem{Brewer1979}
Marylinn~B. Brewer.
\newblock In-group bias in the minimal intergroup situation: A
  cognitive-motivational analysis.
\newblock {\em Psychological Bulletin}, 86:307--324, 1979.

\bibitem{Budescu2003}
David~V. Budescu, Adrian~K. Rantilla, Hsiu-Ting Yu, and Tzur~M. Karelitz.
\newblock The effects of asymmetry among advisors on the aggregation of their
  opinions.
\newblock {\em Organizational Behavior and Human Decision Processes},
  90:178--194, 2003.

\bibitem{Buchel2012}
Berno Buechel, Tim Hellmann, and Stefan Kl\"{o}{\ss}ner.
\newblock Opinion dynamics under conformity.
\newblock Working Paper, 2012.

\bibitem{Buchel2013}
Berno Buechel, Tim Hellmann, and Stefan Kl\"{o}{\ss}ner.
\newblock Opinion dynamics and wisdom under conformity.
\newblock Working Paper, 2013.

\bibitem{Buchel2012b}
Berno Buechel, Tim Hellmann, and Michael Pichler.
\newblock The dynamics of continuous cultural traits in social networks.
\newblock Working Paper, 2012.

\bibitem{Cao2011}
Zhigang Cao, Mingmin Yang, Xinglong Qu, and Xiaoguang Yang.
\newblock Rebels lead to the doctrine of the mean: Opinion dynamic in a
  heterogeneous degroot model.
\newblock In {\em The 6th International Conference on Knowledge, Information
  and Creativity Support Systems}, pages 29--35, Beijing, China, 2011.

\bibitem{Castano2002}
Emanuele Castano, Vincent Yzerbyt, David Bourguignon, and El\'{e}onore Seron.
\newblock Who may enter? {T}he impact of in-group identification on
  in-group/out-group categorization.
\newblock {\em Journal of Experimental Social Psychology}, 38:315--322, 2002.

\bibitem{Chandra2012}
Arun~G. Chandrasekhar, Horacio Larreguy, and Juan~P. Xandri.
\newblock Testing models of social learning on networks: Evidence from a lab
  experiment in the field.
\newblock {\em Working Paper}, 2012.

\bibitem{Charness2007}
Gary Charness, Luca Rigotti, and Aldo Rustichini.
\newblock Individual behavior and group membership.
\newblock {\em American Economic Review}, 97:1340--1352, 2007.

\bibitem{Corazzini2012}
Luca Corazzini, Filippo Pavesi, Beatrice Petrovich, and Luca Stanca.
\newblock Influential listeners: An experiment on persuasion bias in social
  networks.
\newblock {\em European Economic Review}, 56:1276--1288, 2012.

\bibitem{Concordet1785}
Marquis de~Concordet.
\newblock {\em Essai sur l'application de l'analyse \`{a} la probabilit\'{e}
  des d\'{e}cisions rendues \`{a} la pluralit\'{e} des voix}.
\newblock de l'Impr. Royale, Paris, France, 1785.

\bibitem{Deffuant2000}
Guillaume Deffuant, David Neau, Frederic Amblard, and Gerard Weisbuch.
\newblock Mixing beliefs among interacting agents.
\newblock {\em Advances in Complex Systems}, 3:87--98, 2000.

\bibitem{DeGroot1974}
Morris~H. DeGroot.
\newblock Reaching a consensus.
\newblock {\em Journal of the American Statistical Association},
  69(345):118--121, 1974.

\bibitem{Demarzo2003}
Peter~M. DeMarzo, Dimitri Vayanos, and Jeffrey Zwiebel.
\newblock Persuasion bias, social influence, and unidimensional opinions.
\newblock {\em The Quarterly Journal of Economics}, 118(3):909--968, August
  2003.

\bibitem{Deutsch1955}
Morton Deutsch and Harold~B. Gerard.
\newblock A study of normative and informational social influences upon
  individual judgement.
\newblock {\em Journal of Abnormal Psychology}, 51:629--636, 1955.

\bibitem{Douven2009}
Igor Douven and Alexander Riegler.
\newblock Extending the {H}egselmann-{K}rause model ii.
\newblock {\em Proceedings of ECAP 2009}, 2009.

\bibitem{Douven2009b}
Igor Douven and Alexander Riegler.
\newblock Extending the {H}egselmann-{K}rause model iii: From single beliefs to
  complex belief states.
\newblock {\em Episteme}, 6:145--163, 2009.

\bibitem{Douven2010}
Igor Douven and Alexander Riegler.
\newblock Extending the {H}egselmann-{K}rause model i.
\newblock {\em The Logic Journal of the IGPL}, 18:323--335, 2010.

\bibitem{Eger2013}
Steffen Eger.
\newblock Opinion dynamics under opposition.
\newblock http://arxiv.org/pdf/1306.3134v2.pdf, 2013.

\bibitem{Einhorn1977}
Hillel~J. Einhorn, Robin~M. Hogarth, and Eric Klempner.
\newblock Quality of group judgment.
\newblock {\em Psychological Bulletin}, 84:158--172, 1977.

\bibitem{Fehrler2013}
Sebastian Fehrler and Michael Kosfeld.
\newblock Can you trust the good guys? {T}rust within and between groups with
  different missions.
\newblock Working Paper, 2013.

\bibitem{Festinger1957}
Leon Festinger.
\newblock {\em A theory of cognitive dissonance}.
\newblock Row, Peterson, Evanston, IL, 1957.

\bibitem{French1956}
John~R.P. French.
\newblock A formal theory of social power.
\newblock {\em Psychological Review}, 63:181--194, 1956.

\bibitem{Friedkin1990}
Noah~E. Friedkin and Eugene~C. Johnsen.
\newblock Social influence and opinions.
\newblock {\em Journal of Mathematical Sociology}, 15:193--205, 1990.

\bibitem{Friedkin1999}
Noah~E. Friedkin and Eugene~C. Johnsen.
\newblock Social influence networks and opinion change.
\newblock {\em Advances in group processes}, 16:1--29, 1999.

\bibitem{Gale2003}
Douglas Gale and Shachar Kariv.
\newblock Bayesian learning in social networks.
\newblock {\em Games and Economic Behavior}, 45:329--346, 2003.

\bibitem{Galton1907}
Francis Galton.
\newblock Vox populi.
\newblock {\em Nature}, 75:450--451, 1907.

\bibitem{Golub2010}
Benjamin Golub and Matthew~O. Jackson.
\newblock Na\"{i}ve learning in social networks and the wisdom of crowds.
\newblock {\em American Economic Journal: Microeconomics}, 2:112--149, 2010.

\bibitem{Golub2012}
Benjamin Golub and Matthew~O. Jackson.
\newblock How homophily affects the speed of learning and best-response
  dynamics.
\newblock {\em The Quarterly Journal of Economics}, 127:1287--1338, 2012.

\bibitem{Goyal2004}
Sajeev Goyal.
\newblock Learning in networks: a survey.
\newblock In G.~Demage and M.~Wooders, editors, {\em Group formation in
  economics; Networks, Clubs, and Coalitions}, chapter~4. Cambridge University
  Press, Cambridge U.K., 2004.

\bibitem{Groeber2013}
Patrick Groeber, Jan Lorenz, and Frank Schweitzer.
\newblock Dissonance minimization as a microfoundation of social influence in
  models of opinion formation.
\newblock {\em Journal of Mathematical Sociology}, 2013.

\bibitem{Harary1959}
Frank Harary.
\newblock A criterion for unanimity in {F}rench's theory of social power.
\newblock {\em Studies in social power}, 1959.

\bibitem{Hegselmann2002}
Rainer Hegselmann and Ulrich Krause.
\newblock Opinion dynamics and bounded confidence: models, analysis and
  simulation.
\newblock {\em J. Artificial Societies and Social Simulation}, 5(3), 2002.

\bibitem{Hegselmann2005}
Rainer Hegselmann and Ulrich Krause.
\newblock Opinion dynamics driven by various ways of averaging.
\newblock {\em Computational Economics}, 25(4):381--405, 2005.

\bibitem{Hegselmann2006}
Rainer Hegselmann and Ulrich Krause.
\newblock Truth and cognitive division of labour: First steps towards a
  computer aided social epistemology.
\newblock {\em Journal of Artificial Societies and Social Simulation}, 9(3):10,
  2006.

\bibitem{Jackson2002}
Matthew~O. Jackson and Alison Watts.
\newblock On the formation of interaction networks in social coordination
  games.
\newblock {\em Games and Economic Behavior}, 41:265--291, 2002.

\bibitem{Janis1972}
Irving~L. Janis.
\newblock {\em Victims of Groupthink}.
\newblock Houghton Mifflin, Boston, MA, 1972.

\bibitem{Jones1984}
Stephen~R.G. Jones.
\newblock {\em The economics of conformism}.
\newblock B. Blackwell, Oxford; New York, NY, 1984.

\bibitem{Kahnemann1984}
Daniel Kahnemann and Amos Tversky.
\newblock Choices, values, and frames.
\newblock {\em American Psychologist}, 39:341--350, 1984.

\bibitem{Kerr1996}
Nobert~L. Kerr, Robert~J. MacCoun, and Geoffrey~P. Kramer.
\newblock Bias in judgment: Comparing individuals and crowds.
\newblock {\em Psychological Review}, 103:687--719, 1996.

\bibitem{Kerr2004}
Norbert~L. Kerr and R.~Scott Tindale.
\newblock Group performance and decision making.
\newblock {\em Annual Review of Psychology}, 55:623--655, 2004.

\bibitem{Kitts2006}
James~A. Kitts.
\newblock Social influence and the emergence of norms amid ties of amity and
  enmity.
\newblock {\em Simulation Modelling Practice and Theory}, 14:407--422, 2006.

\bibitem{Krause2000}
Ulrich Krause.
\newblock A discrete nonlinear and non-autonomous model of consensus formation.
\newblock In S.~Elaydi, G.~Ladas, J.~Popenda, and R.~Rakowski, editors, {\em
  Communications in Difference Equations}, pages 227--236. Gordon and Breach
  Publ., Amsterdam, 2000.

\bibitem{Kunda1990}
Ziva Kunda.
\newblock The case for motivated reasoning.
\newblock {\em Psychological Bulletin}, 108:480--498, 1990.

\bibitem{Lehrer1983}
Keith Lehrer.
\newblock Rationality as weighted averaging.
\newblock {\em Synthese}, 57:283--295, 1983.

\bibitem{Lehrer1981}
Keith Lehrer and Carl Wagner.
\newblock Rational consensus in science and society.
\newblock {\em Reidel, Dordrecht}, 1981.

\bibitem{Lobel2000}
Joel Lobel.
\newblock Economists' models of learning.
\newblock {\em Journal of Economic Theory}, 94:241--261, 2000.

\bibitem{Lorenz2005}
Jan Lorenz.
\newblock A stabilization theorem for dynamics of continuous opinions.
\newblock {\em Physica A: Statistical Mechanics and its Applications},
  355(1):217--223, 2005.

\bibitem{Lorenz2006}
Jan Lorenz.
\newblock Continuous opinion dynamics of multidimensional allocation problems
  under bounded confidence. {M}ore dimensions lead to better chances for
  consensus.
\newblock {\em European Journal of Economic and Social Systems}, 19:213--227,
  2006.

\bibitem{Lorenz2007}
Jan Lorenz.
\newblock Continuous opinion dynamics under bounded confidence: A survey.
\newblock {\em International Journal of Modern Physics C}, 2007.

\bibitem{Lorenz2011}
Jan Lorenz, Heiko Rauhut, Frank Schweitzer, and Dirk Helbing.
\newblock How social influence can undermine the wisdom of crowd effect.
\newblock {\em Proceedings of the National Academy of Sciences},
  108(22):9020--9025, 2011.

\bibitem{Mackay1841}
Charles Mackay.
\newblock {\em The extraordinary and popular delusions and the madness of
  crowds}.
\newblock Wordsworth Editions Limited, Ware, UK, 4th edition, 1841.

\bibitem{Mannes2009}
Albert~E. Mannes.
\newblock Are we wise about the wisdom of crowds? {T}he use of group judgments
  in belief revision.
\newblock {\em Management Science}, 55:1267--1279, 2009.

\bibitem{McPherson2001}
Miller McPherson, Lynn Smith-Lovin, and James~M. Cook.
\newblock Birds of a feather: {H}omophily in social networks.
\newblock {\em Annual Review of Sociology}, 27:415--444, 2001.

\bibitem{Meyer2000}
Carl~D. Meyer.
\newblock {\em Matrix Analysis and Applied Linear Algebra}.
\newblock SIAM, Philadelphia, 2000.

\bibitem{Pan2010}
Zhengzheng Pan.
\newblock Trust, influence, and convergence of behavior in social networks.
\newblock {\em Mathematical Social Sciences}, 60:69--78, 2010.

\bibitem{Rosenberg2006}
D.~Rosenberg, E.~Solan, and N.~Vieille.
\newblock Informational externalities and convergence of behavior.
\newblock {\em Preprint}, 2006.

\bibitem{Surowieki2004}
James Surowiecki.
\newblock {\em The wisdom of crowds. {W}hy the many are smarter than the few
  and how collective wisdom shapes business, economies, societies and nations}.
\newblock Little, Brown, London, 2004.

\bibitem{Tversky1974}
Amos Tversky and Daniel Kahnemann.
\newblock Judgment under uncertainty: Heuristics and biases.
\newblock {\em Science}, 185:1124--1131, 1974.

\bibitem{Wagner1978}
Carl Wagner.
\newblock Consensus through respect: A model of rational group decision-making.
\newblock {\em Philosophical studies: An international Journal for Philosophy
  in the Analytic Tradition}, 34:335--349, 1978.

\bibitem{Wittenbaum1996}
Gwen~M. Wittenbaum and Garold Stasser.
\newblock Management of information in small groups.
\newblock In J.L. Nye and A.M. Brower, editors, {\em What's Social about Social
  Cognition}, pages 967--978. Sage, Thousand Oaks, CA, 1996.

\bibitem{Yaniv2004}
Ilan Yaniv.
\newblock The benefit of additional opinions.
\newblock {\em American Psychological Society}, 13:75--78, 2004.

\bibitem{Yaniv2000}
Ilan Yaniv and Eli Kleinberger.
\newblock Advice taking in decision making: Egocentric discounting and
  reputation formation.
\newblock {\em Organizational Behavior and Human Decsion Processes},
  83:260--281, 2000.

\bibitem{AcemogluSubmitted}
Ercan Yildiz, Daron Acemoglu, Asuman Ozdaglar, Amin Saberi, and Anna Scaglione.
\newblock Opinion dynamics with stubborn agents.
\newblock {\em Operations Research}, submitted.

\bibitem{Zafar2011}
Basit Zafar.
\newblock An experimental investigation of why individuals conform.
\newblock {\em European Economic Review}, 55:774--798, 2011.

\bibitem{Zhang2013}
Bo-Yu Zhang, Zhi-Gang Cao, Cheng-Zhong Qin, and Xiao-Guang Yang.
\newblock Fashion and homophily, 2013.
\newblock Available at SSRN: http://ssrn.com/abstract=2250898 or
  http://dx.doi.org/10.2139/ssrn.2250898.

\end{thebibliography}
\end{document}